\documentclass[12pt]{article}

\newcommand{\blind}{1}

\usepackage{geometry,placeins, amsthm}
{
	\newtheorem{theorem}{Theorem}
	\newtheorem{lemma}{Lemma}
	\newtheorem{cor}{Corollary}
	 \newtheorem{assumption}{Assumption}
  }
\usepackage{adjustbox}

\usepackage{setspace,natbib,footnote,tablefootnote,threeparttable}
\usepackage{amssymb,fancyhdr}
\usepackage{graphicx,amscd,amsmath,amssymb,amsfonts,verbatim}
\usepackage{mathrsfs,graphics}
\usepackage{bbm}
\usepackage[all]{xy}
\usepackage{authblk}
\usepackage{xcolor}
\usepackage{bm}
\usepackage{bbm}

\usepackage{textcomp}

\usepackage{graphics,graphicx}
\usepackage{bbm}
\usepackage{amsmath}
\usepackage{amssymb}
\usepackage[all]{xy}

\usepackage{verbatim}
\usepackage{mathrsfs}
\usepackage{bbm, url}

\def\E{\mathbbmss{E}}

\def\R{\mathbbmss{R}}

\def\bV{\mathbf{V}}

\def\bQ{\mathbf{Q}}

\def\Ytil{\tilde{Y}_i}
\def\Yhat{\hat{\tilde{Y}}_i}

\def\bX{\mathbf{X}}

\DeclareMathOperator*{\argmin}{argmin}
\def\bbeta{\bm{\beta}}
\def\balpha{\bm{\alpha}}
\def\btheta{\bm{\theta}}

\usepackage{bbm}

\begin{document}

\def\spacingset#1{\renewcommand{\baselinestretch}%
{#1}\small\normalsize} \spacingset{1}


\if1\blind
{
  \title{\bf Robust Q-learning}
  \author{Ashkan Ertefaie$^1$, James R. McKay$^2$, David Oslin$^3$,
    Robert L. Strawderman$^1$  \hspace{.2cm}\\
 
{\small{
    $^1$Department of Biostatistics and Computational Biology, University of Rochester\\
    $^2$Center on the Continuum of Care in the Addictions, Department of Psychiatry, University of Pennsylvania\\
    $^3$Philadelphia Veterans Administration Medical Center, and Treatment Research Center and Center for Studies of Addictions, Department of Psychiatry, University of Pennsylvania }}
}
  \maketitle
} \fi

\if0\blind
{
  \bigskip
  \bigskip
  \bigskip
  \begin{center}
    {\LARGE\bf Robust Q-learning}
\end{center}
  \medskip
} \fi

\bigskip
\begin{abstract}
Q-learning is a regression-based approach that is widely used to formalize the development of an optimal dynamic treatment strategy. Finite dimensional working models are typically used to estimate certain nuisance parameters, and misspecification of these working models can result in residual confounding and/or efficiency loss. We propose a robust Q-learning approach which allows estimating such nuisance parameters using data-adaptive techniques. We study the asymptotic behavior of our estimators and provide simulation studies that highlight the need for and usefulness of
the proposed method in practice. We use the data from the ``Extending Treatment Effectiveness of Naltrexone'' multi-stage randomized 
trial to illustrate our proposed methods.  
\end{abstract}

\noindent%
\vfill

\newpage
\spacingset{1.5}

\section{Introduction}
A dynamic treatment strategy is a sequence of decision rules that maps individual characteristics to a treatment option at each decision point (i.e., a specific point in time in which a treatment is to be considered or altered). An optimal dynamic treatment strategy seeks to make these decisions to maximize a particular expected health outcome \citep{lavori2000design, murphy2005experimental, nahum2012experimental, lei2012smart, davidian2016dynamic}. 
This is similar to clinical decision making whereby care providers tailor the type/dose of treatment over the course of clinical care based on ongoing information regarding patient progress in treatment. 

The main goal of precision medicine (i.e., developing an effective dynamic treatment strategy) is to use patient characteristics to inform a personalized treatment plan as a sequence of decision rules that leads to the best possible health outcome for each patient  \citep{nahum2012experimental, chakraborty2013statistical, kosorok2015adaptive, butler2018incorporating}. Q-learning is a reinforcement learning algorithm that is widely used to estimate an optimal dynamic treatment strategy using data from multi-stage randomized clinical  trials or observational studies \citep{watkins1992q, nahum2012q, laber2014dynamic}.  Starting with the final study stage, Q-learning finds the treatment option that optimizes the desired expected outcome. Fixing the optimally-chosen treatment at the final stage,  Q-learning moves backward to the immediately preceding stage and searches for a treatment option assuming that future treatments will be optimized. The process continues until the first stage is reached. This { backward induction} procedure is designed to avoid treatment options that appear to be optimal in the short term but may lead to a less desirable long-term outcome \citep{chakraborty2013statistical, davidian2016dynamic}.  Similar to other model-based approaches, model misspecification can seriously affect the result of Q-learning and the problem exacerbates substantially as the number of stages increases. Specifically, it can lead to residual confounding and suboptimal  dynamic treatment strategies \citep{zhao2009reinforcement, ertefaie2016q}. An alternative approach is A-learning; this backward induction strategy is also model-based, 
hence subject to the possibility of misspecification, but imposes somewhat less restrictive regression model assumptions through modeling
only the contrasts between treatments and the propensity of treatment assignment given the observed patient history \citep{murphy2003optimal, schulte2014q, shi2018high}. However, this extra robustness comes at the price substantially reduced efficiency; indeed, Q-learning may lead to a parameter estimate that  can be {up to 170\% more 
efficient than A-learning} \citep{schulte2014q}. Policy learning methods are another class of methods for estimating an optimal dynamic treatment regime that circumvent the need for the conditional outcome models by directly optimizing the expected outcome among a class of rules \citep{zhao2009reinforcement, zhao2011reinforcement, zhao2012estimating, song2015sparse}. Despite this appealing feature, policy learning methods are inefficient and  fail to provide reasonable inference for the parameter estimates that define non-smooth decision rules  (e.g., indicator or max operators) due
to slow rates of convergence \citep{kosorok2015adaptive, zhao2015new}.

 Doubly robust estimators that are based on modeling both the treatment and outcome processes have also been proposed for policy learning and for structural nested models \citep{zhang2012robust, 10.2307/24538158, zhang2013robust,  wallace2015doubly}. In general terms, consistency of the doubly robust estimators is guaranteed as long as either the treatment assignment mechanism or the postulated conditional mean outcome models are correctly specified, and semiparametric efficiency follows when both models are correctly specified \citep{rotnitzky1998semiparametric, van2003unified, tsiatis2007semiparametric}.  However, 
while doubly robust estimators give two routes for consistent estimation, the performance of these estimators depends critically on the modeling choice for the indicated treatment and mean outcome parameters. In practice, finite-dimensional models are used that are too restrictive and likely to be misspecified \citep{wallace2015doubly}. \cite{kang2007demystifying} showed that doubly robust estimators can have poor performance when both models are misspecified. To mitigate this problem, bias reduction techniques have been proposed \citep{cao2009improving, vermeulen2015bias, vermeulen2016data}.  An alternative is to use flexible learning-based methods that may reduce the chance of inconsistency \citep{benkeser2017doubly}.

In this paper, we consider the problem of Q-learning for the setting of a two-stage dynamic binary treatment choice regime. 
A typical approach involves postulating linear models for both the first and second stage Q-functions that have
stage-specific main effects involving pretreatment variables and interactions of these variables with the stage-specific treatment choice 
(e.g., \citealp{laber2014dynamic}); however, only the interaction terms in these models directly influence the corresponding treatment decision functions 
that are optimized as part of the Q-learning process. Due to the nature of backward induction, the first stage model is likely to be a complicated 
function of the relevant covariates; misspecification of the main effect models in either stage can induce non-ignorable residual confounding.
To increase robustness, we therefore consider the indicated main effects as unknown nuisance parameters and we adapt an approach originally
proposed in \cite{robinson1988root} for partial linear models that allows us to eliminate these parameters from the Q-functions. In particular, these hard-to-estimate
parameters are replaced with models for the treatment assignment probability and mean outcome given pretreatment covariates that can be more easily estimated using, for example, nonparametric regression methods or related statistical learning methods (e.g., random forests). The resulting transformation process leads to consistent and asymptotically normal estimators of the so-called first- and second-stage blip functions that are robust to misspecification of the main effect nuisance parameters.

\section{Notation and Formulation}

Consider a two-stage study where binary treatment decisions are made at each time point. 
Let $\mathcal{O} = (\bm{X}_1,A_1,\bm{X}_2,A_2,Y)$ follow some probability distribution $P_0$ and 
suppose we observe $N$ independent, identically distributed trajectories of $\mathcal{O}.$
The vector $\bm{X}_1 \in \mathcal{X}_1 \subset \R^{p_1}$ consists of all available baseline covariates measured 
before treatment at the first decision point $A_1\in \{0,1\}$ and the vector $\bm{X}_2 \in \mathcal{X}_2 \subset \R^{p_2}$ 
consists of all available intermediate covariates measured before treatment at the second decision point $A_2\in \{0,1\}$. 
For notational convenience we define  
$\bm{S}^0_i=(\bm{X}_{1i}^\top,A_{1i},\bm{X}_{2i}^\top)^\top\in \mathcal{S} \subset \R^{p_1+p_2+1}$ and 
 $\bm{W}^0_i=\bm{X}_{1i} \in \mathcal{X}_1 \subset \R^{p_1}$.
%
%
For later use, we also define variables $\bm{S}_i$ and $\bm{W}_i, i=1,\ldots,n;$ respectively,
each represents some finite dimensional function of the variables in $\bm{S}^0_i$ and  $\bm{W}^0_i.$
Hence, knowledge of $\bm{S}^0_i$ and $\bm{W}^0_i$ respectively implies knowledge of $\bm{S}_i$ and $\bm{W}_i;$
however, the reverse may not hold.
The observed outcome $Y \in \R$ (measured after $A_2$) is assumed continuous, with a larger value of $Y$ indicating  a better clinical outcome.

\section{A Robust Formulation of Q-learning}
\label{sec:Qform}
The outcome $Y$ is assumed to satisfy the model
\begin{equation}
  \label{eq:causal-model}
  Y_i = \eta_2(\bm{S}^0_{i}) + A_{2i} \cdot \Delta_2(\bm{S}^0_{i}) + \epsilon_{2i}, 
  ~\E( \epsilon_{2i} | \bm{S}^0_{i}, A_{2i} ) = 0, ~i=1,\dotsc,N,
\end{equation}
where the deterministic, unknown real-valued functions $\eta_2(\cdot)$ and $\Delta_2(\cdot)$ are defined on $\mathcal{S}.$
Because the treatment variable $A_2$ is binary,  the additive error model  \eqref{eq:causal-model} 
places no parametric constraints on the conditional mean function.

In Q-learning,  backward induction is used to characterize the optimal dynamic treatment regime. 
Define the second stage Q-function 
\begin{equation}
\label{eq:stdQ2}
Q_{2}(\bm{s}^0, a_{2}) = \E( Y  | \bm{S}^0 = \bm{s}^0, A_{2} = a_2 ) =  \eta_2(\bm{s}^0) + a_2 \Delta_2(\bm{s}^0);
\end{equation}
this measures ``quality'' when treatment $a_2$ is assigned 
to a patient with characteristics $\bm{s}^0$ at the second stage \citep{laber2014dynamic}.
Similarly, define the first stage Q-function as
\begin{equation}
\label{eq:stdQ1}
Q_{1}(\bm{w}^0, a_{1}) = \E\bigl\{ \max_{a_2} Q_2(\bm{S}^0, a_{2})  | \bm{W}^0 = \bm{w}^0, A_{1} = a_1 \bigr\}
\end{equation}
where, for $a_2 \in \{0,1\}$, 
\[
\max_{a_2} Q_2(\bm{s}^0, a_{2}) = \eta_2(\bm{s}^0) + \Delta_2(\bm{s}^0) I\{ \Delta_2(\bm{s}^0) > 0 \}.
\]
Analogously to $Q_{2}(\bm{s}^0, a_{2}),$ $Q_{1}(\bm{w}^0, a_{1})$ measures ``quality'' when treatment $a_1$ is assigned to a patient with 
characteristics $\bm{w}^0$ at baseline, assuming the optimal treatment choice is also made in the second stage.  
Because $A_1$ is binary, $Q_{1}(\bm{w}^0, a_{1})$ can without loss of generality be written as $\eta_1(\bm{w}^0) +  a_1 \Delta_1(\bm{w}^0),$
where the real-valued functions $\eta_1(\cdot)$ and $\Delta_1(\cdot)$ are defined on $\mathcal{X}_1.$ 
It then follows that
\[
\max_{a_1} Q_1(\bm{w}^0, a_{1}) = \eta_1(\bm{w}^0) + \Delta_1(\bm{w}^0) I\{ \Delta_1(\bm{w}^0) > 0 \}.
\]
Taken together, 
the optimal dynamic treatment regime is given by $d^{opt}(\bm{s}^0)=\{d^{opt}_1(\bm{w}^0),d^{opt}_2(\bm{s}^0)\}$, where 
$d^{opt}_1(\bm{w}^0)=I\{\Delta_1(\bm{w}^0)>0\}$ and $d^{opt}_2(\bm{s}^0)=I\{\Delta_2(\bm{s}^0)>0\}$. We note that, under standard causal assumptions
and formulated appropriately, $\Delta_1(\bm{w}^0)$ and $\Delta_2(\bm{s}^0)$ are
commonly referred to as the first and second stage blip functions.

A widely used convention in the literature on dynamic  treatment regimes is to respectively model 
$Q_2(\bm{s}^0, a_{2}; \bm{b}_2) =  \bm{b}^\top_{20} \bm{s}_{20} + a_2 \bm{b}^\top_{21} \bm{s}_{21}$
and
$Q_1(\bm{w}^0, a_{1}; \bm{b}_1) =  \bm{b}^\top_{10} \bm{w}_{10} + a_1 \bm{b}^\top_{11} \bm{w}_{11},$
where $\bm{s}_{2k}, k=0,1$ are sets of variables derived from $\bm{s}^0$ (i.e., a realization of $\bm{S}^0$) and
$\bm{w}_{1k}, k=0,1$ are sets of variables derived from $\bm{w}^0$ (i.e., a realization of $\bm{W}^0$).
These model formulations impose restrictive assumptions on both $\eta_j(\cdot)$ and $\Delta_j(\cdot), j = 1,2$.
However, the decision functions of interest only depend on the $\Delta_j(\cdot), j = 1,2,$ or
per the indicated linear models, on the $\bm{b}_{j1}s.$ Because misspecification of the models for
the nuisance parameters $\eta_j(\cdot), j =1,2$ can induce residual confounding and affect the 
causal interpretation of the interaction terms, we propose a novel modification of the Q-learning
approach that eliminates the need to directly model $\eta_j(\cdot), j=1,2.$ In particular, 
we adapt techniques originally introduced by  \cite{robinson1988root} for root-n consistent
inference in semiparametric regression models, specifically partially linear models, to the
 Q-learning problem.
 
 \subsection{Regression model for Stage 2 decision function}
 \label{stage2}
Define $\mu_{2Y}(\bm{s}^0) = \E( Y | \bm{S}^0 = \bm{s}^0 )$ and $\mu_{2A}(\bm{s}^0) = \E( A_2 | \bm{S}^0 = \bm{s}^0 );$
for now, we will proceed as if these two functions are known. Under 
\eqref{eq:causal-model}, or equivalently \eqref{eq:stdQ2},
$\mu_{2Y}(\bm{s}^0) = \eta_2(\bm{s}^0) + \mu_{2A}(\bm{s}^0)  \Delta_2(\bm{s}^0),$ implying that
\begin{equation}
  \label{eq:causal-model-2}
  Y_i -  \mu_{2Y}( \bm{S}^0_i) = \{ A_{2i} - \mu_{2A}(\bm{S}^0_i) \}  \cdot \Delta_2(\bm{S}^0_{i}) + \epsilon_{2i}, 
  ~\E( \epsilon_{2i} | \bm{S}^0_{i}, A_{2i} ) = 0, ~i=1,\dotsc,N.
\end{equation}
Define the centered second stage Q-function 
\begin{equation}
\label{eq:cenQ2}
Q_{2c}(\bm{s}^0, a_{2}) = \E\{ Y -  \mu_{2Y}(\bm{S}^0) | \bm{S}^0 = \bm{s}^0, A_{2} = a_2 \} = 
\{ a_2 - \mu_{2A}(\bm{s}^0) \} \Delta_2(\bm{s}^0);
\end{equation}
observe that $Q_{2c}(\bm{s}^0, a_{2}) = Q_{2}(\bm{s}^0, a_{2}) - \mu_{2Y}(\bm{s}^0).$ 
Although $Q_{2c}(\bm{s}^0, a_{2})$ differs from $Q_{2}(\bm{s}^0, a_{2}),$ the action
$a_2$ that maximizes $Q_{2c}(\bm{s}^0, a_{2})$ is the same as that which maximizes
$Q_{2}(\bm{s}^0, a_{2})$ because $\mu_{2Y}(\bm{s}^0)$ does not depend on the value of $a_2$.
With $\eta_2(\bm{s}^0)$ eliminated from $Q_{2c}(\bm{s}^0, a_{2}; \bm{\beta}_2),$ we 
propose to model \eqref{eq:cenQ2} via 
\begin{equation}
\label{eq:cenQ2mod}
Q_{2c}( \bm{s}^0, a_{2}; \bm{\beta}_2 ) = \{ a_2 - \mu_{2A}(\bm{s}^0) \} \bm{s}^\top \bm{\beta}_2,
\end{equation}
where $\bm{s}$ denotes the realization of $\bm{S}$
(i.e., some finite set of variables derived from $\bm{S}^0$).
Because 
\[
\E\left\{ Y -  \mu_{2Y}(\bm{S}^0) - Q_{2c}(\bm{S}^0, A_{2}; \bm{\beta}_2) | \bm{S}^0 = \bm{s}^0, A_2 = a_2 \right\} = 
\bigl\{ a_{2} - \mu_{2A}(\bm{s}^0) \bigr\}   \bigl\{ \Delta_2(\bm{s}^0) - \bm{s}^\top \bm{\beta}_2 \bigr\},
\]
it is easily shown that
\begin{eqnarray*}
\bm{\beta}^*_2 
& = & 
\underset{\bm{\beta}_2}{\arg \! \min}~
\E\left[ \left\{ Y -  \mu_{2Y}(\bm{S}^0) - Q_{2c}(\bm{S}^0, A_{2}; \bm{\beta}_2) \right\}^2 \right] \\
& = & \underset{\bm{\beta}_2}{\arg \! \min}~
\E\left[ \bigl\{ A_{2} - \mu_{2A}(\bm{S}^0) \bigr\}^2  \bigl\{ \Delta_2(\bm{S}^0) - \bm{S}^\top \bm{\beta}_2  \bigr\}^2 \right].
\end{eqnarray*}
The second expression shows that $\bm{S}^\top \bm{\beta}^*_2$ is the best (weighted) linear predictor of  $\Delta_2(\bm{S}^0).$
For data $(Y_i,\bm{S}^0_i),i = 1,\ldots,N,$ the above developments further show that one can estimate $\bm{\beta}^*_2$ by 
$\tilde{\bm{\beta}}_{2N},$ calculated as the minimizer of
\[
\sum_{i=1}^N \bigl\{ Y_i -  \mu_{2Y}(\bm{S}^0_i) - Q_{2c}(\bm{S}^0_i, A_{2i}; \bm{\beta}_2) \bigr\}^2.
\]
Finally, maximizing $Q_{2c}( \bm{s}^0, a_{2}; \bm{\beta}^*_2 )$
for $a_2 \in \{0,1\}$ gives $d^{opt}_{2}(\bm{s}; \bm{\beta}^*_2) = I\{\bm{s}^\top \bm{\beta}^*_2 >0\}$
as the optimal model-based treatment decision in Stage 2, with 
a corresponding estimated decision rule of $\tilde{d}^{opt}_2(\bm{s}) = I\{\bm{s}^\top \tilde{\bm{\beta}}_{2N} >0\}$
(i.e., assuming $\mu_{2Y}(\cdot)$ and $\mu_{2A}(\cdot)$ are known).

The calculations above evidently rely on the availability of $\mu_{2Y}(\cdot)$ and  $\mu_{2A}(\cdot).$ The assumption that
$\mu_{2Y}(\cdot)$ is known is particularly unrealistic; hence, in Section \ref{EstOpt}, we establish the properties of the 
corresponding least squares estimator when these functions are estimated from the available
data using suitable consistent nonparametric estimators, such as those derived from random forests \citep[e.g.,][]{scornet2015consistency}
or Super Learner \citep{van2007super}.

\subsection{Regression model for Stage 1 decision function}
\label{stage1}

By construction, the relevant first stage Q-function depends on the model for  \eqref{eq:stdQ2}.
Using the partially linear model of Section \ref{stage2}, the model-based analog of the second stage Q-function \eqref{eq:stdQ2}
is given by
\begin{equation}
\label{eq:Q2mod}
Q_{2}(\bm{s}^0, a_{2}; \bm{\beta}^*_2) = \mu_{2Y}(\bm{s}^0) + Q_{2c}(\bm{s}^0, a_{2}; \bm{\beta}^*_2), 
\end{equation}
where the second term is defined in \eqref{eq:cenQ2mod}. In view of \eqref{eq:stdQ1}, we therefore re-define the first stage
Q-function of interest as
\begin{equation}
\label{eq:newQ1}
Q_{1}(\bm{w}^0, a_{1}) = \E\bigl\{ \max_{a_2} Q_2(\bm{S}^0, a_{2}; \bm{\beta}^*_2)  | \bm{W}^0 = \bm{w}^0, A_{1} = a_1 \bigr\}.
\end{equation}
Without loss of generality, and arguing similarly to the previous section, 
the fact that $A_1$ is binary means 
\begin{equation}
\label{eq:newQ1mod}
Q_{1}(\bm{w}^0, a_{1}) = \eta_1(\bm{w}^0) + a_1 \Delta_1(\bm{w}^0)
\end{equation}
is a saturated nonparametric model for \eqref{eq:newQ1}; here, $a_1 \in \{0,1\}$ and
the real-valued functions $\eta_1(\cdot)$ and $\Delta_1(\cdot)$ are defined on $\mathcal{X}_1$
and are not necessarily assumed to be the same functions that were initially used to define $Q_{1}(\bm{w}^0, a_{1})$ a
the beginning of Section \ref{sec:Qform}.

Define
\begin{equation}
\label{eq:Ytil-2}
\tilde{Y}^\dag = \max_{a_2} Q_{2}(\bm{S}^0, a_{2}; \bbeta^*_2); 
\end{equation}
then, it is not difficult to show that
\begin{eqnarray*}
\label{eq:Q2max}
\tilde{Y}^\dag & = & \mu_{2Y}(\bm{S}^0) + \bm{S}^\top \bbeta^*_2 
\bigl\{ I(\bm{S}^\top \bbeta^*_2  > 0 ) - \mu_{2A}(\bm{S}^0) \bigr\}.
\end{eqnarray*}
Under \eqref{eq:newQ1mod}, 
$\E\bigl( \tilde{Y}^\dag \big |  \bm{W}=\bm{w}^0, A_1 = a_1 \bigr) = \eta_1(\bm{w}^0) + a_1 \Delta_1(\bm{w}^0);$
similarly,
\[
\mu_{1Y}(\bm{w}^0) =  \E\bigl(\tilde{Y}^\dag \big |  \bm{W}=\bm{w}^0 \bigr)
= \eta_1(\bm{w}^0) + \mu_{1A}(\bm{w}^0)  \Delta_1(\bm{w}^0),
\]
where $\mu_{1A}(\bm{w}^0) = \E( A_1 | \bm{W}=\bm{w}^0).$  Similarly to \eqref{eq:cenQ2}, 
 we can write
\begin{equation}
\label{eq:cenQ1}
Q_{1c}(\bm{w}^0, a_{1}) = \E\{ \tilde{Y}^\dag -  \mu_{1Y}(\bm{W}^0) | \bm{W}^0 = \bm{w}^0, A_{1} = a_1 \} = 
\{ a_1 - \mu_{1A}(\bm{w}^0) \} \Delta_1(\bm{w}^0)
\end{equation}
and, considering  \eqref{eq:cenQ2mod}, can model \eqref{eq:cenQ1} via
\begin{equation}
\label{eq:cenQ1mod}
Q_{1c}(\bm{w}^0, a_{1}; \bm{\beta}_1 ) = \{ a_1 - \mu_{1A}(\bm{w}^0) \} \bm{w}^\top \bm{\beta}_1,
\end{equation}
where $\bm{w}$ is defined analogously to $\bm{s}$.
Together, these results imply that
\[
\E\left\{ \tilde{Y}^\dag -  \mu_{1Y}(\bm{W}^0) - Q_{1c}(\bm{W}^0, A_{1}; \bm{\beta}_1) | \bm{W}^0 = \bm{w}^0, A_1 = a_1 \right\} = 
\bigl\{ a_{1} - \mu_{1A}(\bm{w}^0) \bigr\}   \bigl\{ \Delta_1(\bm{w}^0) - \bm{w}^\top \bm{\beta}_1 \bigr\}.
\]
Similarly to the second stage problem, it now follows that
\begin{eqnarray*}
\bm{\beta}^*_1 
& = & 
\underset{\bm{\beta}_1}{\arg \! \min}~
\E\left[ \left\{ \tilde{Y}^\dag -  \mu_{1Y}(\bm{W}^0) - Q_{1c}(\bm{W}^0, A_{1}; \bm{\beta}_1) \right\}^2 \right] \\
& = & \underset{\bm{\beta}_1}{\arg \! \min}~
\E\left[ \bigl\{ A_{1} - \mu_{1A}(\bm{W}^0) \bigr\}^2  \bigl\{ \Delta_1(\bm{W}^0) - \bm{W}^\top \bm{\beta}_1  \bigr\}^2 \right],
\end{eqnarray*}
the latter implying that $\bm{W}^\top \bm{\beta}^*_1$ is the best (weighted) linear predictor of  $\Delta_1(\bm{W}^0).$ 
For data $(Y_i,\bm{S}^0_i),i = 1,\ldots,N,$ 
and assuming that  $\bm{\beta}^*_2,$ $\mu_{jY}(\cdot)$ and $\mu_{jA}(\cdot), j = 1,2$ are all known,
the above developments further imply that one can estimate $\bm{\beta}^*_1$ using
\[
\tilde{\bm{\beta}}_{1N} = \underset{\bm{\beta}_1}{\arg \! \min}
\sum_{i=1}^N \bigl\{ \tilde Y^\dag_i -  \mu_{1Y}(\bm{W}^0_i) - Q_{1c}(\bm{W}^0_i, A_{1i}; \bm{\beta}_1) \bigr\}^2.
\]
Parallel to the second stage problem, 
the optimal model-based treatment decision in Stage 1, assuming
the optimal model-based treatment is also given in Stage 2, would be
$d^{opt}_{1}(\bm{w}; \bm{\beta}^*_1) = I\{\bm{w}^\top \bm{\beta}^*_1 >0\}$,
and may be estimated by $\tilde{d}^{opt}_1(\bm{w}) = 
I\{\bm{w}^\top \tilde{\bm{\beta}}_{1N} >0\}$.

Of course, none of $\bm{\beta}^*_2,$ $\mu_{jY}(\cdot)$ and possibly $\mu_{jA}(\cdot), j = 1,2$ are known in practice; in Section \ref{EstOpt}, we 
establish the properties of the corresponding least squares estimator when these quantities are all
estimated.

{\sc Remark:} As an alternative to \eqref{eq:Ytil-2}, one can substitute 
\begin{equation}
\label{eq:Ytil}
\tilde Y = Y + \bm{S}^\top \bm{\beta}^*_{2} 
\bigl\{ I(\bm{S}^\top  \bm{\beta}^*_{2} > 0 ) - A_2 \bigr\}
\end{equation}
for \eqref{eq:Ytil-2} when calculating $\tilde{\bm{\beta}}_{1N};$
this follows directly from the equalities
$\E( \tilde{Y} | \bm{W}^0 = \bm{w}^0, A_1 = a_1 ) = \E( \tilde{Y}^\dag  | \bm{W}^0 = \bm{w}^0, A_1 = a_1 )$
and 
$\E( \tilde{Y}  | \bm{W}^0 = \bm{w}^0 ) = \E(\tilde{Y}^\dag | \bm{W}^0 = \bm{w}^0 ).$ 
$\Box$

\section{Robust Q-learning: estimation in practice and corresponding theory}
\label{EstOpt}

The two-stage procedure described in the previous section leads to a class of decision rules indexed by finite-dimensional parameter vectors, that is,
$d^{opt}_1(\bm{W};\bm{\beta}_1) =I( \bm{W}^\top \bm{\beta}_1>0)$ and $d^{opt}_1(\bm{S};\bm{\beta}_2)=I( \bm{S}^\top \bm{\beta}_2>0)$ 
\citep[e.g.,][]{chakraborty2013statistical}. 
Although not explicit in prior developments, the variable sets $\bm{S}$ and $\bm{W}$ are each assumed to contain a column of ones, 
so that the main effects of treatment at each stage can be included as part of the decision rule.
The proposed approach eliminates the nuisance parameters $\eta_j(\cdot), j=1,2$ from the first and 
the second stage decision rules at the expense of introducing the four additional unknown functions  $\mu_{jY}(\cdot)$ and $\mu_{jA}(\cdot), j = 1,2.$ The advantage
of the proposed  approach is that the indicated functions depend on observables and can be easily estimated using any nonparametric regression or statistical learning method having sufficiently good prediction performance. Importantly, in the case of a sequentially randomized clinical trial, the functions $\mu_{jA}(\cdot), j = 1,2$ are known and correct models are easily formulated.

\subsection{Estimation in practice}

The  developments in the next two subsections assume that the original sample, with $N$ elements independently and identically distributed
as $P_0,$ has been randomly split into two disjoint and independent samples, say $\bm{D}_{I_n}$ and $\bm{D}_{I^c_n},$ with 
$n = O(N)$ (e.g., $n = N/2$) and where $\bm{I}_n$ and its complement $\bm{I}^c_{n}$ form a partition of the index set $\{1,\ldots,N\}$.
The induced nuisance parameters $\hat \mu_{2Y}(\cdot),$ $\hat \mu_{2A}(\cdot),$ $\hat \mu_{1Y}(\cdot),$ and $\hat \mu_{1A}(\cdot)$
are to be estimated as described earlier using the data in $\bm{D}_{I^c_n};$
the finite dimensional parameters of interest are then estimated using the data $\bm{D}_{I_n},$
treating $\hat \mu_{2Y}(\cdot),$ $\hat \mu_{2A}(\cdot),$ $\hat \mu_{1Y}(\cdot),$ and $\hat \mu_{1A}(\cdot)$
as known functions. 
As described above, the use of such sample-splitting is a particularly simple form of cross-fitting and can be generalized
easily \citep{Cherno18}; our use of sample splitting as described above will be sufficient to establish the main ideas for
both estimation and asymptotics without unnecessarily complicating notation. 
Generalization to cross-fitting is straightforward and will be discussed at the end of Section \ref{EstOpt}.

Let $\hat\mu_{2Y}(\cdot)$, $\hat\mu_{1Y}(\cdot)$, $\hat\mu_{2A}(\cdot)$, and $\hat\mu_{1A}(\cdot)$ denote suitable estimates of 
$\mu_{jY}(\cdot)$ and $\mu_{jA}(\cdot), j = 1,2$ derived from the data in $\bm{D}_{I^c_n}.$
Backward induction, implemented as described earlier with obvious modifications, 
 can be used estimate the optimal dynamic treatment regime. In particular, for the second stage, we compute
\begin{equation}
\label{b2hat-LS}
\hat {\bm{\beta}}_{2n}= 
\underset{\bm{\beta}_2}{\arg \! \min}~
\sum_{i \in \bm{I}_n} \left[ Y_i-\hat\mu_{2Y}(\bm{S}^0_i) - \{A_{2i}-\hat\mu_{2A}(\bm{S}^0_i)\} \cdot \bm{S}_i^\top \bm{\beta}_2    \right]^2.
\end{equation}
To estimate the first stage parameters, we first calculate the estimated first stage pseudo-outcome
\begin{equation}
\label{Yhat}
\hat {\tilde {Y}}_i
=Y_i + \bm{S}_i^\top \hat{\bm{\beta}}_{2n} \{I(\bm{S}_i^\top \hat{\bm{\beta}}_{2n}  >0) -A_{2i}\};
\end{equation}
and then compute
\begin{equation}
\label{b1hat-LS}
\hat {\bm{\beta}}_{1n}= \underset{\bm{\beta}_1}{\arg \! \min}~ \sum_{i \in \bm{I}_n} \left[ \hat {\tilde {Y_i}}-\hat \mu_{1Y}(\bm{W}^0_{i}) - \{A_{1i}-\hat\mu_{1A}(\bm{W}^0_{i})\} \cdot \bm{W}_{i}^\top \bm{\beta}_1    \right]^2.
\end{equation}
The notation in \eqref{b2hat-LS} and \eqref{b1hat-LS} emphasizes the fact that the nuisance parameters
$\mu_{2Y}(\cdot),$ $\mu_{2A}(\cdot),$ $\mu_{1Y}(\cdot),$ and $\mu_{1A}(\cdot)$ are estimated using
the outcome and full set of either second and first stage covariates, whereas the linear specifications used for modeling
the centered Q-functions might not use all available covariate information.

As defined, the pseudo-outcomes are non-smooth functions of the data, hence so is $\hat {\bm{\beta}}_{1n}$; this can cause non-regularity problems for 
$\hat {\bm{\beta}}_{1n}$ \citep{laber2014dynamic}. In particular, when $Pr(| \bm{S}^\top \bm{\beta}^*_{2} | =0)> 0,$ 
i.e., there exists a strata of the covariates $\bm{S}$ used to model the Q-function that occurs with positive probability and 
for which treatment is neither beneficial nor harmful, the estimators of first stage regression coefficients become non-regular 
due to the non-differentiability of the indicator function in  the definition of the pseudo-outcome. 

The proposed Q-learning models essentially utilize the propensity score regression approach of \cite{robins1992estimating} to eliminate the problem 
of mismodeling hard-to-estimate infinite-dimensional parameters (i.e., $\eta_j(\cdot), j=1,2$) on the estimators of the $\bm{\beta}^*_j$s. The resulting 
estimator of $\bm{\beta}^*_j$ is consistent and asymptotically normal under suitable conditions on $\hat\mu_{jA}(.)$ and $\hat\mu_{jY}(.)$, $j=1,2$.  
In particular, the estimate of $\bm{\beta}^*_j$ is robust to misspecification of $\mu_{jY}(\cdot)$ provided that $\mu_{jA}(\cdot)$ is consistently 
estimated where $\bm{\beta}^*_j$ represents the parameters of the best linear approximation of the unknown $\Delta_j(\cdot).$
In practice, we recommend using ensemble learning methods such as Super Learner \citep{van2007super} for estimating 
both $\mu_{jA}(\cdot)$ and $ \mu_{jY}(\cdot)$. Asymptotically, Super Learner performs as well as the best convex combination of the 
base learners in the chosen library, in the sense of minimizing the difference in risk compared to the corresponding oracle estimator. 
Moreover, the size of the library can grow at a polynomial rate compared with the sample size without affecting its oracle performance 
 \citep{van2003unifiedcross, dudoit2003asymptotics, van2006oracle}.  For these reasons, it is recommended that the library consist of a 
 large and diverse set of regression modeling procedures (i.e., nonparametric, semiparametric, parametric).
Importantly, these theoretical results only imply that Super Learner can match the performance of the (unknown) best possible convex 
combination of choices in the specified library. Thus, consistency is not guaranteed unless the corresponding oracle estimator is consistent 
and converges sufficiently fast. However, with the use of a sufficiently flexible library, Super Learner clearly improves one's ability to 
construct a consistent estimator because  it eliminates the need to select and subsequently rely on a single method of estimation.

\subsection{Theoretical results}

To further simplify notation, let $\Delta_{2i} = \Delta_2(\bm{S}^0_{i}),$ $\Delta_{1i} = \Delta_1(\bm{W}^0_{i}),$ 
$\mu_{2Ai} = \mu_{2A}(\bm{S}^0_i),$ $\hat \mu_{2Ai} = \hat \mu_{2A}(\bm{S}^0_i),$ $ \mu_{1Ai} = \mu_{1A}(\bm{W}^0_i),$ 
and $\hat \mu_{1Ai} = \hat \mu_{1A}(\bm{W}^0_i).$  In addition, with
$\bm{x}^{\otimes 2} = \bm{x}  \bm{x}^\top$ for any vector $\bm{x},$
define the matrices
\begin{equation*}
\bV_{2n} = \frac{1}{n} \sum_{i \in \bm{I}_n}   (A_{2i}- \mu_{2Ai})^2  \bm{S}^{\otimes 2}_i  
~\mbox{ and }~
\hat \bV_{2n} = \frac{1}{n} \sum_{i \in \bm{I}_n}   (A_{2i}- \hat \mu_{2Ai})^2  \bm{S}^{\otimes 2}_i   
\end{equation*}
\begin{equation*}
\bV_{1n} = \frac{1}{n} \sum_{i \in \bm{I}_n}   (A_{1i}- \hat \mu_{1Ai})^2  \bm{W}^{\otimes 2}_i  
~\mbox{ and }~
\hat \bV_{1n} = \frac{1}{n} \sum_{i \in \bm{I}_n}   (A_{1i}- \hat \mu_{1Ai})^2  \bm{W}^{\otimes 2}_i.  
\end{equation*}

Let $\| \bm{x} \|_q$ denote the usual Q- norm of a vector $\bm{x}$  for $q=1,2,\infty$. Also, for $Z \sim P$ for 
some probability measure $P,$ suppose $f(\cdot)$ is any real-valued, $P-$measurable function; then, we  
define the $L^2(P)$ norm of $f(\cdot)$ as $\| f(Z) \|_{P,2} = \{ \int \! f(\omega)^2 d P(\omega) \}^{1/2}.$ For a 
real-valued function $h(\bm{s}^0; \bm{D}_{I^c_n})$ defined for $\bm{s}^0 \in \mathcal S$ whose calculation 
may depend on the data contained in $\bm{D}_{I^c_n},$ we can also define the random norm
$\left\|  h(\bm{S}^0; \bm{D}_{I^c_n}) \right\|_{P_0,2}$ as the square-root of
\begin{equation}
\label{P2norm-def}
\left\|  h(\bm{S}^0; \bm{D}_{I^c_n})  \right\|^2_{P_0,2}
= \E \left\{ \left\|  h(\bm{S}^0) \right\|^2_{\mathbb{P}_n,2} \big| \bm{D}_{I^c_n} \right\},
\end{equation}
where $\mathbb{P}_n$ denotes the empirical measure on $\bm{D}_{I_n}.$

Our results are established under the following assumptions.

\begin{assumption} \label{assump:support-X}
(i)   The support of $\bm{W}^0$ and the conditional treatment effect $\Delta_1(\bm{W}^0)$ are uniformly bounded; 
(ii)  the support of $\bm{S}^0$ and the conditional treatment effect $\Delta_2(\bm{S}^0)$ are uniformly bounded; and,
the supports of $\bm{S}$ and $\bm{W}$ are uniformly bounded.
\end{assumption}

\begin{assumption} \label{assump:accuracy-treatment}
  (i) $\left\|  \hat \mu_{1A}(\bm{W}^0; \bm{D}_{I^c_n}) - \mu_{1A}(\bm{W}^0) \right\|^2_{P_0,2} = o_p(n^{-1/2});$ 
  (ii) $\left\|  \hat \mu_{2A}(\bm{S}^0; \bm{D}_{I^c_n}) - \mu_{2A}(\bm{S}^0) \right\|^2_{P_0,2} = o_p(n^{-1/2}).$
\end{assumption}

\begin{assumption}  \label{assump:accuracy-outcome}
  (i) $\left\|  \hat \mu_{1Y}(\bm{W}^0; \bm{D}_{I^c_n}) - \mu_{1Y}(\bm{W}^0) \right\|^2_{P_0,2} = o_p(1);$ 
  (ii) $\left\|  \hat \mu_{2Y}(\bm{S}^0; \bm{D}_{I^c_n}) - \mu_{2Y}(\bm{S}^0) \right\|^2_{P_0,2} = o_p(1).$
\end{assumption}

\begin{assumption}  \label{assump:accuracy-2} 
(i) $\left\|  \hat \mu_{1Y}(\bm{W}^0; \bm{D}_{I^c_n}) - \mu_{1Y}(\bm{W}^0) \right\|_{P_0,2}
\left\|  \hat \mu_{1A}(\bm{W}^0; \bm{D}_{I^c_n}) - \mu_{1A}(\bm{W}^0) \right\|_{P_0,2} = o_p(n^{-1/2});$\\
(ii) $\left\|  \hat \mu_{2Y}(\bm{S}^0; \bm{D}_{I^c_n}) - \mu_{2Y}(\bm{S}^0) \right\|_{P_0,2}
\left\|  \hat \mu_{2A}(\bm{S}^0; \bm{D}_{I^c_n}) - \mu_{2A}(\bm{S}^0) \right\|_{P_0,2} = o_p(n^{-1/2})$
\end{assumption}

\begin{assumption} 
\label{assump:posdef}
There exists $1 \leq n_0 < \infty$ such that $\bV_{jn}$ and $\hat \bV_{jn}, j = 1,2$ are positive definite for $n \geq n_0$.
\end{assumption}

\begin{assumption} 
\label{assump:unique} 
$P\big( | \bm{S}_1^\top  \bm{\beta}^*_{2} | = 0 \big) = 0.$
\end{assumption}

Assumption \ref{assump:support-X} requires no discussion. Assumptions \ref{assump:accuracy-treatment} - \ref{assump:accuracy-2} 
impose reasonable conditions on the estimators of the nuisance parameters estimated using cross fitting that are satisified by many 
machine learning algorithms; see \cite{Cherno18} for further discussion.
Assumption \ref{assump:posdef} imposes reasonable conditions on the existence and uniqueness of the least squares estimators
\eqref{b2hat-LS} and \eqref{b1hat-LS}. Assumption \ref{assump:posdef} combined with independent, identically distributed sampling ensures
that the limiting matrices
\[
\bV_2  =  \E\left\{  \mbox{var}\left( A_{2} | \bm{S}^0 \right) \bm{S}^{\otimes 2} \right\} 
~\mbox{ and }~
\bV_1  =  \E\left\{ \mbox{var}\left( A_{1} | \bm{W}^0 \right) \bm{W}^{\otimes 2} \right\}
\]
both exist and are positive definite.
Finally, Assumption \ref{assump:unique} is imposed to avoid non-regular asymptotic behavior
in the first stage least squares estimator \eqref{b1hat-LS}.  Inferences for the parameters that define the
estimated optimal dynamic treatment regime 
$\hat d^{opt}_{2}(\bm{s}) = I\{\bm{s}^\top \hat {\bm{\beta}}_{2n} >0\}$
and
$\hat d^{opt}_{1}(\bm{w}) = I\{\bm{w}^\top \hat {\bm{\beta}}_{1n} >0\}$
can now be derived using the results in the following theorem.
\begin{theorem}
\label{main thm}
Suppose that Assumptions \ref{assump:support-X} - \ref{assump:unique} hold. 
\begin{itemize}
\item[(a)] Let $\hat {\bm{\beta}}_{2n}$ be given by \eqref{b2hat-LS}. Then,
$
\sqrt{n} ( \hat {\bm{\beta}}_{2n}  - {\bm{\beta}}^*_{2} ) 
  \stackrel{d}{\rightarrow} N( \bm{0},   \bV^{-1}_2 \bQ_2 \bV^{-1}_2 )
  $
where the matrices $\bV_2  =  \E\left\{  \mbox{var}\left( A_{2} | \bm{S}^0 \right) \bm{S}^{\otimes 2} \right\}$ 
and $\bQ_2 = \E ( \bm{J}_2^{\otimes 2} )$ for
\[
\bm{J}_2 = 
\{A_{2}-  \mu_{2A}(\bm{S}^0)\}  \bm{S}  
\left[
Y - \mu_{2Y}(\bm{S}^0)  - \{ A_{2}-  \mu_{2A}(\bm{S}^0) \}  \bm{S}^{\top} {\bm{\beta}}^*_{2}
\right].
\]
\item[(b)] Let $\hat {\bm{\beta}}_{1n}$ be given by \eqref{b1hat-LS}. Then,
$
\sqrt{n} ( \hat {\bm{\beta}}_{1n}  - {\bm{\beta}}^*_{1} ) 
  \stackrel{d}{\rightarrow} N( \bm{0},   \bV^{-1}_1 \bQ_1 \bV^{-1}_1 )
  $
where the matrices $\bV_1  =  \E\left\{ \mbox{var}\left( A_{1} | \bm{W}^0 \right) \bm{W}^{\otimes 2} \right\}$
and $\bQ_1 = \E \{ ( \bm{J}_1 + \bm{K} \bV^{-1}_2 \bm{J}_2) ^{\otimes 2} \}$ for
\[
\bm{J}_1 = \{A_{1}-  \mu_{1A}(\bm{W}^0)\}  \bm{W} 
\left[
\tilde{Y}- \mu_{1A}(\bm{W}^0) - \{ A_{1} - \mu_{1A}(\bm{W}^0) \} \bm{W}^{\top} {\bm{\beta}}^*_{1} 
\right]
\]
and
\[
\bm{K} = \E \big[ \{ A_{1}-  \mu_{1A}(\bm{W}^0) \} \{ I(\bm{S}^\top   \bm{\beta}^*_{2}  > 0 ) - A_{2} \} \bm{W}  \bm{S}^{\top} \big].
\]
\end{itemize}
\end{theorem}
The following corollary to Theorem \ref{main thm} shows
that Assumption \ref{assump:unique} is not required to establish the results
in part (b) in certain settings, in contrast to the standard form of Q-learning. 

\begin{cor}
Suppose Assumptions \ref{assump:support-X} - \ref{assump:posdef} hold. In addition, suppose 
$
\E\left( A_{1i} - {\mu}_{1Ai}  \mid \bm{W}_{i}, \bm{S}_{i},  \bm{I}_n \right) = 0, i \in \bm{I}_n.
$
Then, Theorem \ref{main thm}, part (b) remains true.
\label{cor1}
\end{cor}

The set of variables $\bm{S}$ used for modeling the second stage decision rule are those thought to be potential effect modifiers
for the second stage treatment assigment $A_2;$ hence, a sufficient condition for the Corollary to hold is that $A_{1}$ is 
independent of $\bm{S}$, conditionally on the set of pre-treatment covariates $\bm{W}$ included in the first stage model. Note
that this does not preclude the possibility that $A_1$ can affect variables in $\bm{S}^0$ that are not part of $\bm{S}.$
Due to the way in which $A_{1i} - {\mu}_{1Ai}$ enters the estimating equation of \citet[][Eqn.\ (4)]{wallace2015doubly}, 
it is unclear whether their approach avoids non-regularity under the same conditions as Corollary \ref{cor1} even
when the propensity model $\mu_{1A}(\bm{W}^0)$ has been correctly specified. 

\subsection{Generalization to cross-fitting}
\label{sec:cfit}
Sample splitting, as used in the previous two sections, does not make use of the full sample of $N$ observations
to estimate the finite-dimensional regression parameters, and this can negatively impact efficiency. We now
describe an alternative approach, cross-fitting, that uses the full sample to estimate the desired target parameters.

Suppose that $N = n K$ for some integer $n$ and some integer $K \geq 2.$ Using an extension of previous
notation, we first randomly split the original sample into 
disjoint (hence independent) samples $(\bm{D}_{I_{n,k}})_{k=1}^K$ such that the size of each sample is $n=N/K$ and 
$\bm{I}_{n,k}, k = 1,\ldots,K$ partition the indices $\{1,\ldots,N\}.$  
Analogously to before, define $\bm{I}_{n,k}^c$ as the set of sample indices that are not included in $\bm{I}_{n,k};$ that is, 
$\bm{I}_{n,k}^c = \{1,2,\cdots,N\} \setminus \bm{I}_{n,k}, k=1,\ldots,K$. Then, for each $k=1,2,\cdots,K$, estimate the 
nuisance parameters  $\mu_{2Y}(\cdot),$ $ \mu_{2A}(\cdot),$ $ \mu_{1Y}(\cdot),$ and $ \mu_{1A}(\cdot)$ using the data in $\bm{D}_{I_{n,k}^c};$
we respectively denote these estimators 
$\hat \mu_{2Y}(\cdot;\bm{D}_{I_{n,k}^c}),$ $\hat \mu_{2A}(\cdot;\bm{D}_{I_{n,k}^c}),$ $\hat \mu_{1Y}(\cdot;\bm{D}_{I_{n,k}^c}),$ 
and $\hat \mu_{1A}(\cdot;\bm{D}_{I_{n,k}^c}),$ $k=1,\ldots K.$ 
Finally, we define
\begin{equation}
\label{b2hat-LS-cf}
\hat {\bm{\beta}}_{2n}= 
\underset{\bm{\beta}_2}{\arg \! \min}~
\sum_{k=1}^K \sum_{i \in \bm{I}_{n,k}} \left[ Y_i-\hat\mu_{2Y}(\bm{S}^0_i;\bm{D}_{I_{n,k}^c}) - \{A_{2i}-\hat\mu_{2A}(\bm{S}^0_i;\bm{D}_{I_{n,k}^c})\} \cdot \bm{S}_i^\top \bm{\beta}_2    \right]^2
\end{equation}
and
\begin{equation}
\label{b1hat-LS-cf}
\hat {\bm{\beta}}_{1n}= \underset{\bm{\beta}_1}{\arg \! \min}~ 
\sum_{k=1}^K \sum_{i \in \bm{I}_{n,k}} \left[ \hat {\tilde {Y_i}}-\hat \mu_{1Y}(\bm{W}^0_{i};\bm{D}_{I_{n,k}^c}) - \{A_{1i}-\hat\mu_{1A}(\bm{W}^0_{i};\bm{D}_{I_{n,k}^c})\} \cdot \bm{W}_{i}^\top \bm{\beta}_1    \right]^2.
\end{equation}
This form of cross-fitting essentially corresponds to ``DML2'' as described in \citet[Def.\ 3.2]{Cherno18}. Like sample splitting,  cross-fitting helps
to guarantee that some of the remainder terms in the asymptotic linearity expansion converge to zero at an appropriately fast rate. However, in contrast
to sample splitting, cross-fitting as described above is also capable of asymptotically achieving the same efficiency as in the case 
where estimators of the regression parameters are computed using all $N$ observations (i.e., with $\mu_{jA}(\cdot)$ and $\mu_{jY}(\cdot)$ $j=1,2$ being known).

\section{Simulation studies}
\label{sec:sim}
 We examined the performance of our proposed Q-learning method under different simulation scenarios 
 with various functional complexities and degrees of non-regularity (i.e., violation of Assumption \ref{assump:unique}). 

The main simulation in the regular setting uses the following data generation mechanism. 
Let $\bm{X}_1=(X_{11},X_{12},X_{13},X_{14},X_{15})^\top$ be a 5-dimensional vector of baseline covariates independently
generated and uniformly distributed on [-0.5,0.5]. 
Let  $\bm{X}_2 = (X_{21},X_{22},X_{23},X_{24},X_{25})^\top,$  $X_{2l} = X_{1l} + U_{l},~ l=1,2,3;$ 
$X_{24} = 0.35 X_{15} + U_{4};$ and, $X_{25} = U_{5},$ where $U_{l}, l =1,\ldots,5$ are independent and uniformly
distributed on [-0.5,0.5]. 
It is assumed that only nonresponders to the first stage treatment will receive the second stage treatment. This nonresponse indicator 
$R$ equals 1 if $X_{24}$ is less than its median value and is 0 otherwise. 
Finally, the first and second stage treatments $A_{j}$ are generated from a Bernoulli distribution with success probability 
$\mu_{jA}(\cdot) = [ 1+\exp\{-\lambda_{jA}(\cdot) \} ]^{-1}$, where  $\lambda_{jA}(\cdot)$
depends on either $\bm{S}^0 = (\bm{X}^\top_1,A_{1}, \bm{X}^\top_2)^\top$ ($j$=2) or 
$\bm{W}^0 = \bm{X}_1$ ($j$=1); see Section \ref{sec:reg}.

\subsection{Performance: regular setting} \label{sec:reg}

In this case, we consider performance for models that satisfy Assumption \ref{assump:unique}.
To implement our proposed method, we used the \texttt{R} package \texttt{SuperLearner} \citep{SLref} to estimate 
$\mu_{1Y}(\cdot)$, $\mu_{2Y}(\cdot)$, $\mu_{1A}(\cdot)$, and $\mu_{2A}(\cdot)$.  The library used for \texttt{SuperLearner} included 
generalized linear models (i.e., \texttt{glm}), generalized additive models (i.e., \texttt{gam}; \citealp{GAMref}), multivariate adaptive 
regression splines (i.e., \texttt{earth}; \citealp{MARSref}),
random forests (i.e., \texttt{randomForest}; \citealp{RFref}),  and  
support vector machines (i.e., \texttt{svm} from the \texttt{R} package \texttt{e1071}; \citealp{SVMref});
estimation was implemented with all tuning parameters set to their respective default values.  
This simulation study uses four different functional forms for the treatment assignment model
$\mu_{jA}(\cdot) = [ 1+\exp\{-\lambda_{jA}(\cdot) \} ]^{-1}$, $j=1,2$:
 \begin{itemize}
\item {\it Randomized:} $\lambda_{jA}(\cdot) =0$
\item {\it Linear:} $\lambda_{jA}(\cdot) =2X_{j1}+2X_{j2}+X_{j3}+0.1X_{j4}+0.1X_{j5}$
\item  {\it Quadratic}: $\lambda_{jA}(\cdot) = 1.4 \{(X_{j1}-0.5)^2+(X_{j2}-0.5)^2+0.6(X_{j3}-0.5)^2+0.5(X_{j4}-0.5)^2+0.5(X_{j5}-0.5)^2+X_{j1}+X_{j2}+0.6X_{j3}+0.5X_{j4}+0.5X_{j5}-2\}$.
\item {\it InterQuad}: $\lambda_{jA}(\cdot) = 1.4 \{(X_{j1}-0.5)^2+(X_{j2}-0.5)^2+0.6(X_{j3}-0.5)^2+0.5(X_{j4}-0.5)^2+0.5(X_{j5}-0.5)^2+X_{j1}+X_{j2}+0.6X_{j3}+0.5X_{j4}+0.5X_{j5}+X_{j1}X_{j2}-2\}.$\end{itemize}

The {\it Randomized} model corresponds to a SMART-like trial where simple randomization is used at baseline and then
simple re-randomization occurs among the set of non-responders at the first stage.  The randomization model, part of the trial design,
is therefore known and the inclusion of an appropriate \texttt{glm} model in the Super Learner library should ensure that 
$\mu_{jA}(\cdot)$ can be consistently estimated at the usual parametric rate. 
The other three settings are meant to correspond to increasingly complex observational data settings,  where the ``assignment'' mechanism by 
which patients follow a particular treatment regimen is covariate-dependent, not randomized, and is not considered to be known by design. 
Hence, the analyst cannot knowingly select a correctly specified parametric model a priori. 
In the case of the {\it Linear} model, the inclusion of a \texttt{glm} model in the Super Learner library again 
ensures that $\mu_{jA}(\cdot)$ can be consistently estimated at the usual parametric rate. For the
other two models, the inclusion of methods such as \texttt{gam} and \texttt{randomForest} will help to mitigate, 
but not necessarily eliminate, the possibility of inconsistent estimation. 
These observations highlight the importance of using flexible methods when modeling $\mu_{jA}(\cdot), j = 1,2,$ 
particularly in observational data settings.

The outcome models are given by :
\begin{itemize}
\item Linear$^R$: $Y= \bX_1^\top \balpha_{1}+\bX_2^\top \balpha_{2}+A_1\bX_1^\top \btheta_{1}+ A_2 R \bX_2^\top \btheta_{2}+\epsilon$ 
where $\balpha_{1}=\balpha_{2}=(1,0.1,0.1,0.1,0.1)^\top$, $\btheta_{1}=(0,0,0,0,0)^\top$ and $\btheta_{2}=(1,1,0,0,0)^\top$; 
\item FGS$^R$: $Y= f(\bX_1)+f(\bX_2)+A_1\bX_1^\top \btheta_{1}+A_2 R g(\bX_2)+\epsilon$ where $\btheta_{1}=(0,0,0,0,0)^\top$ and
for $\bm{x} = (x_1,x_2,x_3,x_4,x_5)^\top,$ we set  $g(\bm{x} ) = 2 \sin(\pi x_1 x_2) + 2
  (x_2-0.5)^2$ and
\[
f(\bm{x}) = -1.5 + \sin(\pi x_1 x_2) + 2
  (x_3-0.5)^2 +  x_4 + 1.5 \frac{x_1}{|x_2|+|x_3|} +2 x_1 (x_2+x_3).
  \]
\end{itemize}
The noise variable $\epsilon$ is generated from $ \mathrm{N}(0,\sigma=0.5)$.  

In connecting the above Linear$^R$ outcome model specification with earlier notation, we have  
$\bm{S}^0 = (\bm{X}_1^\top ,A_1,\bm{X}_2^\top)^\top$, $\eta_2(\bm{S}^0) = \bX_1^\top \balpha_{1}+\bX_2^\top \balpha_{2} + A_1\bX_1^\top \btheta_{1}$ 
and $\Delta_2(\bm{S}^0) = R \bX_2^\top \btheta_{2},$ where $R$  is a function of $X_{24}$ only; we further have $\bm{W}^0 = \bm{X}_1$. 
We respectively use $\bm{S} = R (1,X_{21}, X_{22},X_{23})^\top$ and $\bm{W} = (1,X_{11},X_{12})^\top$ for modeling the relevant Q-functions.
In this case, the target of estimation $\bbeta^*_2 = (0,\theta_{21},\theta_{22},\theta_{23})^\top$ and
it can additionally be shown that \eqref{eq:cenQ2} coincides with \eqref{eq:cenQ2mod}.
However, for the FGS$^R$ outcome model, $\eta_2(\bm{S}^0) = f(\bX_1)+f(\bX_2)+A_1\bX_1^\top \btheta_{1}$ and 
$\Delta_2(\bm{S}^0)  = R g(\bX_2);$ here, \eqref{eq:cenQ2} does not coincide with \eqref{eq:cenQ2mod} since the 
linear parametric specification used in the latter is not
equal to $\Delta_2(\bm{S}).$  In this case $\bm{S}^\top\bbeta^*_2$ still exists as the best linear 
projection of $\Delta_2(\bm{S}^0)$ on to the linear space spanned by $\bm{S};$ however, 
its value for this simulation study must be determined  numerically (e.g., through simulation). 

In general, it is not similarly straightforward to characterize the functions $\eta_1(\bm{W}^0)$ and $\Delta_1(\bm{W}^0),$ or the value of $\bbeta^*_1$ 
in the first stage models, without appealing to numerical methods. However, in the current simulation setting, the value of $\bbeta^*_1$ can 
be determined exactly for both the Linear$^R$ and FGS$^R$ outcome model specifications. Specifically, 
neither model involves an interaction between $A_1$ and $A_2$; more generally, there is no correlation between $A_1$ and the
second stage variables $A_2$ and $\bX_2.$ As a result, the linear term  $A_1\bX_1^\top \btheta_{1}$ that appears in both the 
Linear$^R$ and FGS$^R$ outcome model specifications accurately describes the interaction between 
treatment $A_1$ and $\bX_1$ in the true first stage Q-function (i.e., $\Delta_1(\bm{W}^0) = \bm{W}^\top \btheta_{1} = \bX_1^\top \btheta_{1} = \bm{0})$.
It follows that $\bbeta^*_1 = \bm{0}$ and hence that expression \eqref{eq:cenQ1} also coincides with \eqref{eq:cenQ1mod}.

\begin{table}[h]
\centering
\caption{Performance of the proposed Q-learning method for
estimating the second stage parameters under different model complexities.
The true parameters for the linear and FGS outcome models are respectively $\beta^*_{2,1}=1$, $\beta^*_{2,2}=1$ and 
$\beta^*_{2,1} \approx 0$, $\beta^*_{2,2} \approx -2.$ }
\resizebox{\textwidth}{!}
{\begin{tabular}{lcccccc|cccccc|cccc}
  \hline
           & \multicolumn{6}{c}{$\beta^*_{2,1}$} &\multicolumn{6}{c}{$\beta^*_{2,2}$} \\
& \multicolumn{2}{c}{Q$_{N,N}$} &   \multicolumn{2}{c}{Proposed} & \multicolumn{2}{c}{dWOLS$_{N,N}$} & \multicolumn{2}{c}{Q$_{N,N}$} & \multicolumn{2}{c}{Proposed} & \multicolumn{2}{c}{dWOLS$_{N,N}$}  \\
Outcome &  Bias & S.D.&  Bias & S.D. & Bias & S.D.  & Bias & S.D. & Bias & S.D. & Bias & S.D.  \\
  \hline
\multicolumn{13}{c}{\it Randomized Treatment Assignment Model} \\ \hline
Linear$^R$&0.003&0.041&  0.003 & 0.081&0.002 & 0.082 & 0.001 & 0.038  & 0.004 & 0.081& 0.002 & 0.076 \\
FGS$^R$& 0.041& 0.404 & 0.005 & 0.514  & 0.024 & 0.763& 0.002 & 0.241&   0.037&0.211  &   0.031 & 0.377\\\hline
\multicolumn{13}{c}{\it Linear Treatment Assignment Model} \\ \hline
Linear$^R$& 0.004&0.041 & 0.004 & 0.101&0.006 & 0.095 & 0.003& 0.042 & 0.000 & 0.101  & 0.002& 0.098\\
FGS$^R$& 2.500&  0.368 &  0.060 & 0.662  & 0.050 & 0.886 &     2.527& 0.238  &0.064 & 0.365& 0.055 & 0.526\\ \hline
\multicolumn{13}{c}{\it Quadratic Treatment Assignment Model} \\ \hline
Linear$^R$&0.006& 0.040  & 0.005 & 0.082  &0.007 & 0.081  & 0.004 &0.040& 0.011 & 0.082 & 0.003&0.082\\
FGS$^R$&0.797  & 0.419  & 0.093 & 0.586 &0.811 & 0.827 & 0.012 &0.247  & 0.017 & 0.276 & 0.022 & 0.409 \\ \hline
\multicolumn{13}{c}{\it InterQuad Treatment Assignment Model} \\ \hline
Linear$^R$& 0.000& 0.041 & 0.014 & 0.094  &0.001 & 0.084& 0.002  & 0.039  &0.019 & 0.086 & 0.002 & 0.079\\
FGS$^R$& 0.749& 0.470&  0.070 & 0.612 &0.758 & 0.916& 0.442  &  0.234 & 0.019 & 0.271& 0.455 & 0.402\\
\hline
\end{tabular}}
\label{tab:rbeta2}
\end{table}

In our main simulation study, there are 8 possible model combinations represented by the outcome and treatment assignment models,
and within each setting we compare the performance of the proposed method for estimating $\bbeta_j^*,j=1,2$
to the standard form of Q-learning (Q$_{N,N}$) and also to the weighted least squares (dWOLS$_{N,N}$) estimator 
proposed by \cite{wallace2015doubly}. The subscripts on these latter two estimators denote the fact that 
standard errors would normally be calculated using the $N$-out-of-$N$ bootstrap (i.e., in the regular setting). 
In the case of dWOLS$_{N,N},$ linear models are used for the relevant Q-function model specification and logistic 
regression models are used for estimating the treatment assignment probabilities.
The estimation of $\bbeta_j^*,j=1,2$  is not subject to residual confounding
bias for any of the proposed methods under the Linear$^R$ outcome model specification.
However, there is a possibility of such bias under the FGS$^R$ in the case of Q$_{N,N}$ and
dWOLS$_{N,N}$. To be more specific, residual confounding bias under the FGS$^R$ outcome model
is expected for Q$_{N,N}$ regardless of the treatment assignment model.
For dWOLS$_{N,N}$, the {\it Randomized} and {\it Linear} first and second stage treatment assignment models are 
correctly specified and easily modeled. Hence, under the FGS$^R$ outcome model specification, a significant potential for bias 
arises only under the  {\it Quadratic} or {\it InterQuad} treatment assignment rules.  
For the proposed method, residual confounding bias when estimating $\bbeta_j^*$ is not anticipated provided 
that $\mu_{jA}(\cdot), j=1,2$ are sufficiently well-estimated.

We generate 500 datasets of size 2000 to examine the performance of our proposed method
and use cross-fitting as described in Section \ref{sec:cfit} with $K=2$ to estimate the desired target parameters.
Tables \ref{tab:rbeta2} and \ref{tab:rbeta1} show the empirical absolute bias and standard deviations of the 
second and first stage parameter estimates (i.e., standard errors). The values of $\beta^*_{2,1} \approx 0$ and $\beta^*_{2,2} \approx -2$ 
are determined by simulation.  As expected, standard Q-learning performs poorly except under the {\it Randomized}  
treatment assignment model.  The proposed method and dWOLS$_{N,N}$ also perform similarly well under the 
{\it Randomized} and {\it Linear} treatment assignment models for estimating the first and second stage parameters. 
However, under the  FGS$^R$ outcome model, the proposed method exhibits similar biases and substantially smaller
standard errors.
For the {\it Quadratic} and {\it InterQuad} treatment assignment mechanism, both of which are mis-modeled 
in the case of dWOLS$_{N,N},$ the corresponding estimators show substantial bias in some of the parameters, 
whereas those for the proposed method remain comparatively low.
For example, under the  {\it InterQuad} treatment assignment model and FGS$^R$ outcome model,  the proposed method 
respectively results in estimators for $\beta^*_{21}$ and $\beta^*_{22}$ with absolute biases of 0.070 and 0.019; in contrast, 
those for the dWOLS$_{N,N}$ estimators are 0.758 and 0.455, respectively.  We again see a substantial reduction in standard 
errors; in this same example, the standard errors under the proposed method are 0.612 and 0.271, whereas for dWOLS$_{N,N}$ these
are respectively 0.916 and 0.402, the degree of reduction exceeding 30\%. Overall, the proposed method is observed to be
more robust,  typically producing less biased estimators with smaller standard errors compared with the other two approaches.  

The performance of our proposed method was also assessed using smaller sample sizes. Tables S1-S6 in the supplementary material respectively show the 
results for $\bbeta_2^*$ and $\bbeta_1^*$ with sample sizes of 1000, 500, and 250. Overall, the proposed method continues to outperform both Q$_{N,N}$ and dWOLS$_{N,N},$ particularly when the underlying treatment assignment and the outcome models are both nonlinear (i.e., settings in which bias
can be expected for both Q$_{N,N}$ and dWOLS$_{N,N}$.). However, the performance of the proposed method is also affected by sample size. 
For example, under the {\it Linear} treatment assignment and FGS$^R$  outcome model with a sample size of $N=$250, the proposed method 
shows unacceptably high bias when estimating $\bm{\beta}^*_2$ when compared to dWOLS$_{N,N};$ see Table S5. We conjecture that this 
occurs because the information available for estimating the second stage propensity model is limited to the set
non-responders (i.e., 50\% of the sample) that are re-randomized.  The value functions for the estimated rules in all cases were
calculated for all sample sizes and show that the proposed method, followed by dWOLS, typically results in value functions that are closest to optimal; 
see Section 9.4 of the supplementary materials.

The supplementary material includes simulation results in which \texttt{SuperLearner} is replaced by alternative data adaptive techniques. Specifically, in 
Tables S7 and S8 in the supplementary material, the columns RF-RF and GAM-GAM represent modeling approaches in which \texttt{randomForest}
and  \texttt{gam} are used for both the marginalized outcome (i.e., $\mu_{1Y}(\cdot)$ and $\mu_{2Y}(\cdot)$) and treatment assignment models  
(i.e., $\mu_{1A}(\cdot)$ and $\mu_{2A}(\cdot)$). The column RF-GAM instead uses \texttt{randomForest} for the outcome model and \texttt{gam}
for the treatment assignment model.  Comparing these results with those summarized in Tables \ref{tab:rbeta2} and \ref{tab:rbeta1} 
shows that the use of \texttt{SuperLearner} improves performance.

\begin{table}[t]
\centering
\caption{Performance of the proposed Q-learning method for estimating the first stage parameters under different model complexities. 
The true parameters are $\beta^*_{1,1}=\beta^*_{1,2}=0$.}
\resizebox{\textwidth}{!}{\begin{tabular}{lcccccc|cccccc|cccc}
  \hline
           & \multicolumn{6}{c}{$\beta^*_{1,1}$} &\multicolumn{6}{c}{$\beta^*_{1,2}$} \\
& \multicolumn{2}{c}{Q$_{N,N}$} &   \multicolumn{2}{c}{Proposed} & \multicolumn{2}{c}{dWOLS$_{N,N}$} & \multicolumn{2}{c}{Q$_{N,N}$} & \multicolumn{2}{c}{Proposed} & \multicolumn{2}{c}{dWOLS$_{N,N}$}  \\
Outcome &  Bias & S.D.&  Bias & S.D. & Bias & S.D.  & Bias & S.D. & Bias & S.D. & Bias & S.D.  \\
  \hline
\multicolumn{13}{c}{\it Randomized Treatment Assignment Model} \\ \hline
Linear$^R$&0.007& 0.101&  0.001 & 0.105&0.001 & 0.103 & 0.005& 0.092 & 0.001 & 0.105  & 0.007 & 0.100 \\
FGS$^R$& 0.000 & 0.544 &  0.062 & 0.577 & 0.073 & 0.826& 0.000& 0.395& 0.008 & 0.404& 0.003 & 0.542\\\hline
\multicolumn{13}{c}{\it Linear Treatment Assignment Model} \\ \hline
Linear$^R$& 0.173 & 0.103&  0.001 & 0.116 & 0.005 & 0.114 & 0.163& 0.094 &  0.002 & 0.116 & 0.003 & 0.120    \\
FGS$^R$& 1.915& 0.582  &  0.058 & 0.693  & 0.043 & 0.865 &  1.657 & 0.433   & 0.010 & 0.497 &0.016 & 0.617   \\\hline
\multicolumn{13}{c}{\it Quadratic Treatment Assignment Model} \\  \hline
Linear$^R$& 2.404& 0.093 & 0.015 & 0.114  &0.003 & 0.117 &  0.676 & 0.087  & 0.003 & 0.114& 0.000 & 0.122\\
FGS$^R$& 7.235& 0.598 &   0.009 & 0.701&0.281 & 0.820& 1.631& 0.400& 0.044 & 0.503& 0.026 & 0.619\\ \hline
\multicolumn{13}{c}{\it InterQuad Treatment Assignment Model} \\ \hline
Linear$^R$& 2.316&  0.091 &0.001 & 0.120  &0.006 & 0.115 &0.430& 0.089 &0.012 & 0.113 & 0.009 & 0.111\\
FGS$^R$& 7.500&  0.584 &  0.036 & 0.687 &0.344 & 0.862& 2.470&0.414 &0.070 & 0.506& 0.169 & 0.652\\
\hline
\end{tabular}}
\label{tab:rbeta1}
\end{table}


\subsection{Performance: non-regular setting} 
\label{sim-nonreg}

The treatment assignment models considered here are respectively {\it Randomized}, {\it Linear} and {\it InterQuad}, 
defined as in Section \ref{sec:reg}.
Additionally, define 
$\tilde{\bm{X}}_2=(\tilde X_{21},\tilde X_{22},X_{23},X_{24},X_{25})^\top$ where 
$\tilde X_{21}$ is generated from a Bernoulli distribution with 
success probability $[ 1+\exp \{-(2X_{11}+2X_{12}-1)\} ]^{-1},$ 
$\tilde X_{22}$ is generated from a Bernoulli distribution with 
success probability $[ 1+\exp \{-(2X_{12}+X_{21}-1)\} ]^{-1},$ 
$X_{23} = U_{1},$
$X_{24} = 0.35 X_{15} + U_{2},$
and $X_{25} = U_{3},$
where $U_{l}, l =1,\ldots,3$ are independent and uniformly
distributed on [-0.5,0.5]. 
We consider the following outcome models:
\begin{itemize}
\item Linear$^{NR,\varpi}$: $Y= \bX_1^\top \balpha_{1}+ \tilde \bX_2^\top \balpha_{2}+A_1\bX_1^\top \btheta_{1}+  A_2 \cdot (\theta_2  R  \tilde X_{21})+\epsilon$ where $\balpha_{1}=\balpha_{2}=(1,0.1,0.1,0.1,0.1)^\top,$ $\btheta_{1}=(0,0,0,0,0)^\top$ and $\theta_2 = 2 \varpi.$
\item Non-linear$^{NR,\varpi}$: $Y= f(\bX_1)+A_1\bX_1^\top \btheta_{1}+  A_2 \cdot (\theta_2 R  \tilde X_{21})+\epsilon$ where $\btheta_{1}=(0,0,0,0,0)^\top,$
$\theta_2 = 2 \varpi$ and, for $\bm{x} = (x_1,x_2,x_3,x_4,x_5)^\top,$ we set
\[
f(\bm{x}) = -1.5 + \sin(\pi x_1 x_2) + 2
  (x_3-0.5)^2 +  x_4 + 1.5 \frac{x_1}{|x_2|+|x_3|} +2 x_1 (x_2+x_3).
\]
\end{itemize}   
The noise variable $\epsilon$ is generated from $ \mathrm{N}(0,\sigma=0.5)$ and the constant $\varpi \in \{0,1\}$ specifies the degree of non-regularity,
as will be discussed further below. 
In the above models, $\bm{S}^0 = (\bm{X}^\top_1,A_1,\tilde{\bm{X}}^\top_2)^\top$, $\bm{W}^0 = \bm{X}_1,$
$\Delta_2(\bm{S}^0) = \theta_2 R  \tilde X_{21},$ $\eta_2(\bm{S}_0)$ is determined by the remaining model 
terms, and $R$  is a function of $X_{24}$ only.
The second and first stage Q-functions are respectively modeled as linear functions of
$\bm{S}= R (1, \tilde{X}_{21}, \tilde{X}_{22}, X_{23})^\top$ and $\bm{W} = (1, X_{11}, X_{12})^\top.$ 
For both models, it is not difficult to show that $\bm{\beta}^*_{2} = (0,2 \varpi,0,0)^\top$ and that
$\bbeta^*_1 = \bm{0}.$

In both scenarios, for each subject $i$, the first-stage pseudo outcome is defined as in 
\eqref{eq:Ytil} and estimated by substituting in $\hat{\bm{\beta}}_{2n}$ for
$\bm{\beta}^*_{2}.$   The construction of the pseudo-outcome, specifically the projection 
$\bm{S}^\top \bbeta^*_2,$ violates  Assumption \ref{assump:unique}. 
In particular,  $\varpi=0$ corresponds to no second-stage effect modifier, implying that 
$P( |\bm{S}^\top \bbeta^*_2| = 0) = 1$ because $ \bbeta^*_2 = \bm{0}.$
Setting $\varpi=1$ instead implies that there is no second-stage treatment effect
when $R \tilde{X}_{21} = 0$, and a reasonably strong
effect when $R \tilde{X}_{21} = 1;$ in this case,
$0 < P( |\bm{S}^\top \bbeta^*_2| = 0) < 1.$ 
However, the conditions of Corollary \ref{cor1} hold in each case because
$\bm{S}$ does not include $A_1,$ resulting in regular asymptotic behavior for the proposed method. 

Because these simulations focus on coverage rather than bias and standard error, we simulate 1000 datasets of size $N=2000$.
In the non-regular setting considered here, neither Q$_{N,N}$ nor  dWLOS$_{N,N}$ can necessarily be expected to perform well;
hence, we compare our proposed method to a modified version of standard Q-learning and doubly robust weighted least squares
in which the first stage confidence intervals are respectively constructed using a $m$-out-of-$N$ bootstrap technique
as developed in \cite{chakraborty2013statistical} (i.e., Q$_{m,N}^{\kappa}$) and \cite{simoneau2017non} (i.e.,
dWLOS$_{m,N}^{\kappa}$). 
In both of these approaches, the tuning parameter $\kappa \in [0,1)$ determines
the bootstrap sample size $m;$ here, $\kappa=0.05.$   
Table \ref{tab:nr} summarizes the results; for comparison, results obtained using the $N$-out-of-$N$ bootstrap in the first stage
are provided in Table S9 in the supplementary material. 
In these tables, empirical coverages that are significantly over or under the nominal level 0.95 are indicated with a dagger,
with significance being assessed using a binomial test.

The performance of both Q$_{N,N}$ and Q$_{m,N}^{\kappa=0.05}$ relies heavily on the correct 
specification of the outcome model. In those cases where both methods are observed to 
exhibit reasonable performance, Tables \ref{tab:nr} and S9 respectively show that 
$Q_{m,N}^{\kappa=0.05}$ typically over-covers whereas 
$Q_{N,N}$ either under-covers or has close to nominal coverage; 
in contrast, when $Q_{N,N}$ is observed to under-cover to a very significant 
extent, so does $Q_{m,N}^{\kappa=0.05}.$

In general, both dWLOS$_{m,N}^{\kappa=0.05}$ and the proposed method lead to significant improvements in performance. 
Indeed, the proposed method produces valid confidence intervals with coverages close to the nominal level throughout
Tables  \ref{tab:nr} and S9. This can be readily explained by the fact that each setting satisfies the assumptions of Corollary \ref{cor1} despite 
violating Assumption \ref{assump:unique}.
Similarly, we see that dWLOS$_{m,N}^{\kappa=0.05}$ performs reasonably well regardless of the outcome model for both the {\it Randomized}
and {\it Linear} treatment assignment models, since in these two cases the latter can be consistently estimated at a parametric rate.
However, compared to the proposed method, the coverages tend to be slightly conservative, with longer confidence intervals. In these same
cases, dWLOS$_{N,N}$ also performs reasonably, though does have a tendency to under-cover.
Under the {\it InterQuad} treatment assignment model, the performance of both dWLOS$_{m,N}^{\kappa=0.05}$ and
dWLOS$_{N,N}$ declines due to misspecification of the treatment assignment model,
and in the {\it Non-linear}$^{NR,\varpi}$ setting, also the outcome model.
For example,  dWOLS$_{m,N}^{\kappa=0.05}$ exhibits coverage rates as low as 87\%. We conjecture that the combination of 
non-regularity, model  misspecification and residual confounding are the main reasons for the poor performance of Q$_{N,N},$
Q$_{m,N}^{\kappa=0.05}$ and, where observed to be poor, both dWLOS$_{m,N}^{\kappa=0.05}$ and dWLOS$_{N,N}.$
In comparing the two approaches to bootstrapping for both standard Q-learning and dWOLS, our results 
further suggest that tuning $m$ differently (i.e., increasing $m$) may result in better agreement with the nominal 
coverage level in cases where the relevant models are appropriately specified.


Finally, we conducted a related simulation study in which both Assumption \ref{assump:unique} and the conditions of Corollary \ref{cor1}
are violated. Unlike the simulation settings above, this example considers a case where the first and second stage treatments
interact with each other. This modified study is described in Section 9.3 of the supplementary document, where we compare
the proposed approach with dWLOS$_{m,N}^{\kappa=0.05};$ the results are summarized in Table S10. Overall, the methods
perform as expected. In particular,  the proposed method demonstrates either nominal or modest undercoverage for the first stage  regression
parameters and dWLOS$_{m,N}^{\kappa=0.05}$ demonstrates conservative coverage except in cases where 
the required conditions for consistency are violated.

\begin{table}[t]
\centering
\caption{Performance of proposed Q-learning method under different levels of non-regularity.}
\resizebox{\textwidth}{!} {\begin{tabular}{lccc|ccc}
  \hline
           & \multicolumn{3}{c}{$\beta^*_{1,1}$} &\multicolumn{3}{c}{$\beta^*_{1,2}$} \\
Models & Q$_{m,N}^{\kappa=0.05}$ & Proposed & dWOLS$_{m,N}^{\kappa=0.05}$ & Q$_{m,N}^{\kappa=0.05}$ & Proposed & dWOLS$_{m,N}^{\kappa=0.05}$ \\
  \hline
\multicolumn{7}{c}{\it Randomized Treatment Assignment Model} \\ \hline
Linear$^{NR,0}$&  0.976(0.31)$^\dag$ & 0.956(0.42) & 0.969(0.48)$^\dag$ & 0.976(0.31)$^\dag$ & 0.959(0.40) & 0.964(0.48) \\ 
 Non-linear$^{NR,0}$&  0.988(1.14)$^\dag$ & 0.965(1.40)$^\dag$ & 0.964(1.63) & 0.981(0.53)$^\dag$ & 0.952(0.56) & 0.975(0.65)$^\dag$\\ 
Linear$^{NR,1}$&  0.965(0.51)$^\dag$ & 0.962(0.46) & 0.980(0.82)$^\dag$ & 0.960(0.51) & 0.966(0.45)$^\dag$ & 0.964(0.83) \\ 
 Non-linear$^{NR,1}$& 0.984(1.13)$^\dag$ & 0.966(1.41)$^\dag$ & 0.960(1.71) & 0.963(0.59) & 0.946(0.60) & 0.964(0.82)\\ \hline
\multicolumn{7}{c}{\it Linear Treatment Assignment Model} \\ \hline
Linear$^{NR,0}$& 0.961(0.32)  & 0.965(0.45)$^\dag$ & 0.971(0.53)$^\dag$ & 0.954(0.32) & 0.948(0.44) & 0.968(0.53)$^\dag$ \\ 
 Non-linear$^{NR,0}$& 0.521(1.15)$^\dag$ & 0.949(1.58) & 0.955(1.83) & 0.190(0.58)$^\dag$ & 0.955(0.65) & 0.975(0.78)$^\dag$\\ 
Linear$^{NR,1}$&  0.907(0.51)$^\dag$ & 0.953(0.51) & 0.968(0.89)$^\dag$ &  0.900(0.51)$^\dag$ & 0.955(0.50) & 0.965(0.89)$^\dag$ \\ 
 Non-linear$^{NR,1}$&  0.450(1.14)$^\dag$ & 0.952(1.61) & 0.957(1.93) & 0.163(0.63)$^\dag$ & 0.948(0.70) & 0.974(0.93)$^\dag$ \\  \hline
\multicolumn{7}{c}{\it InterQuad Treatment Assignment Model} \\ \hline
Linear$^{NR,0}$&  0.982(0.34)$^\dag$ & 0.966(0.46)$^\dag$ & 0.982(0.54)$^\dag$ & 0.975(0.34)$^\dag$ & 0.957(0.45) & 0.969(0.53)$^\dag$ \\
 Non-linear$^{NR,0}$& 0.918(1.21)$^\dag$ & 0.959(1.45) & 0.846(1.70)$^\dag$ & 0.950(0.61) & 0.962(0.63) & 0.866(0.74)$^\dag$\\ 
Linear$^{NR,1}$&  0.967(0.56)$^\dag$ & 0.964(0.52) & 0.972(0.91)$^\dag$ & 0.959(0.56) & 0.950(0.51) & 0.964(0.91)\\ 
 Non-linear$^{NR,1}$& 0.912(1.20)$^\dag$ & 0.964(1.47) & 0.871(1.79)$^\dag$ & 0.943(0.66) & 0.967(0.68)$^\dag$ & 0.910(0.92)$^\dag$\\ 
\hline
\end{tabular}}
{\small Numbers in parentheses correspond to average confidence interval length.}
\label{tab:nr}
\end{table}

\section{Application}

We use the data from the {\it Extending Treatment Effectiveness of Naltrexone (ExTENd)} clinical trial to illustrate our method. Naltrexone  (NTX) is an opioid receptor antagonist used in the prevention of relapse to alcoholism. Even though NTX has been shown to be efficacious in those that adhere to treatment, its use by clinicians has been limited, at least in some cases, because  adherence rates are often negatively impacted by the fact that NTX diminishes the pleasurable effects of alcohol use.

\begin{table}[t]
\centering
\caption{EXTEND data. List of baseline and time-varying covariates.}
\resizebox{\textwidth}{!}{\begin{tabular}{l lccc}
  \hline
Covariate & Description  \\
  \hline
gender &  binary variable coded 1 for female  \\
edu &  years of education    \\
race &  binary variable coded 1 for white and 0 otherwise  \\
alcyears &  years of lifetime alcohol use  \\
intox &  years of drinking to intoxication  \\
married &  marital status coded 1 for married and 0 otherwise  \\
ethnic &  binary variable coded 1 for non-hispanic and o for hispanic \\
ocds$_0$ &  obsessive-compulsive drinking scale (higher value means more severe craving) \\
pacs$_0$ &  Penn Alcohol Craving Scale (higher value means more severe craving) \\
A$_1$ &  stage 1  treatment option coded as 1 for lenient definition and 0 for  stringent\\
apc$_1$ &  average number of pills taken per day during stage 1 \\
pdhd$_1$ &  percent days heavy drinking during stage 1\\
pacs$_1$ &  Penn Alcohol Craving Scale (higher value means more severe craving) during stage 1 \\
mcs$_1$ &  mental composite score during stage 1 (higher value means better health condition) \\
\hline
\end{tabular}}
\label{tab:var}
\end{table}

In the ExTENd  study (Figure \ref{fig:1}), at the first decision stage, patients were randomized to one of two definitions of non-response while receiving NTX: 
(1) Stringent: a patient is a non-responder if (s)he has two or more heavy drinking days in the first 8 weeks ($A_1=0$); (2) Lenient: a patient is a non-responder 
if (s)he has five or more heavy drinking days in the first 8 weeks ($A_1=1$). At the second decision stage, the treatment assignment mechanism depends on response status.
Specifically, define $A_2 = 1$ if the current treatment (NTX) is augmented, and zero otherwise; in addition, we let
$\bar{R}$ denote the indicator of response to treatment. Then, among responders ($\bar{R}=1$), patients are randomized 
(with equal probability) to augment NTX with telephone disease management (NTX+TDM; $A_2 = 1$) or to maintain NTX alone ($A_2 = 0$). For 
non-responders ($\bar{R}=0$), patients are instead randomized (with equal probability) to augment NTX with combined behavioral intervention 
(NTX+CBI; $A_2 = 1$) or to CBI alone ($A_2 = 0$). In the latter case, maintenance on NTX alone is replaced with an alternative treatment
due to non-response.
The {primary outcome} is the proportion of abstinence days over 24 weeks. The list of baseline and time varying variables that are used in our analyses are given in Table \ref{tab:var}. There are multiple measurements of time-varying variables during the first stage. We denote the average of these variables as mcs$_1$, pacs$_1$, pdhd$_1$, and apc$_1$. 

Standardized differences in means for each covariate (i.e.,  differences in means divided by the corresponding pooled standard deviation) were used to check the covariate balance across the treatment groups. Figure \ref{fig:balance} indicates that there is a good balance of baseline covariates across the levels of $A_1$ (circle). However, we see some imbalance across the levels of second stage treatment options. This is more evident in the non-responder group (triangle). Absolute standardized differences exceeding 0.1 or 0.2 are respectively referred to as mild and substantial imbalance, and can potentially induce bias in the evaluation of effect modifiers 
if not taken into account \citep{austin2009using}. In this figure, sdApc$_1$ and sdPdhd$_1$ respectively represent the standard deviation of the 
indicated variables during the first stage.

\begin{figure}[t] 
\centering 
\includegraphics[scale=0.45]{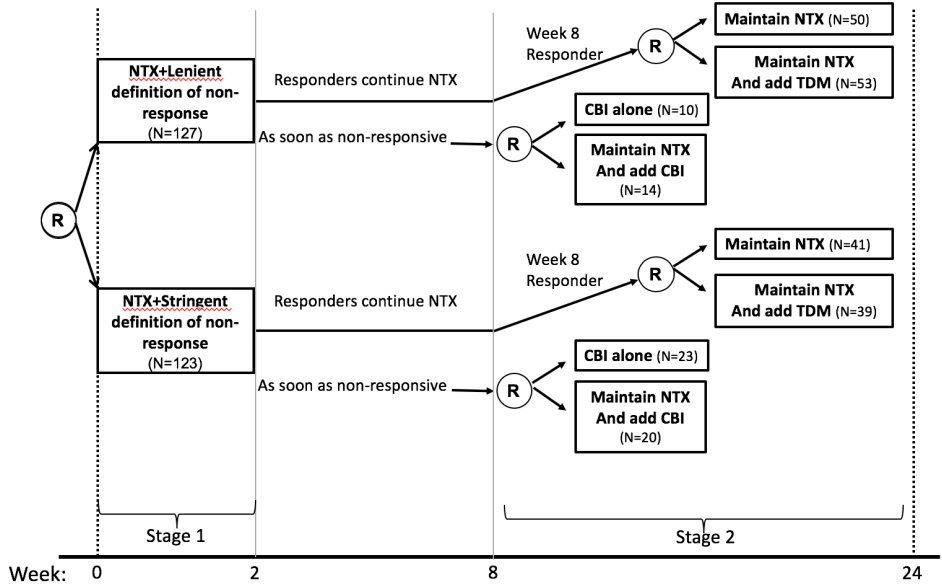}
\caption{\footnotesize ExTENd study design. The {\sffamily\textregistered}  notation represents instances of randomization; the $N$ values 
in this figure represent the subsequent number of patients assigned to each treatment option. } \label{fig:1}
\end{figure}

\begin{figure}[t] 
\centering 
\includegraphics[scale=0.55]{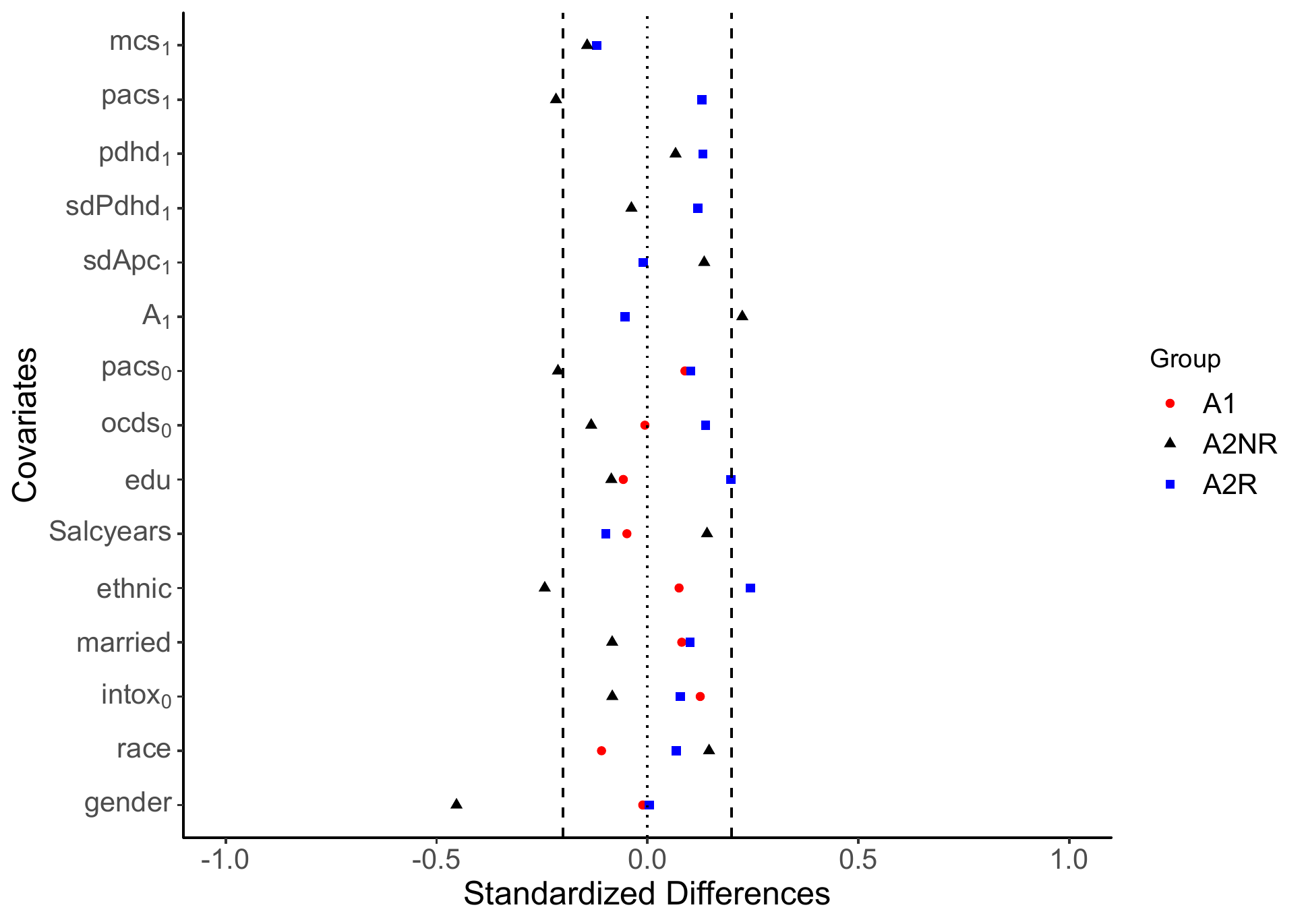}
\caption{\footnotesize ExTENd study. Covariate imbalance across different treatment groups. $A_1$: stage 1 treatment option; $A_{2NR}$: 
stage 2 treatment options among non-responders; $A_{2R}$: stage 2 treatment options among responders.  
The dashed vertical lines show cut points at $\pm$ 0.2. } \label{fig:balance}
\end{figure}

We analyzed the data using the proposed method, dWLOS$_{m,N}^\kappa$ and  Q-learning (Q$_{m,N}$) approaches; 
the results are summarized in Table \ref{tab:extend}. The latter two methods use the $N$-out-of-$N$ bootstrap
for calculating standard errors in the second stage model and the $m$-out-of-$N$ bootstrap
for calculating standard errors in the first stage model.
Referring to earlier notation, the first stage covariate vector $\bm{W}^0$ consists of the predictors
gender, race, alcyr$_0,$ intox$_0,$ and ocds$_0$, and the second stage covariate vector $\bm{S}^0$ consists
of all the predictors listed in Table \ref{tab:var}, along with response status.
First stage regression models are fit using $\bm{W} = \bm{W}^0.$ The description of the second stage model
predictor $\bm{S}$ is more involved.  
Specifically, let $\bm{Z} = (Z_1, \ldots, Z_6)^\top$ contain the variables  gender,  A$_1,$ intox$_0$, ocds$_0$, pacs$_1$, and mcs$_1;$
then, we define $\bm{S} = (\bar{R}, 1-\bar{R}, \bar{R} Z_1, \bar{R} Z_2, (1-\bar{R}) \bm{Z}^\top)^\top$. 
As specified, the second stage model allows the set of possible effect modifiers to differ between responders and non-responders,
with some overlap in the case of gender and A$_1.$
We used \texttt{SuperLearner} to estimate $\mu_{1Y}(\mathbf{w}^0)$ and $\mu_{2Y}(\mathbf{s}^0),$  employing the same library as 
we did in the simulation study and respectively using $\bm{W}^0$ and $\bm{S}^0$ for general confounding control.
In view of the fact that the randomization mechanism is known, and mostly successful in view of the overall degree of balance 
observed in Figure \ref{fig:balance}, the treatment propensities $\mu_{1A}(\mathbf{w}^0)$ and $\mu_{2A}(\mathbf{s}^0)$
are estimated using logistic regression models. Specifically, the former is estimated as a function of gender, and the 
latter is estimated using gender, response status, and the interaction between gender and response status.
The parameters of the Q-functions used by dWOLS are assumed to follow linear models (i.e., including the main effects). 
Similarly, for standard Q-learning, linear working models respectively replace $\eta_2(\bm{S}^0)$ and $\eta_1(\bm{W}^0)$.  

As shown in Table \ref{tab:extend}, the signs of all predictor effects are the same for all methods, though magnitudes and 
confidence intervals differ. None of the effect modifiers in the second stage are deemed statistically significant among 
responders using any of the 3 methods. The proposed Q-learning method suggests that both ocds$_0$ and mcs$_1$ 
are significant effect modifiers of $A_2$ among non-responders; specifically, individuals with higher ocsd$_0$ and mcs$_1$ 
would benefit from CBI. Similarly, dWOLS identifies mcs$_1$ as a significant effect modifier among non-responders, whereas 
none of the effect modifiers are identified as significant using standard Q-learning.
For the first stage model, the proposed Q-learning method shows that the years of drinking (i.e., alcyr$_0$) and gender 
significantly modify the effect of $A_1$. In particular,  female individuals and those with more years of drinking would benefit from 
being treated under a stringent definition of non-response. This makes sense because, for example, individuals with more years 
of drinking at baseline likely have a higher craving for alcohol and require more immediate attention and rescue treatments 
(i.e., $A_2$ for non-responders). 
In contrast, neither standard Q-learning nor dWLOS detects any effect modifiers.  With the exception of
the interaction between $A_1$ and gender, the first stage point estimates are rather similar across the 3 methods,
highlighting the fact that the differences in significance stem from the tighter confidence intervals obtained
using the proposed methods (i.e., compared to those produced using the $m$-out-of-$N$ bootstrap).   
 
Indeed, for both stages, the dWLOS and standard Q-learning methods yield point estimates that are mostly similar 
to each other.
These similarities are expected for two reasons. First, under successful randomization, we would not generally expect misspecification 
of the functional form of the main effects in the Q-function to bias the estimate of interaction terms (i.e., $\bm{\beta}^*_2$). Second, the 
linear models being used in the Q-models are identical in both cases; the only difference is that parameter estimation is carried out using 
weighted versus unweighted least squares.   
Figure \ref{fig:balance} demonstrates the presence of random confounding among non-responders 
in the second stage for gender, ethnicity, $A_1,$ pacs$_0$ and pacs$_1,$ whereas there is good balance 
among the first stage predictors. Comparing dWLOS and standard Q-learning, we see that the largest differences 
in point estimates occur in the second stage model among non-responders for $A_1$ and pacs$_1.$

Comparing the proposed method to both dWOLS and standard Q-learning, we observe somewhat greater 
disparity in point estimates. These differences occur primarily among non-responders in the second stage,
and include interactions between $A_2$ and each of gender,  $A_1$ and pacs$_1;$ as noted above, 
gender,  $A_1$ and pacs$_1$ all demonstrate substantial imbalance among non-responders in the second 
stage. The largest difference among regression coefficients in the first stage model occurs for gender, consistent 
with the disparities observed in the second stage model as well as propagation of those differences through 
the backward induction process.  We conjecture that modeling the 
true main effects (i.e., $\mu_{jY}(\cdot)$, $j=1,2$) using Super Learner may help to reduce small sample biases 
when compared to the more restrictive linear models used by both dWOLS and standard Q-learning.

\begin{table}[t]
\centering
\caption{ExTENd data. The $^\dagger$ indicates significant coefficients at a Type I error rate of 5\%. CI represents the confidence interval.}
{{\begin{tabular}{l ccccccc}
  \hline
 Q-function & \multicolumn{2}{c}{Proposed} & \multicolumn{2}{c}{dWLOS$_{m,N}^{\kappa=0.05}$} & \multicolumn{2}{c}{Q$_{m,N}^{\kappa=0.05}$} \\
Models & Est & 95\% CI & Est & 95\% CI & Est & 95\% CI\\
  \hline
\multicolumn{2}{l}{Stage 2} \\
\multicolumn{2}{l}{Responders} \\
$A_2$ &  0.02 & (-0.07,0.11)& 0.01 & (-0.07,0.09)&0.01 & (-0.07,0.09) \\
$A_2:gender$ &  0.07 & (-0.11,0.24)  &0.11 & (-0.08,0.29)& 0.11 & (-0.08,0.28)  \\
$A_2:A_1$ &  0.01 & (-0.11,0.13) & 0.01 & (-0.11,0.12) &0.01 & (-0.11,0.12)\\
\multicolumn{2}{l}{Non-responders} \\
$A_2$ &  -0.07 & (-0.23,0.10) & -0.02 & (-0.28,0.22) &-0.02 & (-0.28,0.23)\\
$A_2:gender$ & 0.27 & (-0.18,0.71)  &0.13 & (-0.43,0.81)&0.14 & (-0.43,0.79) \\
$A_2:intox_0$ &  0.09 & (-0.16,0.35)  &  0.06 & (-0.21,0.46) & 0.06 & (-0.22,0.49)  \\
$A_2:ocds_0$ &  -0.19$^\dagger$ & (-0.34,-0.03)  & -0.19 & (-0.43,0.08) &-0.18 & (-0.43,0.11) \\
$A_2:A_1$ &  -0.21  & (-0.45,0.01)& -0.23& (-0.57,0.26)& -0.16 & (-0.55,0.30) \\
$A_2:pacs_1$ &  0.10& (-0.02,0.21)& 0.07 & (-0.12,0.25)&0.03 & (-0.18,0.23)\\
$A_2:mcs_1$ &  -0.18 $^\dagger$ & (-0.29,-0.06) &-0.17 $^\dagger$& (-0.37,0.00)&-0.17 & (-0.39,0.01) \\
\\ \hline
\multicolumn{2}{l}{Stage 1} \\ 
$A_1$&  -0.06 & (-0.18,0.06) & -0.05 & (-0.21,0.08)&-0.06 & (-0.20,0.09)   \\
$A_1:gender$&  -0.24$^\dagger$ & (-0.46,-0.01)  & -0.18 & (-0.43,0.14) &-0.17 & (-0.37,0.14) \\
$A_1:race$&  0.08 & (-0.05,0.21)  & 0.08 & (-0.07,0.21)&0.08 & (-0.08,0.24) \\
$A_1:alcyr_0$& -0.07$^\dagger$  & (-0.14,0.00)  & -0.06 & (-0.15,0.03)&-0.06 & (-0.16,0.03) \\
$A_1:intox_0$&  0.13 & (-0.03,0.29)& 0.14 & (-0.06,0.33)&0.15 & (-0.10,0.33) \\
$A_1:ocds_0$&  0.02 & (-0.05,0.09)& 0.01 & (-0.07,0.09)&0.01 & (-0.08,0.08)\\
\hline
\end{tabular}
\label{tab:extend}}} 
\end{table}

\section{Discussion}

Much of the current work on Q-learning continues to involve parametric working models despite the fact that
finite-dimensional models are generally too restrictive to permit consistent estimation of nuisance parameters. 
We proposed a robust Q-learning approach where the working models need not all be linear and, specifically, 
where the main effects that do not influence the optimal decision rules are estimated using data-adaptive approaches.
Our simulation studies highlight the value of our proposed approach compared with existing Q-learning methods. The 
proposed method also performed relatively well in simulations when key regularity assumption 
(i.e., Assumption \ref{assump:unique}) is violated; however, we cannot expect this in all scenarios, as the underlying theory 
and simulation results show  otherwise.

An important advantage of the proposed method is that it does not suffer from the curse of dimensionality, 
as it produces root-$n$ consistent estimators even when $\hat{\mu}_{jA}(\cdot)$ and $\hat{\mu}_{jY}(\cdot)$ $(j=1,2)$
are estimated at  rates slower than root-$n$. This important property facilitates the use of nonparametric methods
like Super Learner for estimating these unknown functions, substantially reducing the chance of model misspecification. 
A second important feature of the proposed approach is that consistent estimation of the treatment models leads 
to consistent estimation of the blip function parameters, whether or not these models or those for $\mu_{jY}(\cdot), j=1,2$ 
are correctly specified. However, the proposed estimators are not doubly robust, in that we require that the 
$\mu_{jA}(\cdot)$s are consistently estimated at a sufficiently fast rate. True double robustness 
under \eqref{eq:causal-model} for $\bbeta^*_2$ requires that one either consistently estimates $\mu_{2A}(\bm{S}^0)$ or the treatment-free conditional 
mean model $E[ Y - A_2 \bm{S}^\top \bbeta^*_2 | \bm{S}^0] = \eta_2(\bm{S}^0);$ similarly, for $\bbeta^*_1,$
one must either consistently estimate $\mu_{1A}(\bm{W}^0)$ or $E[\tilde Y - A_1 \bm{W}^\top \bbeta^*_1 | \bm{W}^0] = \eta_1(\bm{W}^0).$
Because the expectation operator is linear, correct specification of both treatment-free models essentially relies on
both $\mu_{jY}(\cdot)$ and $\mu_{jA}(\cdot), j=1,2$ being correctly specified. This limitation on the practicality of finding
a truly doubly robust estimator applies  to the proposed approach as well as that taken in \cite{wallace2015doubly}. 
Further research on doubly robust estimation in this class of problems is merited.

Although data-adaptive estimation methods reduce the risk of inconsistency, there is still a chance that one or more nuisance parameters will be
estimated inconsistently. Further research is needed to study the behavior of the proposed methods under inconsistent estimation of a nuisance parameter. 
In particular, \cite{benkeser2017doubly} showed that when nuisance parameters are estimated using data-adaptive approaches, inconsistently estimating 
one nuisance parameter may lead to an irregular estimator having a convergence rate slower than root-$n$. These authors proposed a targeted minimum 
loss-based approach to resolve the issue \citep{van2014targeted2}. Generalization of the method of \cite{benkeser2017doubly} to a multi-stage  decision making 
process would be an interesting topic for future research. Studying the asymptotic behavior of an appropriate version of the bootstrap in our proposed Q-learning 
method is also of interest as it can potentially resolve the non-regularity issues in settings where both Assumption \ref{assump:unique} fails and  Corollary \ref{cor1}  fail \citep[e.g.][]{chakraborty2013inference}. Finally, in practice, there are  often many candidate variables to be considered when constructing a decision rule. The inclusion of spurious variables in these analyses can substantially reduce the quality of the estimated decision rules. Although one can adapt the proposed methods to obtain regularized estimators of the target parameters, valid post-selection inference remains a challenge and merits further research \citep{berk2013valid, fithian2014optimal}.

\bibliographystyle{Biometrika}
\bibliography{DQL-bib}
\pagebreak

\setcounter{table}{0}
\renewcommand{\thetable}{S\arabic{table}}

\setcounter{figure}{0}
\renewcommand{\thefigure}{S\arabic{figure}}

\section*{Supplementary material to Robust Q-learning}

Let $Z \sim P$ for some probability measure $P$ and suppose $f(\cdot)$ is any real-valued, $P-$measurable function; then, we  
define the $L^2(P)$ norm of $f(\cdot)$ as $\| f(Z) \|_{P,2} = \{ \int \! f(\omega)^2 d P(\omega) \}^{1/2}.$ 
In addition, let $\| \bm{x} \|_q$ denote the usual $q-$ norm of a vector $\bm{x}$  for $q=1,2,\infty$.
The following general lemmas will be helpful  in our proofs.

\begin{lemma}\label{lem:condcvg}
Let $A_n$ and $B_n$ be sequences of random vectors, $n \geq 1$.  Let $\epsilon > 0$ be arbitrary and, for any vector norm, 
suppose that $\lim_{n \rightarrow \infty} P( \| A_n \| > \epsilon | B_n)  = 0.$ Then, 
$\lim_{n \rightarrow \infty} P( \| A_n \| > \epsilon)  = 0.$ By Chebyshev's inequality, a sufficient condition for proving that $\lim_{n \rightarrow \infty} P( \| A_n \| > \epsilon)  = 0$ 
is that $\lim_{n \rightarrow \infty} E( \| A_n \|^q |  B_n)  = 0$ from some $q \geq 1.$ 
\end{lemma}
The above lemma essentially repeats Lemma 6.1 in \cite{Cherno18}
and will not be proved here.  The following lemma is a direct consequence of a well-known result and also has an easy proof; see, for example, 
\cite{Stewart69}.
\begin{lemma}\label{lem:matinvcvg}
Let $\bm{M}_n$ and $\hat {\bm{M}}_n$ be two sequences of square matrices and let $\| \cdot \|$ be any proper matrix norm. 
Suppose  there exists $n_0 < \infty$ such that  (i) $\bm{M}^{-1}_n$ and $\bm{M}_n$ exist for $n \geq n_0$, with 
 $0 < C_1 \leq \| \bm{M}^{-1}_n \| \leq C_2 < \infty;$ 
and, (ii) $\| \hat {\bm{M}}_n -  \bm{M}_n \| \leq (2 \| \bm{M}^{-1}_n \|)^{-1}$.
Then, 
\[
\| \hat {\bm{M}}^{-1}_n -  \bm{M}^{-1}_n \| \leq 2  C^2_2 \| \hat {\bm{M}}_n -  \bm{M}_n \|.
\]
\end{lemma}

We will also have need of the following lemma.
\begin{lemma}\label{lem:help1cvg}
Let $\bm{B}_1,\ldots,\bm{B}_N$ be independent, identically distributed vectors from $P_0,$ where $\bm{B}_i \in {\mathcal B} \subset \R^d$. 
Let $\bm I_n$ be a randomly chosen subset of the integers $1,\ldots,N$ of length $n = O(N)$ and let its complement $\bm I^c_n$ have $N-n = O(N)$ elements. 
Let $\bm{F}_{I_n}$ and $\bm{F}_{I^c_n}$ be the corresponding disjoint subsets of $\bm{B}_1,\ldots,\bm{B}_{N}.$ Let $\gamma_j: {\mathcal B} \rightarrow \R, j=1,2$ and let $\hat \gamma_j(\cdot; \bm{F}_{I^c_n})$ be an estimator of $\gamma_j(\cdot)$ derived from the data $\bm{F}_{I^c_n}.$ 
Finally, define 
\begin{equation}
\label{Gfun}
\bm{G}_{n,N}  =  \frac{1}{n} \sum_{i \in \bm{I}_n} \left\{ \hat \gamma_1(\bm{B}_i; \bm{F}_{I^c_n}) - \gamma_1(\bm{B}_i) \right\}
\left\{ \hat \gamma_2(\bm{B}_i; \bm{F}_{I^c_n}) - \gamma_2(\bm{B}_i) \right\} \bm{h}( \bm{B}_i ) 
\end{equation}
where $\bm{h}(\bm{B}_i )$ is any finite dimensional vector- or matrix-valued function of $\bm{B}_i$ such that $P(\| \bm{h}(\bm{B}_i ) \|_{\infty} \leq C) = 1$
for some $C < \infty$. Then,
\begin{equation}
\label{mybound}
\| \bm{G}_{n,N} \|_{\infty} \leq C \left\|  \hat \gamma_1(\bm{B}; \bm{F}_{I^c_n}) - \gamma_1(\bm{B}) \right\|_{\mathbb{P}_n,2}
\left\|  \hat \gamma_2(\bm{B}; \bm{F}_{I^c_n}) - \gamma_2(\bm{B}) \right\|_{\mathbb{P}_n,2},
\end{equation}
where $\mathbb{P}_n$ is the empirical measure on $\bm{F}_{I_n}$. Moreover, for $j=1,2$
 define 
\begin{equation}
\label{normident}
\left\|  \hat \gamma_j(\bm{B}; \bm{F}_{I^c_n}) - \gamma_j(\bm{B}) \right\|^2_{P_0,2} 
= \E \left\{ \left\|  \hat \gamma_j(\bm{B}; \bm{F}_{I^c_n}) - \gamma_j(\bm{B}) \right\|^2_{\mathbb{P}_n,2} \big| \bm{F}_{I^c_n} \right\}
\end{equation}
and suppose \eqref{normident} 
is $o_p(N^{-a_j}),$ where $a_j \geq 0.$ Then,
$\| \bm{G}_{n,N} \|_{\infty} = o_p(N^{-(a_1+a_2)/2})$.
\end{lemma}

\begin{proof}
Let $r_{ij}(\bm{F}_{I^c_n}) = \hat \gamma_j(\bm{B}_i; \bm{F}_{I^c_n}) - \gamma_j(\bm{B}_i)$ for $i=1,\ldots,n$ and
$j=1,2$. Under the assumption that
$P(\| \bm{h}(\bm{B}_i ) \|_{\infty} \leq C) = 1,$ 
the triangle and Cauchy-Schwarz equalities imply
\begin{eqnarray*}
\| \bm{G}_{n,N} \|_{\infty} & \leq & \frac{1}{n} \sum_{i \in \bm{I}_n}
\left| r_{i1}(\bm{F}_{I^c_n}) r_{i2}(\bm{F}_{I^c_n})  \right|
 \| \bm{h}( \bm{B}_i ) \|_{\infty} \\
 & \leq & C
\left[ \frac{1}{n} \sum_{i \in \bm{I}_n}  \left\{ r_{i1}(\bm{F}_{I^c_n}) \right\}^2 \right]^{1/2}
\left[ \frac{1}{n} \sum_{i \in \bm{I}_n}  \left\{ r_{i2}(\bm{F}_{I^c_n}) \right\}^2 \right]^{1/2},
 \end{eqnarray*}
the representation \eqref{mybound} now following immediately from the definition of the $L^2(Q)$ norm given earlier 
using $Q = \mathbb{P}_n$. 
To establish that $\| \bm{G}_{n,N} \|_{\infty} = o_p(N^{-(a_1+a_2)}),$ 
we first use Markov's inequality: for any $\epsilon > 0,$
\[
P \left( \| \bm{G}_{n,N} \|_{\infty} > \epsilon | \bm{F}_{I^c_n} \right) \leq \epsilon^{-1} \E \left( \| \bm{G}_{n,N} \|_{\infty}  | \bm{F}_{I^c_n} \right).
\]
Using \eqref{mybound} and the Cauchy-Schwarz inequality again, it follows that
\begin{eqnarray*}
P \left( \| \bm{G}_{n,N} \|_{\infty} > \epsilon | \bm{F}_{I^c_n} \right) & \leq & \frac{C}{\epsilon}
\E \left\{ \prod_{j=1}^2 \left\|  \hat \gamma_j(\bm{B}; \bm{F}_{I^c_n}) - \gamma_j(\bm{B}) \right\|_{\mathbb{P}_n,2} \big| \bm{F}_{I^c_n} \right\} \\
& \leq & \frac{C}{\epsilon} \prod_{j=1}^2
\E \left\{  \left\|  \hat \gamma_j(\bm{B}; \bm{F}_{I^c_n}) - \gamma_j(\bm{B}) \right\|^2_{\mathbb{P}_n,2} \big| \bm{F}_{I^c_n} \right\}^{1/2} \\
& = & \frac{C}{\epsilon} \prod_{j=1}^2 \left\|  \hat \gamma_j(\bm{B}; \bm{F}_{I^c_n}) - \gamma_j(\bm{B}) \right\|_{P_0,2}
\end{eqnarray*}
the last result following directly from  \eqref{normident}.
Under the stated assumptions, the right-hand side is now seen to be
$o_p(N^{-(a_1+a_2)/2}),$ as desired. 
\end{proof}
The statement and proof of Lemma \ref{lem:help1cvg} employs a simple form of sample splitting in which the unknown function $\gamma_j(\cdot)$ is
estimated by $\hat \gamma_j(\cdot; \bm{F}_{I^c_n})$ from a sample $\bm{F}_{I^c_n}$ that is independent of the
$\bm{B}_i$s (i.e., $\bm{F}_{I_n}$) appearing in the calculation of  \eqref{Gfun}. Importantly, Lemma
\ref{lem:help1cvg} does not preclude the possibility that $\gamma_1(\cdot) = \gamma_2(\cdot)$ and
$\hat \gamma_1(\cdot) = \hat \gamma_2(\cdot);$ in this case, \eqref{Gfun} reduces to
\begin{equation}
\label{Gfun2}
\bm{L}_{n,m}  =  \frac{1}{n} \sum_{i \in \bm{I}_n} \left\{ \hat \gamma(\bm{B}_i; \bm{F}_{I^c_n}) - \gamma(\bm{B}_i) \right\}^2 \bm{h}( \bm{B}_i ) 
\end{equation}
and $\| \bm{G}_{n,N} \|_{\infty} = o_p(N^{-a})$ for some $a \geq 0 $ provided that
 \begin{equation}
\label{normident4}
\left\|  \hat \gamma(\bm{B}; \bm{F}_{I^c_n}) - \gamma(\bm{B}) \right\|^2_{P_0,2} = o_p(N^{-a}).
\end{equation}
A related lemma now follows.

\begin{lemma}\label{lem:help2cvg}
Let $(R_1,\bm{B}_1),\ldots, (R_n,\bm{B}_N)$ be independent, identically distributed vectors from $P_0,$ where $\bm{B}_i \in {\mathcal B} \subset \R^d$
and $E_i \in \R$. Suppose $\E( R_i | \bm{B}_i ) = 0$ and  $\mbox{var}( R_i | \bm{B}_i ) = \theta^2_i \in (0,C_1]$ for $i = 1 \ldots N$ and
a constant $C_1 < \infty.$
Let $\bm I_n$ be a randomly chosen subset of the integers $1,\ldots,N$ of length $n = O(N)$ and let its complement 
$\bm I^c_n$ have $N-n$ elements.
Let $\bm{F}_{I_n}$ and $\bm{F}_{I^c_n}$ be the corresponding disjoint subsets of $(R_1,\bm{B}_1),\ldots,(R_N,\bm{B}_{N}).$ 
Let $\gamma: {\mathcal B} \rightarrow \R$ and let $\hat \gamma(\cdot; \bm{F}_{I^c_n})$ be an estimator of $\gamma(\cdot)$ 
derived from the data $\bm{F}_{I^c_n}.$  Finally, define 
\begin{equation}
\label{Lfun}
\bm{L}_{n,N}  =  \frac{1}{n} \sum_{i \in \bm{I}_n} R_i 
\left\{ \hat \gamma(\bm{B}_i; \bm{F}_{I^c_n}) - \gamma(\bm{B}_i) \right\} \bm{h}( \bm{B}_i ) 
\end{equation}
where $\bm{h}(\bm{B}_i )$ is any finite dimensional vector-valued function of $\bm{B}_i$ 
such that $P(\| \bm{h}(\bm{B}_i ) \|_{2} \leq C_2) = 1$ for $C_2 < \infty$. 
Suppose that
\begin{equation}
\label{normident3}
\left\|  \hat \gamma(\bm{B}; \bm{F}_{I^c_n}) - \gamma(\bm{B}) \right\|^2_{P_0,2} 
= \E \left\{ \left\|  \hat \gamma(\bm{B}; \bm{F}_{I^c_n}) - \gamma(\bm{B}) \right\|^2_{\mathbb{P}_n,2} \big| \bm{F}_{I^c_n} \right\}
\end{equation}
is $o_p(N^{-a})$ where $a \geq 0.$ Then,
$\| \bm{L}_{n,N} \|_{\infty} = o_p(N^{-(1+a)/2}).$
\end{lemma}

\begin{proof}
The proof  relies on a variant of Chebyshev's inequality. Let
\[
\bm{L}_{n,N,j} = \frac{1}{n} \sum_{i \in \bm{I}_n} R_i 
\left\{ \hat \gamma(\bm{B}_i; \bm{F}_{I^c_n}) - \gamma(\bm{B}_i) \right\} \bm{h}_j( \bm{B}_i )
\]
be the $j^{th}$ element of $\bm{L}_{n,N}.$  Let ${\mathcal B}_n = \{ \bm I_n,  ( \bm{B}_k, k \in \bm I_n) \}.$ Then,
it is easy to show that
\[
\E\left( \bm{L}_{n,N,j} \big| \bm{F}_{I^c_n} \right) 
~=~
\E \left\{ 
\E\left( \bm{L}_{n,N,j} \big| \bm{F}_{I^c_n},  {\mathcal B}_n \right) 
\big| \bm{F}_{I^c_n} \right\}
~=~0;
\]
this follows from calculating the inner expectation on the right-hand-side and using
the assumption that $\E( R_i | \bm{B}_i ) = 0$ for every $i.$ 
Using a similar conditioning argument, 
\begin{eqnarray*}
\mbox{var}\left( \bm{L}_{n,N,j} \big| \bm{F}_{I^c_n}\right) & = & 
\E\left\{ \mbox{var}\left( \bm{L}_{n,N,j} \big| \bm{F}_{I^c_n},  {\mathcal B}_n \right) \big| \bm{F}_{I^c_n} \right\}.
\end{eqnarray*}
Straightforward calculations now show
\begin{eqnarray*}
\mbox{var}\left( \bm{L}_{n,N,j} \big| \bm{F}_{I^c_n},  {\mathcal B}_n \right)
& = & 
 \frac{1}{n^2} \sum_{i \in \bm{I}_n}
\left\{ \hat \gamma(\bm{B}_i; \bm{F}_{I^c_n}) - \gamma(\bm{B}_i) \right\} ^2
\bm{h}^2_j( \bm{B}_i ) \mbox{var}\left(R_i | \bm{B}_i \right),
\end{eqnarray*}
implying that
\begin{eqnarray*}
\mbox{var}\left( \bm{L}_{n,N,j} \big| \bm{F}_{I^c_n} \right) 
& = & 
\E \left[ \frac{1}{n^2} \sum_{i \in \bm{I}_n}
\left\{ \hat \gamma(\bm{B}_i; \bm{F}_{I^c_n}) - \gamma(\bm{B}_i) \right\}^2 
\bm{h}^2_j( \bm{B}_i ) \theta^2_i \Big| \bm{F}_{I^c_n} \right] \\
& \leq &
\frac{C_1 C_2}{n} \E \left[ \frac{1}{n} \sum_{i \in \bm{I}_n}
\left\{ \hat \gamma(\bm{B}_i; \bm{F}_{I^c_n}) - \gamma(\bm{B}_i) \right\}^2 \Big| \bm{F}_{I^c_n}
 \right] \\
 & = & 
 \frac{C_1 C_2}{n} 
 \left\|  \hat \gamma(\bm{B}; \bm{F}_{I^c_n}) - \gamma(\bm{B}) \right\|^2_{P_0,2} \\
 & = & o_p(N^{-(1+a)}),
\end{eqnarray*}
the last step following from the assumptions on \eqref{normident3} made in the statement of
the lemma and the fact that $n = O(N)$.  Using a vector form of Chebyshev's inequality,
it can then be shown that $\| \bm{L}_{n,N} \|_2 = o_p(N^{-(1+a)/2});$ since
$\| \bm{L}_{n,N} \|_{\infty} \leq \| \bm{L}_{n,N} \|_2$, the stated result follows.
\end{proof}

\section{Proof of Theorem 1}

To review our main assumptions, we assume that we observe $N$ independently identically distributed trajectories of 
$(\bm{X}_1,A_1,\bm{X}_2,A_2,Y) \sim P_0$. The vector $\bm{X}_1 \in \mathcal{X}_1 \subset \R^{p_1}$ 
consists of baseline covariates measured before treatment at the first decision point $A_1\in \{0,1\}$ and the vector 
$\bm{X}_2 \in \mathcal{X}_2 \subset \R^{p_2}$ consists of intermediate covariates measured before treatment at the second decision point $A_2\in \{0,1\}$. 
For notational convenience we define  $\bm{S}^0_i=(\bm{X}_{1i}^\top,A_{1i},\bm{X}_{2i}^\top)^\top\in \mathcal{S} \subset \R^{p_1+p_2+1}$ and 
 $\bm{W}^0_i=\bm{X}_{1i} \in \mathcal{X}_1 \subset \R^{p_1}$. 
We will also have need to define the variables $\bm{S}_i$ and $\bm{W}_i, i=1,\ldots,N;$ respectively,
each represents some finite dimensional function of the variables in $\bm{S}^0_i$ and  $\bm{W}^0_i.$
We note that knowledge of $\bm{S}^0_i$ and $\bm{W}^0_i$ respectively implies knowledge of $\bm{S}_i$ and $\bm{W}_i;$
however, the reverse may not hold.
The observed outcome $Y \in \R$ (measured after $A_2$) is assumed continuous, with a larger value of $Y$ indicating  a better clinical outcome.

The developments below assume that the original sample, with elements independently and identically distributed
as $P_0,$ has been split into two independent samples, say
$\bm{D}_{I_n}$ and $\bm{D}_{I^c_n},$ being respectively of sizes $n = O(N)$ and $N-n = O(N)$. 
The nuisance parameters
$\hat \mu_{2Y}(\cdot),$ $\hat \mu_{2A}(\cdot),$ $\hat \mu_{1Y}(\cdot),$ and $\hat \mu_{1A}(\cdot)$
are estimated using the data in $\bm{D}_{I^c_n};$
the finite dimensional parameters of interest are then estimated using the data $\bm{D}_{I_n},$
treating $\hat \mu_{2Y}(\cdot),$ $\hat \mu_{2A}(\cdot),$ $\hat \mu_{1Y}(\cdot),$ and $\hat \mu_{1A}(\cdot)$
as if they were known functions. 
As developed here, our use of sample-splitting is a simple form of cross-fitting and can be generalized easily 
to make better use of the full sample
\citep{Cherno18}; the simpler form used here suffices to establish
the main ideas of the proofs. Lemmas \ref{lem:help1cvg} and \ref{lem:help2cvg} play
an important role in several of the proofs; since $n = O(N),$ statements
of the form $o_p(N^{-a})$ and $o_p(n^{-1})$ are equivalent, we 
use the latter to emphasize that the technical arguments rely on 
sample splitting, where a sample $\bm{D}_{I_n}$ of
size $n$ is used to estimate the finite dimensional parameters
of interest.

To simplify notation, where needed all calculations implicitly condition
on the set of selected indices $\bm{I}_n$. Using notation from the main paper, let 
$\Delta_{2i} = \Delta_2(\bm{S}^0_{i}),$ $\Delta_{1i} = \Delta_1(\bm{W}^0_{i}),$ 
$\mu_{2Ai} = \mu_{2A}(\bm{S}^0_i),$ $\hat \mu_{2Ai} = \hat \mu_{2A}(\bm{S}^0_i),$ $ \mu_{1Ai} = \mu_{1A}(\bm{W}^0_i),$ 
and $\hat \mu_{1Ai} = \hat \mu_{1A}(\bm{W}^0_i).$  In addition, as in the main paper, we define the matrices
\begin{equation*}
\bV_{2n} = \frac{1}{n} \sum_{i \in \bm{I}_n}   (A_{2i}- \mu_{2Ai})^2  \bm{S}^{\otimes 2}_i  
~\mbox{ and }~
\hat \bV_{2n} = \frac{1}{n} \sum_{i \in \bm{I}_n}   (A_{2i}- \hat \mu_{2Ai})^2  \bm{S}^{\otimes 2}_i   
\end{equation*}
\begin{equation*}
\bV_{1n} = \frac{1}{n} \sum_{i \in \bm{I}_n}   (A_{1i}- \hat \mu_{1Ai})^2  \bm{W}^{\otimes 2}_i  
~\mbox{ and }~
\hat \bV_{1n} = \frac{1}{n} \sum_{i \in \bm{I}_n}   (A_{1i}- \hat \mu_{1Ai})^2  \bm{W}^{\otimes 2}_i  
\end{equation*}
where $\bm{x}^{\otimes 2} = \bm{x}  \bm{x}^\top$ for any vector $\bm{x}$.

We make the following assumptions.

\begin{assumption} \label{assump:support-X}
(i)   The support of $\bm{W}^0$ and the conditional treatment effect $\Delta_1(\bm{W}^0)$ are uniformly bounded; 
(ii)  the support of $\bm{S}^0$ and the conditional treatment effect $\Delta_2(\bm{S}^0)$ are uniformly bounded; and,
the supports of $\bm{S}$ and $\bm{W}$ are uniformly bounded.
\end{assumption}

\begin{assumption} \label{assump:accuracy-treatment}
  (i) $\left\|  \hat \mu_{1A}(\bm{W}^0; \bm{D}_{I^c_n}) - \mu_{1A}(\bm{W}^0) \right\|^2_{P_0,2} = o_p(n^{-1/2});$ 
  (ii) $\left\|  \hat \mu_{2A}(\bm{S}^0; \bm{D}_{I^c_n}) - \mu_{2A}(\bm{S}^0) \right\|^2_{P_0,2} = o_p(n^{-1/2}).$
\end{assumption}

\begin{assumption}  \label{assump:accuracy-outcome}
  (i) $\left\|  \hat \mu_{1Y}(\bm{W}^0; \bm{D}_{I^c_n}) - \mu_{1Y}(\bm{W}^0) \right\|^2_{P_0,2} = o_p(1);$ 
  (ii) $\left\|  \hat \mu_{2Y}(\bm{S}^0; \bm{D}_{I^c_n}) - \mu_{2Y}(\bm{S}^0) \right\|^2_{P_0,2} = o_p(1).$
\end{assumption}

\begin{assumption}  \label{assump:accuracy-2} 
(i) $\left\|  \hat \mu_{1Y}(\bm{W}^0; \bm{D}_{I^c_n}) - \mu_{1Y}(\bm{W}^0) \right\|_{P_0,2}
\left\|  \hat \mu_{1A}(\bm{W}^0; \bm{D}_{I^c_n}) - \mu_{1A}(\bm{W}^0) \right\|_{P_0,2} = o_p(n^{-1/2});$\\
(ii) $\left\|  \hat \mu_{2Y}(\bm{S}^0; \bm{D}_{I^c_n}) - \mu_{2Y}(\bm{S}^0) \right\|_{P_0,2}
\left\|  \hat \mu_{2A}(\bm{S}^0; \bm{D}_{I^c_n}) - \mu_{2A}(\bm{S}^0) \right\|_{P_0,2} = o_p(n^{-1/2})$
\end{assumption}

\begin{assumption} 
\label{assump:posdef}
There exists $1 \leq n_0 < \infty$ such that $\bV_{jn}$ and $\hat \bV_{jn}, j = 1,2$ are positive definite for $n \geq n_0$.
\end{assumption}

\begin{assumption} 
\label{assump:unique} 
$P\big( | \bm{S}_1^\top  \bm{\beta}^*_{2} | = 0 \big) = 0.$
\end{assumption}

We prove this Theorem with help from the following lemma. 
\begin{lemma} \label{lem:beta1-fixed}
Suppose Assumptions \ref{assump:support-X}, \ref{assump:accuracy-treatment},  and
\ref{assump:posdef} hold. Let $d_2 = dim(\bm{S})$ and $d_1 = dim(\bm{W})$.
Then $\| \tilde{\bm{\beta}}^*_{2n}-{\bm{\beta}}^{*}_{2n}    \|_\infty=o_p(n^{-1/2})$ and $\| \tilde{\bm{\beta}}^*_{1n}-{\bm{\beta}}^{*}_{1n}    \|_\infty=o_p(n^{-1/2})$, where 
  \begin{align*}
\tilde{\bm{\beta}}^*_{2n}&= \argmin_{\bm{\beta}_{2} \in \R^{d_2}} \sum_{i \in \bm{I}_n}  \left\{ A_{2i}- \hat \mu_{2A}(\bm{S}^0_i)\right\}^2 \left\{  \Delta_2(\bm{S}^0_{i}) -   \bm{S}^\top_{i} \bm{\beta}_{2}  \right\}^2, \\
\bm{\beta}^*_{2n}&= \argmin_{\bm{\beta}_{2} \in \R^{d_2}} \sum_{i \in \bm{I}_n}  \left\{ A_{2i}- \mu_{2A}(\bm{S}^0_i)\right\}^2 \left\{  \Delta_2(\bm{S}^0_{i}) -   \bm{S}^\top_{i} \bm{\beta}_{2}  \right\}^2, \\
\tilde{\bm{\beta}}^*_{1n}&= \argmin_{\bm{\beta}_{1} \in \R^{d_1}} \sum_{i \in \bm{I}_n}  \left\{ A_{1i}- \hat \mu_{1A}(\bm{W}_i)\right\}^2 \left\{  \Delta_1(\bm{W}^0_{i}) -   \bm{W}^\top_{i} \bm{\beta}_{1}  \right\}^2, \\
\bm{\beta}^*_{1n}&= \argmin_{\bm{\beta}_{1} \in \R^{d_1}} \sum_{i \in \bm{I}_n}  \left\{ A_{1i}- \mu_{1A}(\bm{W}_i)\right\}^2 \left\{  \Delta_1(\bm{W}^0_{i}) -   \bm{W}^\top_{i} \bm{\beta}_{1}  \right\}^2.
\end{align*}
 \end{lemma}

\begin{proof}
Below, we will prove that $\| \tilde{\bm{\beta}}^*_{2n}-{\bm{\beta}}^{*}_{2n}    \|_\infty=o_p(n^{-1/2});$ the result that
$\| \tilde{\bm{\beta}}^*_{1n}-{\bm{\beta}}^{*}_{1n}    \|_\infty=o_p(n^{-1/2})$ follows from essentially identical arguments. 
Using the definitions of $\tilde{\bm{\beta}}^*_{2n}$ and ${\bm{\beta}}^{*}_{2n}$ and assuming $n$ is large enough so that
Assumption \ref{assump:posdef} holds, straightforward algebra shows 
\begin{align*}
 \tilde{\bm{\beta}}^*_{2n}-{\bm{\beta}}^{*}_{2n} = &
   (\hat \bV_{2n}^{-1}- \bV_{2n}^{-1}) \left\{ \frac{1}{n} \sum_{i \in \bm{I}_n}   (A_{2i}-  \mu_{2Ai})^2  \bm{S}_i \Delta_{2i}  \right\} \\
   & +\hat \bV_{2n}^{-1} \left\{ \frac{1}{n} \sum_{i \in \bm{I}_n}   ( \mu_{2Ai}- \hat \mu_{2Ai})^2  \bm{S}_i \Delta_{2i}  \right\}.
\end{align*}
Taking norms and using the triangle inequality, it can be shown that
\begin{equation}
\label{b2bound}
\sqrt n \| \tilde{\bm{\beta}}^*_{2n}-{\bm{\beta}}^{*}_{2n}\|_\infty \leq \sqrt n   \bigl\| \hat \bV_{2n}^{-1}- \bV_{2n}^{-1} \bigr\|_\infty  \, (A_n + B_n) + \sqrt n \bigl\| \bV_{2n}^{-1} \bigr\|_\infty \, B_n
\end{equation}
where
\begin{align*}
A_n & = \bigg\|  \frac{1}{n} \sum_{i \in \bm{I}_n}   (A_{2i}-  \mu_{2Ai})^2  \bm{S}_i \Delta_{2i} \biggr\|_\infty \\
B_n & =   \bigg\|  \frac{1}{n} \sum_{i \in \bm{I}_n}   ( \hat \mu_{2Ai}- \mu_{2Ai})^2 \bm{S}_i \Delta_{2i} \biggr\|_\infty.
\end{align*}
Suppose that $\|\hat \bV_{2n} - \bV_{2n}\|_{\infty} = o_p( n^{-1/2}).$  Then, by Lemma \ref{lem:matinvcvg}
and Assumption \ref{assump:posdef}, we have for $n$ sufficiently large that
\begin{equation}
\label{V2invcvg}
\|\hat \bV_{2n}^{-1} - \bV_{2n}^{-1}\|_{\infty} \leq K_1 \|\hat \bV_{2n} - \bV_{2n}\|_{\infty} 
\end{equation}
for any constant $K_1$ such that $2  \| \bV_{2n}^{-1} \|^2_{\infty} \leq K_1$.
It can be seen that 
\[
A_n \stackrel{p}{\rightarrow} \Big\|  \E\left[  \{A_2 - \mu_{2A}(\bm{S}^0)\}^2 \bm{S} \Delta_2(\bm{S}^0) \right] \Bigr\|_\infty
= \Big\| \E\left[   \bm{S} \Delta_2(\bm{S}^0) \mbox{var}( A_2 | \bm{S}^0) \right] \Bigr\|_\infty < \infty.
\]
Since cross-fitting is used to estimate 
$\hat \mu_{2A}(\cdot),$ we also see that $B_n$ is an example of \eqref{Gfun2};
hence, using Lemma \ref{lem:help1cvg} and Assumption \ref{assump:accuracy-treatment}, 
it follows that $B_n = o_p( n^{-1/2}).$ It follows from these results and \eqref{V2invcvg} 
that \eqref{b2bound} is $o_p(1).$ 

In order to prove that $\|\hat \bV_{2n} - \bV_{2n}\|_{\infty} = o_p(n^{-1/2}),$ we begin by writing
$\|\hat \bV_{2n} - \bV_{2n}\|_{\infty} \leq C_n + 2 D_n,$ where
\begin{align*}
C_n & =  \left\|  \frac{1}{n} \sum_{i \in \bm{I}_n}   ( \mu_{2Ai}-  \hat \mu_{2Ai})^2  \bm{S}^{\otimes 2}_{i}  \right \|_{\infty} \\
D_n & =  \left\| \frac{1}{n} \sum_{i \in \bm{I}_n} (A_{2i}-   \mu_{2Ai})(  \hat \mu_{2Ai} - \mu_{2Ai})  \bm{S}^{\otimes 2}_{i}   \right \|_{\infty}.
\end{align*}
Again, because  cross-fitting is used to estimate 
$\hat \mu_{2A}(\cdot),$  we can see that $C_n$ is also an example of \eqref{Gfun2}
and it follows by previously stated arguments that $C_n = o_p( n^{-1/2})$. 
In order to establish the behavior of $D_n$, we first note that
$D_n = \max_{j=1,\ldots,d_2} \sum_{k=1}^{d_2} 
\left| H_{njk}  \right|$
where $d_2$ is finite and
\[
H_{njk} =  \frac{1}{n} \sum_{i \in \bm{I}_n} (A_{2i}-   \mu_{2Ai})(  \hat \mu_{2Ai} - \mu_{2Ai})  S_{ik} S_{ij}. 
\]
It suffices to establish the behavior of $H_{njk}$. First, using the definition
of $\mu_{2Ai} = \mu_{2A}(\bm{S}^0_i) = \E( A_{2i} | \bm{S}^0_i)$ and the fact that
$\hat \mu_{2Ai} = \hat \mu_{2A}(\bm{S}^0_i)$ where $\hat \mu_{2A}(\cdot)$ is estimated
from data $\bm{D}_{I^c_n}$ that is independent of $\bm{S}^0_i \in \bm{D}_{I_n}$ for each $i,$
it is easy to see that $\E( H_{njk} | \bm{S}^0_1, \ldots,  \bm{S}^0_n, \bm{D}_{I^c_n}) = 0$ and hence that
$\E( H_{njk} | \bm{D}_{I^c_n}) = 0.$ Using these same properties, it is also easily
shown that
\[
\mbox{var}( H_{njk} | \bm{D}_{I^c_n}) = \frac{1}{n^2} \sum_{i \in \bm{I}_n} 
\E\left\{ 
 (  \hat \mu_{2Ai} - \mu_{2Ai})^2  
 \mbox{var}( A_{2i} | \bm{S}^0_i)
 (S_{ik} S_{ij})^2 | \bm{D}_{I^c_n}
\right\}.
\]
Under Assumption \ref{assump:support-X}, we can find a constant $K_2 < \infty$ such that
\begin{eqnarray*}
\mbox{var}( H_{njk} | \bm{D}_{I^c_n}) & \leq & \frac{K_2}{n}  
\E\left[ \frac{1}{n} \sum_{i \in \bm{I}_n}  \{ \hat \mu_{2A}(\bm{S}^0_i) - \mu_{2A}(\bm{S}^0_i) \}^2  \Big| \bm{D}_{I^c_n} \right] \\
& = & 
\frac{K_2}{n} \E\left\{ \| \hat \mu_{2A}(\bm{S}^0; \bm{D}_{I^c_n}) - \mu_{2A}(\bm{S}^0) \|^2_{\mathbb{P}_n,2}   \Big| \bm{D}_{I^c_n} \right\} \\
& = & 
\frac{K_2}{n}  \| \hat \mu_{2A}(\bm{S}^0; \bm{D}_{I^c_n}) - \mu_{2A}(\bm{S}^0) \|^2_{P_0,2},
\end{eqnarray*}
By Chebyshev's inequality, for all $\epsilon > 0$ we then have 
\[
P\left( n^{1/2} \left| H_{njk} \right|  > \epsilon \bigl| \bm{D}_{I^c_n} \right) \leq \frac{K_2}{\epsilon^2} \| \hat \mu_{2A}(\bm{S}^0; \bm{D}_{I^c_n}) - \mu_{2A}(\bm{S}^0) \|^2_{P_0,2},
\]
where the right-hand side is $o_p(n^{-1/2})$ by Assumption \ref{assump:accuracy-treatment}.
Lemma \ref{lem:condcvg} now implies that
$H_{njk} = o_p(n^{-1/2})$ and hence that $D_n = o_p(n^{-1/2}).$  Therefore, 
$\|\hat \bV_{2n} - \bV_{2n}\|_{\infty} \leq C_n + 2 D_n = o_p(n^{-1/2}),$ proving the desired result.
A similar argument shows $\sqrt n\| \tilde{\bm{\beta}}^*_{1n}-{\bm{\beta}}^{*}_{1n}\|_{\infty} = o_p(1)$.
\end{proof}

As in the main paper, we define
\begin{eqnarray}
\label{b2star}
\bm{\beta}^*_{2} &=& \argmin_{\bm{\beta}_{2} \in \R^{d_2}} \E\left[ 
\mbox{var}\left( A_{2} | \bm{S}^0 \right) \left\{  \Delta_2(\bm{S}^0) -   \bm{S}^\top \bm{\beta}_{2}  \right\}^2 \right], \\
\label{V2}
\bV_2 & = & \E\left\{  \mbox{var}\left( A_{2} | \bm{S}^0 \right) \bm{S}^{\otimes 2} \right\} \\
\label{b1star}
\bm{\beta}^*_{1}&= &\argmin_{\bm{\beta}_{1} \in \R^{d_1}} \E\left[ 
\mbox{var}\left( A_{1} | \bm{W}^0 \right) \left\{  \Delta_1(\bm{W}^0) -   \bm{W}^\top \bm{\beta}_{1}  \right\}^2 \right], \\
\label{V1}
\bV_1 & = & \E\left\{ \mbox{var}\left( A_{1} | \bm{W}^0 \right) \bm{W}^{\otimes 2} \right\}.
\end{eqnarray}
As $n \rightarrow \infty$, it is easy to see that $\bm{\beta}^*_{2n}$ converges to 
\[
\bm{\beta}^*_{2} = \bV_2^{-1} \E\left\{
\mbox{var}\left( A_{2} | \bm{S}^0 \right) \bm{S} \Delta_2(\bm{S}^0) \right\};
\]
similarly, $\bm{\beta}^*_{1n}$ converges to
\[
\bm{\beta}^*_{1} = \bV_1^{-1} \E\left\{
\mbox{var}\left( A_{1} | \bm{W}^0 \right) \bm{W} \Delta_1(\bm{W}^0) \right\}.
\]

With these preliminaries in place, we can now prove the main result.

\begin{proof}[Proof of Theorem 1, part (a)]

We desire to show that $\hat {\bm{\beta}}_{2n}$ is an asymptotically linear estimator of $\bm{\beta}_2^*$ with the
claimed influence function, where $\bm{\beta}^{*}_{2}$ is defined in \eqref{b2star} and
\begin{align*}
\hat {\bm{\beta}}_{2n}&= \argmin_{\bm{\beta}_2\in \R^{d_2}} \sum_{i \in \bm{I}_n} \left[ Y_i-\hat\mu_{2Y}(\bm{S}^0_i) - \{A_{2i}-\hat\mu_{2A}(\bm{S}^0_i)\} \cdot \bm{S}_i^\top \bm{\beta}_2    \right]^2.
\end{align*}
In view of Lemma \ref{lem:beta1-fixed}, we can proceed by establishing asymptotically linear representations for both
$\sqrt{n}( \hat {\bm{\beta}}_{2n} - \tilde {\bm{\beta}}^*_{2n})$ and
$\sqrt{n}({\bm{\beta}}^*_{2n} -  {\bm{\beta}}^*_{2});$ combined, these will lead to
that for $\sqrt{n}( \hat {\bm{\beta}}_{2n} -  {\bm{\beta}}^*_{2}).$

Recalling notation introduced earlier, it is easy to show that
\begin{align}
\nonumber
\sqrt{n} ( \hat {\bm{\beta}}_{2n} & - \tilde {\bm{\beta}}^*_{2n} ) = \\
\label{beta2diffrep}
&   \hat \bV_{2n}^{-1}  \left[
\frac{1}{\sqrt{n}} \sum_{i \in \bm{I}_n}   
(A_{2i}-  \hat \mu_{2Ai})  \bm{S}_i 
\{ Y_i-\hat\mu_{2Yi} -(A_{2i}-  \hat \mu_{2Ai})  \Delta_{2i} \}
\right]
\end{align}

By adding and subtracting terms and using the model assumptions on $Y_i$, one can write
\begin{equation}
\label{Ycenter}
Y_i-\hat\mu_{2Yi}
   -(A_{2i}-  \hat \mu_{2Ai})  \Delta_{2i} = \epsilon_{2i} + (\mu_{2Yi} - \hat\mu_{2Yi}) + (\hat \mu_{2Ai} - \mu_{2Ai}) \Delta_{2i}.
\end{equation}
The decomposition \eqref{Ycenter} implies that the term in the square brackets on the right-hand side of \eqref{beta2diffrep}
can be decomposed into six terms:
\begin{align}
\frac{1}{\sqrt{n}} \sum_{i \in \bm{I}_n}   
(A_{2i}-   \hat \mu_{2Ai})   & \bm{S}_i
\{ Y_i-\hat\mu_{2Yi} -(A_{2i}-   \mu_{2Ai})  \Delta_{2i} \}  \\
\label{b2decomp-1}
& =  \frac{1}{\sqrt{n}} \sum_{i \in \bm{I}_n}   \epsilon_{2i} (A_{2i}-  \mu_{2Ai})  \bm{S}_i  \\
\label{b2decomp-2}
& -  \frac{1}{\sqrt{n}} \sum_{i \in \bm{I}_n}   (A_{2i}-  \mu_{2Ai})   (\hat \mu_{2Yi} - \mu_{2Yi}) \bm{S}_i \\
\label{b2decomp-3}
& +  \frac{1}{\sqrt{n}} \sum_{i \in \bm{I}_n}   (A_{2i}-  \mu_{2Ai})   (\hat \mu_{2Ai} - \mu_{2Ai}) \Delta_{2i} \bm{S}_i \\
\label{b2decomp-4}
& -  \frac{1}{\sqrt{n}} \sum_{i \in \bm{I}_n}   \epsilon_{2i} (\hat \mu_{2Ai} -  \mu_{2Ai})    \bm{S}_i\\
\label{b2decomp-5}
& +  \frac{1}{\sqrt{n}} \sum_{i \in \bm{I}_n}   (\hat \mu_{2Ai} -  \mu_{2Ai})  (\hat \mu_{2Yi} - \mu_{2Yi})  \bm{S}_i  \\
\label{b2decomp-6}
& -  \frac{1}{\sqrt{n}} \sum_{i \in \bm{I}_n}   (\hat \mu_{2Ai} -  \hat \mu_{2Ai})^2  \Delta_{2i} \bm{S}_i . 
\end{align}
Under the assumptions of this theorem, the central limit theorem establishes the asymptotic normality 
of \eqref{b2decomp-1}, which is $O_p(1)$. The terms \eqref{b2decomp-2}-\eqref{b2decomp-4} are each seen
to be examples to which Lemma \ref{lem:help2cvg} applies; under Assumptions
\ref{assump:support-X}-\ref{assump:posdef}, it follows that each term is $o_p(1)$. 
The terms \eqref{b2decomp-5} and  \eqref{b2decomp-6} are both seen
to be examples to which Lemma \ref{lem:help1cvg} applies; again, under Assumptions
\ref{assump:support-X}-\ref{assump:posdef}, each term is $o_p(1).$ Because
$\| \hat \bV_{2n}^{-1} - \bV^{-1}_{2n} \|_{\infty} = o_p(1)$ it follows that
$\sqrt{n} ( \hat {\bm{\beta}}_{2n}  - \tilde {\bm{\beta}}^*_{2n} )$ can be written
\begin{equation}
\label{b2n-influence1}
\sqrt{n} ( \hat {\bm{\beta}}_{2n}  - \tilde {\bm{\beta}}^*_{2n} ) = 
 \frac{1}{\sqrt{n}} \sum_{i \in \bm{I}_n}   \bV^{-1}_{2n} \epsilon_{2i} (A_{2i}-  \mu_{2Ai})  \bm{S}_i  + o_p(1).
\end{equation}

Turning to $\sqrt{n}({\bm{\beta}}^*_{2n} -  {\bm{\beta}}^*_{2}),$ we can write
\begin{align}
\nonumber
 \sqrt{n}({\bm{\beta}}^*_{2n} -  {\bm{\beta}}^*_{2})   & =  
\sqrt{n} \left[ \biggl\{ n^{-1} \sum_{i \in \bm{I}_n}   \bV^{-1}_{2n}  (A_{2i}-  \mu_{2Ai})^2 \Delta_{2i}  \bm{S}_i \biggr\}  - {\bm{\beta}}^*_{2} \right] \\
\label{b2n-influence2}
& = 
\frac{1}{\sqrt{n}}  \bV^{-1}_{2n} \sum_{i \in \bm{I}_n}   (A_{2i}-  \mu_{2Ai})^2  \bm{S}_i \bigl( \Delta_{2i}  - \bm{S}^{\top}_i {\bm{\beta}}^*_{2} \bigr). 
\end{align}
Hence, using \eqref{b2n-influence1} and \eqref{b2n-influence2} and collecting terms,
\[
\sqrt{n} ( \hat {\bm{\beta}}_{2n}  -  {\bm{\beta}}^*_{2} )  = 
\frac{1}{\sqrt{n}}  \bV^{-1}_{2n} \sum_{i \in \bm{I}_n}  
 (A_{2i}-  \mu_{2Ai})  \bm{S}_i  H_{2i}
+ o_p(1) 
\]
where $
H_{2i} =  \epsilon_{2i} + 
 (A_{2i}-  \mu_{2Ai})  \bigl( \Delta_{2i}  - \bm{S}^{\top}_i {\bm{\beta}}^*_{2} \bigr).$
Using the fact that 
$Y_i - \mu_{2Yi} = \epsilon_{2i} + (A_{2i}-  \mu_{2Ai})  \Delta_{2i},$
we may write
\begin{eqnarray}
\nonumber
H_{2i} & = & Y_i - \mu_{2Yi} - (A_{2i}-  \mu_{2Ai})  \Delta_{2i} + 
 (A_{2i}-  \mu_{2Ai})  \bigl( \Delta_{2i}  - \bm{S}^{\top}_i {\bm{\beta}}^*_{2} \bigr) \\
 & = &Y_i - \mu_{2Yi} - (A_{2i}-  \mu_{2Ai})  \bm{S}^{\top}_i {\bm{\beta}}^*_{2}.
\end{eqnarray}
Consequently,
\[
\sqrt{n} ( \hat {\bm{\beta}}_{2n}  - {\bm{\beta}}^*_{2} ) 
  =    \frac{1}{\sqrt{n}} \sum_{i \in \bm{I}_n}   \mbox{\rm Inf}_{2in} +o_p(1)
\]
where
\[
 \mbox{\rm Inf}_{2in}  =
 \bV^{-1}_{2n} (A_{2i}-  \mu_{2Ai})  \bm{S}_i  
 \bigl\{
 Y_i - \mu_{2Yi} - (A_{2i}-  \mu_{2Ai})  \bm{S}^{\top}_i {\bm{\beta}}^*_{2}
 \bigr\}.
\]
Since
$\| \bV_{2n}^{-1} - \bV^{-1}_{2} \|_{\infty} = o_p(1),$
it now follows that
\[
\sqrt{n} ( \hat {\bm{\beta}}_{2n}  - {\bm{\beta}}^*_{2} ) 
  =    \frac{1}{\sqrt{n}} \sum_{i \in \bm{I}_n}   \mbox{\rm Inf}_{2i} +o_p(1),
\]
where
\begin{equation}
\label{final-influence2}
 \mbox{\rm Inf}_{2i}  =
 \bV^{-1}_2 
 (A_{2i}-  \mu_{2Ai})  \bm{S}_i  
 \bigl\{
 Y_i - \mu_{2Yi} - (A_{2i}-  \mu_{2Ai})  \bm{S}^{\top}_i {\bm{\beta}}^*_{2} \bigr\}
 \end{equation}
has mean zero and variance $\bV^{-1}_2 \bm{Q}_2 \bV^{-1}_2$ where
$\bQ_2 = \E ( \bm{J}_2^{\otimes 2} )$ and
\begin{equation}
\label{J2}
\bm{J}_2 = 
\{A_{2}-  \mu_{2A}(\bm{S}^0)\}  \bm{S}  
\left[
Y - \mu_{2Y}(\bm{S}^0)  - \{ A_{2}-  \mu_{2A}(\bm{S}^0) \}  \bm{S}^{\top} {\bm{\beta}}^*_{2}
\right].
\end{equation}
\end{proof}

\begin{proof}[Proof of Theorem 1, part (b)]
As in part (a), we need to show that $\hat {\bm{\beta}}_{1n}$ is an asymptotically linear estimator of $\bm{\beta}_1^*$ with a 
certain influence function, where $\bm{\beta}^{*}_{1}$ is defined in \eqref{b1star} and
\begin{align*}
\hat {\bm{\beta}}_{1n}&= \argmin_{\bm{\beta}_1\in \R^{d_1}} \sum_{i \in \bm{I}_n} 
\left[ \Yhat - \hat\mu_{1Y}(\bm{W}^0_i) - \{A_{1i}-\hat\mu_{1A}(\bm{W}^0_i)\} \cdot \bm{W}_i^\top \bm{\beta}_1    \right]^2,
\end{align*}
where $\Yhat$ is calculated as
\begin{equation}
\label{eq:Ytil-hat}
\Yhat = Y_i + \bm{S}_i^\top \hat{\bm{\beta}}_{2n}
\bigl\{ I(\bm{S}_i^\top  \hat{\bm{\beta}}_{2n}  > 0 ) - A_{2i} \bigr\}.
\end{equation}

Proceeding similarly to the proof of part (a), we will establish asymptotically linear representations for both
$\sqrt{n}( \hat {\bm{\beta}}_{1n} - \tilde {\bm{\beta}}^*_{1n})$ and
$\sqrt{n}({\bm{\beta}}^*_{1n} -  {\bm{\beta}}^*_{1});$ combining these will
provide the claimed influence function for $\sqrt{n}( \hat {\bm{\beta}}_{1n} -  {\bm{\beta}}^*_{1}).$

We begin with $\sqrt{n}( \hat {\bm{\beta}}_{1n} - \tilde {\bm{\beta}}^*_{1n}).$
Define $\epsilon_{1i} = \Ytil - \E( \Ytil \mid \bm{W}^0_i, A_{1i}),$
where $\Ytil$ is given by
\begin{equation}
\label{eq:Ytil}
\Ytil = Y_i + \bm{S}_i^\top \bm{\beta}^*_{2}
\bigl\{ I(\bm{S}_i^\top  \bm{\beta}^*_{2}  > 0 ) - A_{2i} \bigr\};
\end{equation}
by construction, 
$\epsilon_{1i}  = \Ytil - \mu_{1Ai} - (A_{1i} - \mu_{1Ai}) \Delta_{1i}$
and
$\E(\epsilon_{1i} \mid \bm{W}^0_i, A_{1i})=0$. 
In addition, let $\hat \delta_i = \Yhat - \Ytil.$  Similarly to the proof
in part (a), we can decompose $\sqrt{n}( \hat {\bm{\beta}}_{1n} - \tilde {\bm{\beta}}^*_{1n})$ into the
sum of several terms:

\newpage 

\begin{align}
\nonumber
\sqrt{n}( \hat {\bm{\beta}}_{1n} - & \tilde {\bm{\beta}}^*_{1n}) =  \\
\label{T1}
& \hat \bV_{1n}^{-1} \Big\{  \frac{1}{\sqrt{n}} \sum_{i \in \bm{I}_n} (A_{1i} -
  {\mu}_{1Ai}) \big(  \epsilon_{1i}  + \hat \delta_i   \big)    \bm{W}_i  \Big\} \\
\label{T2}
& - \hat \bV_{1n}^{-1} \Big\{  \frac{1}{\sqrt{n}} \sum_{i \in \bm{I}_n} (A_{1i} - {\mu}_{1Ai})  (\hat {\mu}_{1Yi}- {\mu}_{1Yi}) \bm{W}_i \Big\} \\
\label{T3}
&  + \hat \bV_{1n}^{-1} \Big\{  \frac{1}{\sqrt{n}} \sum_{i \in \bm{I}_n} (A_{1i} -  {\mu}_{1Ai}) (\hat {\mu}_{1Ai}- {\mu}_{1Ai}) \Delta_{1i} \bm{W}_i  \Big\} \\
\label{T4}
& - \hat \bV_{1n}^{-1} \Big\{  \frac{1}{\sqrt{n}} \sum_{i \in \bm{I}_n} ( \hat {\mu}_{1Ai}- {\mu}_{1Ai})   \epsilon_{1i}   \bm{W}_i \Big\}  \\
\label{T5}
& - \hat \bV_{1n}^{-1} \Big\{  \frac{1}{\sqrt{n}} \sum_{i \in \bm{I}_n} ( \hat {\mu}_{1Ai}- {\mu}_{1Ai})  \hat \delta_i     \bm{W}_i \Big\}  \\
\label{T6}
 & -   \hat \bV_{1n}^{-1} \Big\{  \frac{1}{\sqrt{n}} \sum_{i \in \bm{I}_n} ( \hat {\mu}_{1Ai}- {\mu}_{1Ai})  (\hat {\mu}_{1Yi}- {\mu}_{1Yi})    \bm{W}_i  \Big\} \\
 \label{T7}
 & +  \hat \bV_{1n}^{-1} \Big\{ \frac{1}{\sqrt{n}} \sum_{i \in \bm{I}_n} (\hat {\mu}_{1Ai}- {\mu}_{1Ai})^2  \Delta_{1i} \bm{W}_i
      \Big\}  
\end{align}

Assuming that $\hat \mu_{1A}(\cdot)$ and $\hat \mu_{1Y}(\cdot)$ are estimated similarly to $\hat \mu_{2A}(\cdot)$ and $\hat \mu_{2Y}(\cdot)$
(i.e., meaning, sample splitting has been used) and in view of the fact that $\hat \bV_{1n}$ is a consistent estimator of
$\bV_{1},$ Lemma  \ref{lem:help2cvg} implies that the terms \eqref{T2}, \eqref{T3}, and \eqref{T4} are all $o_p(1)$
under Assumptions \ref{assump:support-X} -- \ref{assump:posdef}; similarly,
Lemma  \ref{lem:help1cvg} implies that  the terms \eqref{T6} and \eqref{T7}  are also $o_p(1)$
under these same assumptions. To complete this part of the proof, we must 
therefore establish the asymptotic behavior of \eqref{T1} and \eqref{T5}, both of
which depend on the asymptotic behavior of $\hat \delta_i = \Yhat - \Ytil.$ 
The terms \eqref{T1} and \eqref{T5} isolate the potential for non-regular behavior; however,
as we will see, Assumption \ref{assump:unique} is only needed for controlling such behavior  in \eqref{T1}.

To determine the asymptotic behavior of \eqref{T1}, let $D_ i = I(\bm{S}_i^\top   \bm{\beta}^*_{2}  > 0 ) - A_{2i}$ and
\[
\hat R_{ni} = I( \bm{S}_i^\top  \hat{\bm{\beta}}_{2n}  > 0) - I(\bm{S}_i^\top  \bm{\beta}^*_{2}  > 0)
\]
for $i \in \bm{I}_n.$

Algebra now shows
\begin{align}
\nonumber
\frac{1}{\sqrt{n}} \sum_{i \in \bm{I}_n}  (A_{1i} - {\mu}_{1Ai})  \bm{W}_i   & \big(  \epsilon_{1i}  + \hat \delta_i   \big)   = \\
\label{G1}
&  \frac{1}{\sqrt{n}} \sum_{i \in \bm{I}_n}  (A_{1i} - {\mu}_{1Ai})  \bm{W}_i  \{ \epsilon_{1i}  + \bm{S}_i^\top  (\hat{\bm{\beta}}_{2n}-\bm{\beta}^*_{2})  D_i  \} \\
\label{G2}
&  + \frac{1}{\sqrt{n}} \sum_{i \in \bm{I}_n}  (A_{1i} - {\mu}_{1Ai})  \bm{W}_i  \bm{S}_i^\top  \bm{\beta}^*_{2}  \hat R_{ni}     \\
\label{G3}
&  + \frac{1}{\sqrt{n}} \sum_{i \in \bm{I}_n}  (A_{1i} - {\mu}_{1Ai})  \bm{W}_i  \bm{S}_i^\top  ( \hat{\bm{\beta}}_{2n}-\bm{\beta}^*_{2})  \hat R_{ni}.
\end{align}
Although \eqref{G2} and \eqref{G3} can be easily combined, treating these two terms separately turns out to be advantageous.
We first consider \eqref{G2}. Note that $\hat R_{ni} \in \{-1,0,1\}$ and, importantly, that $| \hat R_{ni} | \leq R_{ni},$ where
$
R_{ni} = I\{ 0 \leq | \bm{S}_i^\top  \bm{\beta}^*_{2} | \leq | \bm{S}_i^\top  ( \hat{\bm{\beta}}_{2n}-\bm{\beta}^*_{2}) | \}.
$
It follows that
\[
| \bm{S}_i^\top  \bm{\beta}^*_{2} | R_{ni} \leq | \bm{S}_i^\top  ( \hat{\bm{\beta}}_{2n}-\bm{\beta}^*_{2}) | R_{ni},
\]
an inequality that is trivially true when $R_{ni} = 0$ and true for $R_{ni} = 1$ in view of its definition.
Consequently, considering the $j^{th}$ element of the vector $\bm{W}_i,$ we have
\begin{eqnarray*}
\nonumber
 \big| \sum_{i \in \bm{I}_n}  (A_{1i} - {\mu}_{1Ai})  W_{ij}  \bm{S}_i^\top  \bm{\beta}^*_{2}  \hat R_{ni}   \bigr|
 & \leq & 
 \sum_{i \in \bm{I}_n}   \big|  (A_{1i} - {\mu}_{1Ai}) W_{ij}  \big|   \big|  \bm{S}_i^\top  \bm{\beta}^*_{2}   \big|   \big|  \hat R_{ni} \big| \\
& \leq & 
 \sum_{i \in \bm{I}_n}   \big|  (A_{1i} - {\mu}_{1Ai})  W_{ij}  \big|   \big|  \bm{S}_i^\top  \bm{\beta}^*_{2}   \big|   R_{ni}  \\
& \leq & 
 \sum_{i \in \bm{I}_n}   \big|  A_{1i} - {\mu}_{1Ai} \big| \big|  W_{ij}  \big|   \big|  \bm{S}_i^\top  ( \hat{\bm{\beta}}_{2n}-\bm{\beta}^*_{2})   \big|   R_{ni} \\
& \leq & 4 C
 \sum_{i \in \bm{I}_n}     \big|  \bm{S}_i^\top  ( \hat{\bm{\beta}}_{2n}-\bm{\beta}^*_{2})   \big|   R_{ni},
\end{eqnarray*}
the last step following from the fact that $A_{1i}$ is binary, $\hat \mu_{1A}(\cdot) \in [0,1],$ and $|W_{ij}|$ is bounded, say, by 
a finite constant $C$. 
Considering \eqref{G3}, a similar calculation shows that
\begin{eqnarray*}
\nonumber
 \big|
 \sum_{i \in \bm{I}_n}  (A_{1i} - {\mu}_{1Ai})  \bm{W}_i  \bm{S}_i^\top  ( \hat{\bm{\beta}}_{2n}-\bm{\beta}^*_{2})  \hat R_{ni}
  \bigr|
 & \leq & 
 4 C
 \sum_{i \in \bm{I}_n}     \big|  \bm{S}_i^\top  ( \hat{\bm{\beta}}_{2n}-\bm{\beta}^*_{2})   \big|   R_{ni}.
 \end{eqnarray*}
Therefore, 
\begin{eqnarray*}
\mbox{ \eqref{G2} + \eqref{G3} } & \leq & 
8 C \frac{1}{\sqrt{n}} 
 \sum_{i \in \bm{I}_n}     \big|  \bm{S}_i^\top  ( \hat{\bm{\beta}}_{2n}-\bm{\beta}^*_{2})   \big|   R_{ni} \\
& \leq & 8 C \sqrt{n}
\left[ \frac{1}{n} \sum_{i \in \bm{I}_n}   \left\{  \bm{S}_i^\top  ( \hat{\bm{\beta}}_{2n}-\bm{\beta}^*_{2})   \right\}^2 \right]^{1/2}
\left[ \frac{1}{n} \sum_{i \in \bm{I}_n}   R_{ni}^2 \right]^{1/2} \\
& = & 8 C \lambda_{1n} \sqrt{n} \| \hat{\bm{\beta}}_{2n}-\bm{\beta}^*_{2} \|_2 \, \bar{R}_n^{1/2},
\end{eqnarray*}
where $ \| \cdot \|_2$ is the usual Euclidean vector norm and $\lambda_{1n} > 0$ is the square root of the maximum eigenvalue of 
$n^{-1} \sum_{i \in \bm{I}_n}   \bm{S}^{\otimes 2}_i.$
Because $\sqrt{n} \| \hat{\bm{\beta}}_{2n}-\bm{\beta}^*_{2} \|_2 = O_p(1)$ and $\lambda_{1n}$ converges to a finite constant as
$n \rightarrow \infty$ under our assumptions, it follows that \eqref{G2} + \eqref{G3} is $o_p(1)$ if $\bar{R}_n = o_p(1).$ However,
Markov's inequality implies that
\[
P( \bar{R}_n > \alpha \mid \bm{S}^0_1,\ldots, \bm{S}^0_n) \leq (n \alpha)^{-1}  \sum_{i \in \bm{I}_n}  \E(R_{ni} 
\mid  \bm{S}^0_1,\ldots, \bm{S}^0_n
)
\]
for any $\alpha > 0,$ where
$\E(R_{ni} \mid  \bm{S}^0_1,\ldots, \bm{S}^0_n)  =  
P\{ R_{ni} = 1 \mid  \bm{S}^0_1,\ldots, \bm{S}^0_n \}.$
Letting
\[
I^*_i = I\big( | \bm{S}_i^\top  \bm{\beta}^*_{2} | = 0 \big),
\]
an easy conditioning argument shows 
$\E(R_{ni} \mid  \bm{S}^0_1,\ldots, \bm{S}^0_n ) = 
I^*_i + (1-I^*_i) P_{ni}$ where
\[
P_{ni} =
P\{  | \bm{S}_i^\top  ( \hat{\bm{\beta}}_{2n}-\bm{\beta}^*_{2}) | \geq k_i \mid  (\bm{S}^0_1,\ldots, \bm{S}^0_n), \, k_i > 0\}
\]
for $k_i = | \bm{S}_i^\top  \bm{\beta}^*_{2} |.$
Letting $n \rightarrow \infty,$ the fact that $\hat{\bm{\beta}}_{2n} \stackrel{p}{\rightarrow} \bm{\beta}^*_{2}$
as $n \rightarrow \infty$ implies $P_{ni} \rightarrow 0$ for each $i$; hence,
\[
\lim_{n \rightarrow \infty}
P( \bar{R}_n > \alpha \mid \bm{S}^0_1,\ldots, \bm{S}^0_n) \leq \frac{1}{n \alpha}  \lim_{n \rightarrow \infty}  \sum_{i \in \bm{I}_n}  I^*_i.
\]
However, under our assumptions, 
\[
 \lim_{n \rightarrow \infty} \frac{1}{n} \sum_{i \in \bm{I}_n}  I^*_i \stackrel{p}{\rightarrow} P\big( | \bm{S}_1^\top  \bm{\beta}^*_{2} | = 0 \big)
\]
and it follows from Assumption \ref{assump:unique} that  $\bar{R}_n \stackrel{p}{\rightarrow} 0.$

Because $\| \hat \bV_{1n}^{-1} - \bV^{-1}_{1n} \|_{\infty} = o_p(1),$ we have now shown that
\begin{align}
\mbox{\eqref{T1}} 
  \label{b1n-influence1}
& = \bV_{1n}^{-1}  \Big[ \frac{1}{\sqrt{n}} \sum_{i \in \bm{I}_n}  (A_{1i} - {\mu}_{1Ai})  \bm{W}_i  \{ \epsilon_{1i}  + \bm{S}_i^\top  (\hat{\bm{\beta}}_{2n}-\bm{\beta}^*_{2})  D_i  \} + o_p(1) \Big]
\end{align}
under Assumptions \ref{assump:support-X} -- \ref{assump:unique}, where we recall the notation $D_ i = I(\bm{S}_i^\top   \bm{\beta}^*_{2}  > 0 ) - A_{2i}.$

To establish \eqref{T5}, observe that we may similarly decompose it as above, leading to
\begin{align}
\nonumber
\frac{1}{\sqrt{n}} \sum_{i \in \bm{I}_n}  (  \hat{\mu}_{1Ai} & - {\mu}_{1Ai})  \bm{W}_i    \hat \delta_i      = \\
\label{G4}
&  \frac{1}{\sqrt{n}} \sum_{i \in \bm{I}_n}  ( \hat{\mu}_{1Ai} - {\mu}_{1Ai}) \bm{W}_i  \epsilon_{1i}   \\
\label{G5}
&  + \frac{1}{\sqrt{n}} \sum_{i \in \bm{I}_n}  ( \hat{\mu}_{1Ai} - {\mu}_{1Ai})  \bm{W}_i  \bm{S}_i^\top  \bm{\beta}^*_{2}  \hat R_{ni}     \\
\label{G6}
&  + \frac{1}{\sqrt{n}} \sum_{i \in \bm{I}_n} ( \hat{\mu}_{1Ai} - {\mu}_{1Ai})  \bm{W}_i  \bm{S}_i^\top   ( \hat{\bm{\beta}}_{2n}-\bm{\beta}^*_{2})  \hat R_{ni}     \\
\label{G7}
&  + \frac{1}{\sqrt{n}} \sum_{i \in \bm{I}_n}  ( \hat{\mu}_{1Ai} - {\mu}_{1Ai}) \bm{W}_i  \bm{S}_i^\top  ( \hat{\bm{\beta}}_{2n}-\bm{\beta}^*_{2})  D_i.
\end{align}
The term \eqref{G4} can be handled using Lemma  \ref{lem:help2cvg}. The remaining terms can be handled similarly to \eqref{G2} and \eqref{G3};
however, the required decomposition of terms differs some and, importantly, can make use of Assumption \ref{assump:accuracy-treatment}.
In particular, establishing the behavior of \eqref{G4}-\eqref{G7} can be done under 
Assumptions \ref{assump:support-X} -- \ref{assump:posdef}, without additionally imposing Assumption \ref{assump:unique}, 
showing that any effect of non-regularity is limited to the behavior of \eqref{G2} and \eqref{G3}, or equivalently,
\[
 \frac{1}{\sqrt{n}} \sum_{i \in \bm{I}_n}  (A_{1i} - {\mu}_{1Ai})  \bm{W}_i  \bm{S}_i^\top   \hat{\bm{\beta}}_{2n}  \hat R_{ni}.
 \]

The above proof establishes an asymptotic linear representation for
$\sqrt{n}( \hat {\bm{\beta}}_{1n} - \tilde {\bm{\beta}}^*_{1n}).$
Turning to $\sqrt{n}({\bm{\beta}}^*_{1n} -  {\bm{\beta}}^*_{1}),$ we can write
\begin{align}
 \sqrt{n}({\bm{\beta}}^*_{1n} -  {\bm{\beta}}^*_{1})   
\label{b1n-influence2}
& = 
\frac{1}{\sqrt{n}}  \bV^{-1}_{1n} \sum_{i \in \bm{I}_n}   (A_{1i}-  \mu_{1Ai})^2  \bm{W}_i \bigl( \Delta_{1i}  - \bm{W}^{\top}_i {\bm{\beta}}^*_{1} \bigr). 
\end{align}
Hence, using \eqref{b1n-influence1} and \eqref{b1n-influence2} and collecting terms, it follows that
\[
\sqrt{n} ( \hat {\bm{\beta}}_{1n}  -  {\bm{\beta}}^*_{1} ) = 
\frac{1}{\sqrt{n}}  \bV^{-1}_{1n} \sum_{i \in \bm{I}_n}  
 (A_{1i}-  \mu_{1Ai})  \bm{W}_i  H_{1i}
+ o_p(1) 
\]
where 
$H_{1i} =  \epsilon_{1i} + 
 (A_{1i}-  \mu_{1Ai})  \bigl( \Delta_{1i}  - \bm{W}^{\top}_i {\bm{\beta}}^*_{1} \bigr)
 + \bm{S}^{\top}_i ( \hat{\bm{\beta}}_{2n} - {\bm{\beta}}^*_{2} ) D_i.$
 Because 
 \[
 \epsilon_{1i}  = \Ytil - \mu_{1Ai} - (A_{1i} - \mu_{1Ai}) \Delta_{1i}
 \]
we have
\[
H_{1i} =  \Ytil - \mu_{1Ai} - (A_{1i} - \mu_{1Ai}) \bm{W}^{\top}_i {\bm{\beta}}^*_{1}
 + \bm{S}^{\top}_i ( \hat{\bm{\beta}}_{2n} - {\bm{\beta}}^*_{2} ) 
 \{ I(\bm{S}_i^\top   \bm{\beta}^*_{2}  > 0 ) - A_{2i} \}.
 \]
 Using the fact that $\| \bV_{1n}^{-1} - \bV^{-1}_{1} \|_{\infty} = o_p(1),$ it follows that
\[
\sqrt{n} ( \hat {\bm{\beta}}_{1n}  -  {\bm{\beta}}^*_{1} ) = 
\frac{1}{\sqrt{n}}  \bV^{-1}_{1} \sum_{i \in \bm{I}_n}  
 (A_{1i}-  \mu_{1Ai})  \bm{W}_i  H_{1i}
+ o_p(1). 
\]
Defining
\[
\bm{K}_n = \frac{1}{n}   \sum_{i \in \bm{I}_n}  
 (A_{1i}-  \mu_{1Ai}) \{ I(\bm{S}_i^\top   \bm{\beta}^*_{2}  > 0 ) - A_{2i} \} \bm{W}_i  \bm{S}^{\top}_i 
\]
and letting $\bm K$ denote its limit in probability, the results from part (a), in particular \eqref{final-influence2}, 
now imply that
\[
\sqrt{n} ( \hat {\bm{\beta}}_{1n}  -  {\bm{\beta}}^*_{1} ) = 
\frac{1}{\sqrt{n}}  \sum_{i \in \bm{I}_n}  \mbox{Inf}_{1i} + o_p(1). 
\]
for
\begin{equation}
\label{final-influence1}
\mbox{Inf}_{1i} = \bV^{-1}_{1} 
\bigl[  (A_{1i}-  \mu_{1Ai})  \bm{W}_i  
 \big\{ 
 \Ytil - \mu_{1Ai} - (A_{1i} - \mu_{1Ai}) \bm{W}^{\top}_i {\bm{\beta}}^*_{1} \big\}
 + \bm{K} \, \mbox{Inf}_{2i} \bigr].
 \end{equation}
This representation result implies 
$
\sqrt{n} ( \hat {\bm{\beta}}_{1n}  - {\bm{\beta}}^*_{1} ) 
  \stackrel{d}{\rightarrow} N( \bm{0},   \bV^{-1}_1 \bQ_1 \bV^{-1}_1 )
  $
where we define the matrix $\bQ_1 = \E \{ ( \bm{J}_1 + \bm{K} \bV^{-1}_2 \bm{J}_2) ^{\otimes 2} \},$ 
$\bm{J}_2$ is given in \eqref{J2}, 
\[
\bm{J}_1 = \{A_{1}-  \mu_{1A}(\bm{W}^0)\}  \bm{W} 
\left[
\tilde{Y}- \mu_{1A}(\bm{W}^0) - \{ A_{1} - \mu_{1A}(\bm{W}^0) \} \bm{W}^{\top} {\bm{\beta}}^*_{1} 
\right]
\]
and
\[
\bm{K} = \E \big[ \{ A_{1}-  \mu_{1A}(\bm{W}^0) \} \{ I(\bm{S}^\top   \bm{\beta}^*_{2}  > 0 ) - A_{2} \} \bm{W}  \bm{S}^{\top} \big].
\]

 \end{proof}
 
\begin{proof}[Proof of Corollary to Theorem 1]

As established in the proof of Theorem 1, the regularity Assumption \ref{assump:unique} is imposed
only to control the potentially non-regular behavior of the terms \eqref{G2} and \eqref{G3}. The origin
of this non-regular behavior is the dependence of each term on
\[
\hat R_{ni} = I( \bm{S}_i^\top  \hat{\bm{\beta}}_{2n}  > 0) - I(\bm{S}_i^\top  \bm{\beta}^*_{2}  > 0), i \in \bm{I}_n.
\]
In view of the proof of Theorem 1, establishing that each of  \eqref{G2} and \eqref{G3} is $o_p(1)$ is sufficient
to prove the corollary as stated.

To simply the proof of these results, let $\hat{\bm{\beta}}^{(-i)}_{2n}$ be the least squares estimator based
on the subset of subjects that excludes subject $i,$ and define
\[
\hat {\tilde R}_{ni} = I( \bm{S}_i^\top  \hat{\bm{\beta}}^{(-i)}_{2n}  > 0) - I(\bm{S}_i^\top  \bm{\beta}^*_{2}  > 0), i \in \bm{I}_n
\]
and also
\[
\tilde R_{ni} = I\{ 0 \leq | \bm{S}_i^\top  \bm{\beta}^*_{2} | \leq | \bm{S}_i^\top  ( \hat{\bm{\beta}}^{(-i)}_{2n}-\bm{\beta}^*_{2}) | \};
\]
similarly to before, $| \hat {\tilde R}_{ni} | \leq \tilde R_{ni}.$

Then, considering \eqref{G2} with  $\hat R_{ni}$ replaced by $\hat {\tilde R}_{ni}$, we may write
\[
E\left[ \mbox{\eqref{G2}} \right]
 \, = \, \frac{1}{\sqrt{n}} \sum_{i \in \bm{I}_n}  E\left(  \bm{W}_i  \bm{S}_i^\top  \bm{\beta}^*_{2}  \hat {\tilde R}_{ni}  
\, E\left[ A_{1i} - {\mu}_{1Ai}  \mid \bm{W}_i , \bm{S}_i,   \bm{I}_n, \hat{\bm{\beta}}^{(-i)}_{2n} \right] \mid \bm{I}_n \right).
 \]
In view of the definition of $\hat{\bm{\beta}}^{(-i)}_{2n},$ we have
\[
E\left[ A_{1i} - {\mu}_{1Ai}  \mid \bm{W}_i , \bm{S}_i,   \bm{I}_n, \hat{\bm{\beta}}^{(-i)}_{2n} \right]
= E\left[ A_{1i} - {\mu}_{1Ai}  \mid \bm{W}_i , \bm{S}_i,   \bm{I}_n \right] = 0,
\]
the last equality following by assumption. Therefore, it follows that 
$E\left[ \mbox{\eqref{G2}} \mid \bm{I}_n \right]= 0.$
Arguing similarly and using the conditional variance formula,
\[
var\left[ \mbox{\eqref{G2}} \mid \bm{I}_n \right] = \frac{1}{n} 
\sum_{i \in \bm{I}_n} 
 E \left(  \bm{W}^{\otimes 2}_i  \left( \bm{S}_i^\top  \bm{\beta}^*_{2} \right)^2 \left[ A_{1i} - {\mu}_{1Ai}  \right]^2 \hat {\tilde R}^2_{ni}   \mid \bm{I}_n \right).
\]
Let $V_{kj}$ denote the $(k,j)$ element of the matrix on the right-hand side of this last expression. Then, 
under Assumptions \ref{assump:support-X} -- \ref{assump:posdef}, it can be shown that 
\[
V_{kj} \leq 
\frac{C}{n} 
\sum_{i \in \bm{I}_n} 
 E \left( ( \bm{S}_i^\top  \bm{\beta}^*_{2} )^2 \, {\tilde R}_{ni}   \mid \bm{I}_n \right)
\]
for some finite constant $C > 0.$ Similarly to the proof of Theorem 1, the fact that
\[
| \bm{S}_i^\top  \bm{\beta}^*_{2} | \tilde R_{ni} \leq | \bm{S}_i^\top  ( \hat{\bm{\beta}}^{(-i)}_{2n}-\bm{\beta}^*_{2}) | {\tilde R}_{ni}
\]
now implies the inequality
\[
V_{kj} \leq 
\frac{C}{n} 
\sum_{i \in \bm{I}_n} 
 E \left(  \left[ \bm{S}_i^\top  ( \hat{\bm{\beta}}^{(-i)}_{2n}-\bm{\beta}^*_{2}) \right]^2  \, {\tilde R}_{ni}   \mid \bm{I}_n \right),
\]
from which it follows that
\begin{eqnarray}
\nonumber
V_{kj}
 & \leq & 
 \frac{C}{n} 
 \sum_{i \in \bm{I}_n}   E \left(    (\hat{\bm{\beta}}^{(-i)}_{2n}-\bm{\beta}^*_{2})^\top 
  \bm{S}^{\otimes 2}_i   ( \hat{\bm{\beta}}^{(-i)}_{2n}-\bm{\beta}^*_{2})    {\tilde R}_{ni} \mid \bm{I}_n \right),\\
  \label{key bound}
 & \leq & 
 C \lambda^2_{1n} 
  E \left(    \| \hat{\bm{\beta}}^{(-i)}_{2n}-\bm{\beta}^*_{2} \|^2_2 \mid \bm{I}_n \right),
\end{eqnarray}
where  $\lambda^2_{1n} > 0$ is the maximum eigenvalue of $n^{-1} \sum_{i \in \bm{I}_n}   \bm{S}^{\otimes 2}_i.$
Because the right-hand side of \eqref{key bound} goes to zero as $ n \rightarrow \infty$ for all $(k,j),$ it follows that \eqref{G2} is $o_p(1).$ A similar argument
establishes that \eqref{G3} is $o_p(1).$ 

 \end{proof}
 

\section{Additional simulation studies}
\subsection{Various sample sizes} \label{sec:vsamplesize}

In this section, we complement the main simulation study by examining the performance of our proposed method under various sample sizes. 
We generate 500 datasets with sample sizes $N$ of 1000, 500, and 250 using the same generative model as in Section 5.1 in the paper,
and except in Section 2.2, use the same methods to estimate $\mu_{1Y}(\cdot)$, $\mu_{2Y}(\cdot)$, $\mu_{1A}(\cdot)$, and $\mu_{2A}(\cdot)$
as in the main simulation study.

Tables \ref{tab:1000beta2}-\ref{tab:250beta1} in this document respectively summarize the results for sample sizes $N$ of 1000, 500, and
250 sample sizes. Overall, our method outperforms both Q$_{N,N}$ and dWOLS$_{N,N}$ for all sample sizes, particularly when the underlying 
treatment assignment and the outcome models are nonlinear. More specifically,  when the postulated parametric models for the nuisance 
parameters are correctly specified, the bias of the dWOLS$_{N,N}$ estimators are comparable to the estimators obtained by our proposed method,
but the latter usually has a substantially smaller standard error. However, when the parametric models for the outcome and the treatment assignment models 
are both misspecified, the dWOLS$_{N,N}$ estimators exhibit large biases. 
When $N=250$, under the linear treatment assignment model and FGS$^R$ outcome model, the proposed method shows larger bias in estimating 
$\bm{\beta}_2$ than the dWOLS$_{N,N}$ (Table \ref{tab:250beta2}). We conjecture that this occurs because there are only a relatively small number of 
units that are rerandomized at stage 2 (i.e., 50\%), and that this subsequently affects the performance of Super Learner. 

\begin{table}[t]
\centering
\caption{Performance of the proposed Q-learning method in estimating the second stage parameters under different model complexities ($N=$1000). 
The true parameter values for the linear and FGS outcome models are $\beta^*_{2,1}=1$, $\beta^*_{2,2}=1$ and 
$\beta^*_{2,1} \approx 0$, $\beta^*_{2,2}  \approx -2,$ respectively.}
\resizebox{\textwidth}{!}
{\begin{tabular}{lcccccc|cccccc|cccc}
  \hline
           & \multicolumn{6}{c}{$\beta^*_{2,1}$} &\multicolumn{6}{c}{$\beta^*_{2,2}$} \\
& \multicolumn{2}{c}{Q$_{N,N}$} &   \multicolumn{2}{c}{Proposed} & \multicolumn{2}{c}{dWOLS$_{N,N}$} & \multicolumn{2}{c}{Q$_{N,N}$} & \multicolumn{2}{c}{Proposed} & \multicolumn{2}{c}{dWOLS$_{N,N}$}  \\
Models &  Bias & S.D.&  Bias & S.D. & Bias & S.D.  & Bias & S.D. & Bias & S.D. & Bias & S.D.  \\
  \hline
\multicolumn{13}{c}{\it Randomized Treatment Assignment Model} \\ \hline
Linear$^R$&0.007 & 0.055 & 0.041 & 0.112 & 0.007 & 0.113 & 0.003 & 0.056 & 0.008 & 0.112 & 0.003 & 0.114 \\
FGS$^R$& 0.002 & 0.581 & 0.015 & 0.787 & 0.010 & 1.136 & 0.055 & 0.370 & 0.055 & 0.424 & 0.049 & 0.667\\\hline
\multicolumn{13}{c}{\it Linear Treatment Assignment Model} \\ \hline
Linear$^R$&0.000 & 0.060 & 0.005 & 0.145 & 0.004 & 0.147 & 0.011 & 0.063 & 0.026 & 0.144 & 0.020 & 0.149  \\
FGS$^R$& 2.235 & 2.490 & 0.039 & 1.010 & 0.035 & 2.224 & 2.478 & 0.809 & 0.067 & 0.529 & 0.016 & 0.968\\ \hline
\multicolumn{13}{c}{\it Quadratic Treatment Assignment Model} \\ \hline
Linear$^R$&0.004 & 0.054 & 0.006 & 0.117 & 0.003 & 0.108 & 0.002 & 0.059 & 0.015 & 0.116 & 0.005 & 0.121\\
FGS$^R$&0.872 & 0.630 & 0.166 & 0.754 & 0.873 & 1.284 & 0.012 & 0.332 & 0.050 & 0.410 & 0.021 & 0.570 \\ \hline
\multicolumn{13}{c}{\it InterQuad Treatment Assignment Model} \\ \hline
Linear$^R$& 0.008 & 0.056 & 0.009 & 0.118 & 0.007 & 0.113 & 0.004 & 0.057 & 0.024 & 0.117 & 0.005 & 0.116 \\
FGS$^R$& 0.718 & 0.665 & 0.077 & 0.817 & 0.703 & 1.253 & 0.442 & 0.333 & 0.068 & 0.416 & 0.437 & 0.560\\
\hline
\end{tabular}}
\label{tab:1000beta2}
\end{table}


\begin{table}[t]
\centering
\caption{Performance of the proposed Q-learning method in estimating the first stage parameters under different model complexities ($N=$1000). 
The true parameter values are $\beta^*_{1,1}=\beta^*_{1,2}=0$.}
\resizebox{\textwidth}{!}{\begin{tabular}{lcccccc|cccccc|cccc}
  \hline
           & \multicolumn{6}{c}{$\beta^*_{1,1}$} &\multicolumn{6}{c}{$\beta^*_{1,2}$} \\
& \multicolumn{2}{c}{Q$_{N,N}$} &   \multicolumn{2}{c}{Proposed} & \multicolumn{2}{c}{dWOLS$_{N,N}$} & \multicolumn{2}{c}{Q$_{N,N}$} & \multicolumn{2}{c}{Proposed} & \multicolumn{2}{c}{dWOLS$_{N,N}$}  \\
Models &  Bias & S.D.&  Bias & S.D. & Bias & S.D.  & Bias & S.D. & Bias & S.D. & Bias & S.D.  \\
  \hline
\multicolumn{13}{c}{\it Randomized Treatment Assignment Model} \\ \hline
Linear$^R$&0.001 & 0.140 & 0.004 & 0.150 & 0.006 & 0.150 & 0.013 & 0.138 & 0.008 & 0.140 & 0.005 & 0.149  \\
FGS$^R$& 0.024 & 0.823 & 0.034 & 0.858 & 0.015 & 1.072 & 0.035 & 0.584 & 0.027 & 0.638 & 0.011 & 0.866\\\hline
\multicolumn{13}{c}{\it Linear Treatment Assignment Model} \\ \hline
Linear$^R$& 0.156 & 0.152 & 0.003 & 0.173 & 0.016 & 0.177 & 0.153 & 0.131 & 0.011 & 0.163 & 0.008 & 0.172   \\
FGS$^R$& 2.206 & 3.284 & 0.192 & 1.132 & 0.197 & 4.588 & 1.859 & 2.836 & 0.148 & 0.855 & 0.167 & 4.377   \\\hline
\multicolumn{13}{c}{\it Quadratic Treatment Assignment Model} \\  \hline
Linear$^R$& 2.395 & 0.138 & 0.023 & 0.171 & 0.008 & 0.173 & 0.679 & 0.123 & 0.013 & 0.162 & 0.004 & 0.164 \\
FGS$^R$& 7.259 & 0.889 & 0.043 & 0.963 & 0.286 & 1.213 & 1.656 & 0.827 & 0.059 & 0.735 & 0.041 & 1.201\\ \hline
\multicolumn{13}{c}{\it InterQuad Treatment Assignment Model} \\ \hline
Linear$^R$& 2.297 & 0.143 & 0.012 & 0.171 & 0.012 & 0.178 & 0.420 & 0.128 & 0.017 & 0.161 & 0.010 & 0.166\\
FGS$^R$& 7.580 & 0.829 & 0.098 & 0.975 & 0.279 & 1.203 & 2.348 & 0.705 & 0.057 & 0.715 & 0.275 & 0.908\\
\hline
\end{tabular}}
\label{tab:1000beta1}
\end{table}


\begin{table}[t]
\centering
\caption{Performance of the proposed Q-learning method in estimating the second stage parameters under different model complexities ($N=$500). 
The true parameter values for the linear and FGS outcome models are $\beta^*_{2,1}=1$, $\beta^*_{2,2}=1$ and 
$\beta^*_{2,1} \approx 0$, $\beta^*_{2,2}  \approx -2,$ respectively.}
\resizebox{\textwidth}{!}
{\begin{tabular}{lcccccc|cccccc|cccc}
  \hline
           & \multicolumn{6}{c}{$\beta^*_{2,1}$} &\multicolumn{6}{c}{$\beta^*_{2,2}$} \\
& \multicolumn{2}{c}{Q$_{N,N}$} &   \multicolumn{2}{c}{Proposed} & \multicolumn{2}{c}{dWOLS$_{N,N}$} & \multicolumn{2}{c}{Q$_{N,N}$} & \multicolumn{2}{c}{Proposed} & \multicolumn{2}{c}{dWOLS$_{N,N}$}  \\
Models &  Bias & S.D.&  Bias & S.D. & Bias & S.D.  & Bias & S.D. & Bias & S.D. & Bias & S.D.  \\
  \hline
\multicolumn{13}{c}{\it Randomized Treatment Assignment Model} \\ \hline
Linear$^R$&0.004 & 0.083 & 0.013 & 0.168 & 0.004 & 0.169 & 0.009 & 0.079 & 0.023 & 0.166 & 0.010 & 0.161 \\
FGS$^R$&0.111 & 0.972 & 0.113 & 1.143 & 0.097 & 1.906 & 0.029 & 0.464 & 0.015 & 0.582 & 0.006 & 0.815\\\hline
\multicolumn{13}{c}{\it Linear Treatment Assignment Model} \\ \hline
Linear$^R$&0.014 & 0.087 & 0.021 & 0.207 & 0.007 & 0.214 & 0.017 & 0.087 & 0.006 & 0.205 & 0.019 & 0.211  \\
FGS$^R$& 2.409 & 0.909 & 0.117 & 1.311 & 0.136 & 2.250 & 2.448 & 0.575 & 0.185 & 0.761 & 0.061 & 2.111\\ \hline
\multicolumn{13}{c}{\it Quadratic Treatment Assignment Model} \\ \hline
Linear$^R$&0.003 & 0.083 & 0.027 & 0.171 & 0.007 & 0.169 & 0.004 & 0.079 & 0.007 & 0.169 & 0.004 & 0.163\\
FGS$^R$&0.738 & 0.805 & 0.005 & 1.101 & 0.695 & 1.533 & 0.008 & 0.523 & 0.049 & 0.658 & 0.002 & 0.932 \\ \hline
\multicolumn{13}{c}{\it InterQuad Treatment Assignment Model} \\ \hline
Linear$^R$& 0.005 & 0.086 & 0.014 & 0.171 & 0.010 & 0.172 & 0.003 & 0.083 & 0.007 & 0.169 & 0.005 & 0.170  \\
FGS$^R$& 0.653 & 1.085 & 0.008 & 1.229 & 0.627 & 2.162 & 0.468 & 0.551 & 0.046 & 0.668 & 0.436 & 0.977\\
\hline
\end{tabular}}
\label{tab:500beta2}
\end{table}


\begin{table}[t]
\centering
\caption{Performance of the proposed Q-learning method in estimating the first stage parameters under different model complexities ($N=$500). 
The true parameter values are $\beta^*_{1,1}=\beta^*_{1,2}=0$.}
\resizebox{\textwidth}{!}{\begin{tabular}{lcccccc|cccccc|cccc}
  \hline
           & \multicolumn{6}{c}{$\beta^*_{1,1}$} &\multicolumn{6}{c}{$\beta^*_{1,2}$} \\
& \multicolumn{2}{c}{Q$_{N,N}$} &   \multicolumn{2}{c}{Proposed} & \multicolumn{2}{c}{dWOLS$_{N,N}$} & \multicolumn{2}{c}{Q$_{N,N}$} & \multicolumn{2}{c}{Proposed} & \multicolumn{2}{c}{dWOLS$_{N,N}$}  \\
Models &  Bias & S.D.&  Bias & S.D. & Bias & S.D.  & Bias & S.D. & Bias & S.D. & Bias & S.D.  \\
  \hline
\multicolumn{13}{c}{\it Randomized Treatment Assignment Model} \\ \hline
Linear$^R$&0.007 & 0.208 & 0.017 & 0.217 & 0.020 & 0.205 & 0.007 & 0.189 & 0.008 & 0.204 & 0.011 & 0.206  \\
FGS$^R$& 0.107 & 1.286 & 0.061 & 1.181 & 0.036 & 1.556 & 0.006 & 0.871 & 0.024 & 0.870 & 0.004 & 1.078\\\hline
\multicolumn{13}{c}{\it Linear Treatment Assignment Model} \\ \hline
Linear$^R$& 0.147 & 0.227 & 0.015 & 0.241 & 0.016 & 0.238 & 0.167 & 0.183 & 0.025 & 0.228 & 0.009 & 0.225  \\
FGS$^R$& 1.882 & 1.131 & 0.038 & 1.242 & 0.003 & 1.475 & 1.615 & 0.801 & 0.016 & 0.925 & 0.035 & 1.117   \\\hline
\multicolumn{13}{c}{\it Quadratic Treatment Assignment Model} \\  \hline
Linear$^R$& 2.392 & 0.187 & 0.003 & 0.243 & 0.006 & 0.242 & 0.676 & 0.180 & 0.004 & 0.230 & 0.007 & 0.241 \\
FGS$^R$& 7.441 & 1.172 & 0.028 & 1.292 & 0.406 & 1.580 & 1.457 & 0.786 & 0.018 & 0.962 & 0.032 & 1.183 \\ \hline
\multicolumn{13}{c}{\it InterQuad Treatment Assignment Model} \\ \hline
Linear$^R$& 2.325 & 0.199 & 0.009 & 0.245 & 0.001 & 0.250 & 0.414 & 0.170 & 0.006 & 0.231 & 0.009 & 0.236\\
FGS$^R$& 7.412 & 1.249 & 0.199 & 1.394 & 0.176 & 1.998 & 2.368 & 0.938 & 0.015 & 1.026 & 0.136 & 1.353\\
\hline
\end{tabular}}
\label{tab:500beta1}
\end{table}

\begin{table}[t]
\centering
\caption{Performance of the proposed Q-learning method in estimating the second stage parameters under different model complexities ($N=$250). 
The true parameter values for the linear and FGS outcome models are $\beta^*_{2,1}=1$, $\beta^*_{2,2}=1$ and 
$\beta^*_{2,1} \approx 0$, $\beta^*_{2,2}  \approx -2,$ respectively.}
\resizebox{\textwidth}{!}
{\begin{tabular}{lcccccc|cccccc|cccc}
  \hline
           & \multicolumn{6}{c}{$\beta^*_{2,1}$} &\multicolumn{6}{c}{$\beta^*_{2,2}$} \\
& \multicolumn{2}{c}{Q$_{N,N}$} &   \multicolumn{2}{c}{Proposed} & \multicolumn{2}{c}{dWOLS$_{N,N}$} & \multicolumn{2}{c}{Q$_{N,N}$} & \multicolumn{2}{c}{Proposed} & \multicolumn{2}{c}{dWOLS$_{N,N}$}  \\
Models &  Bias & S.D.&  Bias & S.D. & Bias & S.D.  & Bias & S.D. & Bias & S.D. & Bias & S.D.  \\
  \hline
\multicolumn{13}{c}{\it Randomized Treatment Assignment Model} \\ \hline
Linear$^R$&0.035 & 0.120 & 0.038 & 0.261 & 0.031 & 0.250 & 0.008 & 0.116 & 0.012 & 0.252 & 0.008 & 0.244  \\
FGS$^R$&0.032 & 1.138 & 0.053 & 1.541 & 0.070 & 2.229 & 0.058 & 0.694 & 0.089 & 0.959 & 0.045 & 1.340\\\hline
\multicolumn{13}{c}{\it Linear Treatment Assignment Model} \\ \hline
Linear$^R$&0.008 & 0.132 & 0.027 & 0.302 & 0.007 & 0.336 & 0.009 & 0.127 & 0.004 & 0.291 & 0.007 & 0.315  \\
FGS$^R$& 2.315 & 1.237 & 0.488 & 1.755 & 0.286 & 2.616 & 2.519 & 0.717 & 0.682 & 1.111 & 0.274 & 1.545\\ \hline
\multicolumn{13}{c}{\it Quadratic Treatment Assignment Model} \\ \hline
Linear$^R$&0.014 & 0.120 & 0.023 & 0.265 & 0.020 & 0.257 & 0.034 & 0.127 & 0.032 & 0.257 & 0.028 & 0.268\\
FGS$^R$&0.738 & 2.374 & 0.061 & 1.937 & 0.522 & 5.242 & 0.051 & 0.736 & 0.004 & 1.167 & 0.078 & 1.400 \\ \hline
\multicolumn{13}{c}{\it InterQuad Treatment Assignment Model} \\ \hline
Linear$^R$&0.010 & 0.127 & 0.042 & 0.256 & 0.007 & 0.262 & 0.009 & 0.120 & 0.018 & 0.250 & 0.012 & 0.247    \\
FGS$^R$& 0.794 & 1.066 & 0.265 & 1.655 & 0.804 & 2.139 & 0.428 & 0.691 & 0.000 & 1.093 & 0.457 & 1.244 \\
\hline
\end{tabular}}
\label{tab:250beta2}
\end{table}


\begin{table}[t]
\centering
\caption{Performance of the proposed Q-learning method in estimating the first stage parameters under different model complexities ($N=$250). 
The true parameter values are $\beta^*_{1,1}=\beta^*_{1,2}=0$.}
\resizebox{\textwidth}{!}{\begin{tabular}{lcccccc|cccccc|cccc}
  \hline
           & \multicolumn{6}{c}{$\beta^*_{1,1}$} &\multicolumn{6}{c}{$\beta^*_{1,2}$} \\
& \multicolumn{2}{c}{Q$_{N,N}$} &   \multicolumn{2}{c}{Proposed} & \multicolumn{2}{c}{dWOLS$_{N,N}$} & \multicolumn{2}{c}{Q$_{N,N}$} & \multicolumn{2}{c}{Proposed} & \multicolumn{2}{c}{dWOLS$_{N,N}$}  \\
Models &  Bias & S.D.&  Bias & S.D. & Bias & S.D.  & Bias & S.D. & Bias & S.D. & Bias & S.D.  \\
  \hline
\multicolumn{13}{c}{\it Randomized Treatment Assignment Model} \\ \hline
Linear$^R$&0.020 & 0.284 & 0.005 & 0.308 & 0.008 & 0.303 & 0.031 & 0.244 & 0.008 & 0.291 & 0.014 & 0.299  \\
FGS$^R$& 0.166 & 1.389 & 0.019 & 1.661 & 0.092 & 2.477 & 0.159 & 1.200 & 0.019 & 1.274 & 0.041 & 1.698 \\\hline
\multicolumn{13}{c}{\it Linear Treatment Assignment Model} \\ \hline
Linear$^R$& 0.194 & 0.284 & 0.050 & 0.338 & 0.012 & 0.330 & 0.174 & 0.262 & 0.016 & 0.320 & 0.002 & 0.349  \\
FGS$^R$& 1.960 & 1.575 & 0.153 & 1.704 & 0.089 & 1.919 & 1.641 & 1.261 & 0.055 & 1.286 & 0.126 & 1.540  \\\hline
\multicolumn{13}{c}{\it Quadratic Treatment Assignment Model} \\  \hline
Linear$^R$& 2.413 & 0.289 & 0.016 & 0.351 & 0.001 & 0.339 & 0.664 & 0.257 & 0.016 & 0.336 & 0.026 & 0.367 \\
FGS$^R$& 7.840 & 2.066 & 0.106 & 1.855 & 0.398 & 2.460 & 1.373 & 1.273 & 0.093 & 1.457 & 0.096 & 2.040 \\ \hline
\multicolumn{13}{c}{\it InterQuad Treatment Assignment Model} \\ \hline
Linear$^R$& 2.329 & 0.294 & 0.052 & 0.349 & 0.007 & 0.358 & 0.451 & 0.251 & 0.006 & 0.334 & 0.024 & 0.369\\
FGS$^R$&7.423 & 1.469 & 0.212 & 1.837 & 0.107 & 2.335 & 2.288 & 1.188 & 0.079 & 1.461 & 0.238 & 1.721\\
\hline
\end{tabular}}
\label{tab:250beta1}
\end{table}

\FloatBarrier

\subsection{Performance under alternative nonparametric estimation methods}

In this section, we assess the performance of the proposed method when the nuisance parameters are instead estimated using either random forests or generalized additive models, recalling that both were included as part of the library used by Super Learner. As in the main simulation study, we generated 500 datasets of size $N=$2000 using the same generative model as in Section \ref{sec:vsamplesize}. In Table \ref{tab:MLbeta2} and \ref{tab:MLbeta1}, columns RF-RF and GAM-GAM respectively represent cases where a random forest and a generalized additive model are used for both the treatment assignment (i.e., $\mu_{1A}(\cdot)$ and $\mu_{2A}(\cdot)$) and the outcome models (i.e., $\mu_{1Y}(\cdot)$ and $\mu_{2Y}(\cdot)$). The column RF-GAM represents a case where a random forest is used for the outcome model and a generalized additive model is used for the treatment assignment model.  Overall, as the complexity of the outcome or treatment assignment model increases, modeling the nuisance parameters using a random forest results in estimators that are less biased compared with those obtained by the generalized additive model. However, comparing these results with those presented in Tables 1 and 2, we observe that Super Learner typically performs better than these other data adaptive methods.

\begin{table}[h]
\centering
\caption{Performance of the proposed Q-learning method in estimating the second stage parameters 
using machine learning methods  under different model complexities ($N=$2000). 
The true parameter values for the linear and FGS outcome models are $\beta^*_{2,1}=1$, $\beta^*_{2,2}=1$ and 
$\beta^*_{2,1} \approx 0$, $\beta^*_{2,2}  \approx -2,$ respectively.}
\resizebox{\textwidth}{!}
{\begin{tabular}{lcccccc|cccccc|cccc}
  \hline
           & \multicolumn{6}{c}{$\beta^*_{2,1}$} &\multicolumn{6}{c}{$\beta^*_{2,2}$} \\
& \multicolumn{2}{c}{RF-RF} &   \multicolumn{2}{c}{GAM-GAM} & \multicolumn{2}{c}{RF-GAM} & \multicolumn{2}{c}{RF-RF} &   \multicolumn{2}{c}{GAM-GAM} & \multicolumn{2}{c}{RF-GAM}  \\
Models &  Bias & S.D.&  Bias & S.D. & Bias & S.D.  & Bias & S.D. & Bias & S.D. & Bias & S.D.  \\
  \hline
\multicolumn{13}{c}{\it Randomized Treatment Assignment Model} \\ \hline
Linear$^R$&0.021 & 0.088 & 0.051 & 0.078 & 0.001 & 0.086 & 0.021 & 0.084 & 0.048 & 0.078 & 0.021 & 0.084  \\
FGS$^R$& 0.096 & 0.521 & 0.236 & 0.684 & 0.049 & 0.499 & 0.047 & 0.282 & 0.091 & 0.436 & 0.091 & 0.275\\\hline
\multicolumn{13}{c}{\it Linear Treatment Assignment Model} \\ \hline
Linear$^R$&0.005 & 0.107 & 0.009 & 0.103 & 0.006 & 0.110 & 0.002 & 0.103 & 0.012 & 0.103 & 0.011 & 0.107  \\
FGS$^R$& 0.447 & 0.687 & 0.023 & 0.823 & 0.037 & 0.712 & 0.522 & 0.405 & 0.074 & 0.553 & 0.064 & 0.419\\ \hline
\multicolumn{13}{c}{\it Quadratic Treatment Assignment Model} \\ \hline
Linear$^R$&0.013 & 0.091 & 0.033 & 0.080 & 0.065 & 0.085 & 0.014 & 0.088 & 0.042 & 0.080 & 0.002 & 0.083\\
FGS$^R$&0.194 & 0.628 & 0.764 & 0.737 & 0.301 & 0.616 & 0.039 & 0.337 & 0.091 & 0.501 & 0.022 & 0.352 \\ \hline
\multicolumn{13}{c}{\it InterQuad Treatment Assignment Model} \\ \hline
Linear$^R$&0.035 & 0.091 & 0.076 & 0.081 & 0.044 & 0.085 & 0.038 & 0.087 & 0.076 & 0.081 & 0.010 & 0.083     \\
FGS$^R$& 0.192 & 0.651 & 0.654 & 0.761 & 0.219 & 0.648 & 0.071 & 0.329 & 0.304 & 0.499 & 0.065 & 0.344 \\
\hline
\end{tabular}}
\label{tab:MLbeta2}
\end{table}

\begin{table}[t]
\centering
\caption{Performance of the proposed Q-learning method in estimating the first stage parameters using machine learning methods under different model complexities ($N=$2000). The true parameter values are $\beta^*_{1,1}=\beta^*_{1,2}=0$.}
\resizebox{\textwidth}{!}{\begin{tabular}{lcccccc|cccccc|cccc}
  \hline
           & \multicolumn{6}{c}{$\beta^*_{1,1}$} &\multicolumn{6}{c}{$\beta^*_{1,2}$} \\
& \multicolumn{2}{c}{RF-RF} &   \multicolumn{2}{c}{GAM-GAM} & \multicolumn{2}{c}{RF-GAM} & \multicolumn{2}{c}{RF-RF} &   \multicolumn{2}{c}{GAM-GAM} & \multicolumn{2}{c}{RF-GAM}  \\
Models &  Bias & S.D.&  Bias & S.D. & Bias & S.D.  & Bias & S.D. & Bias & S.D. & Bias & S.D.  \\
  \hline
\multicolumn{13}{c}{\it Randomized Treatment Assignment Model} \\ \hline
Linear$^R$&0.016 & 0.104 & 0.017 & 0.108 & 0.003 & 0.107 & 0.029 & 0.099 & 0.004 & 0.100 & 0.003 & 0.102  \\
FGS$^R$& 0.029 & 0.537 & 0.200 & 0.725 & 0.014 & 0.612 & 0.020 & 0.413 & 0.063 & 0.514 & 0.012 & 0.444 \\\hline
\multicolumn{13}{c}{\it Linear Treatment Assignment Model} \\ \hline
Linear$^R$& 0.008 & 0.116 & 0.000 & 0.127 & 0.003 & 0.122 & 0.008 & 0.110 & 0.005 & 0.119 & 0.008 & 0.116  \\
FGS$^R$& 0.020 & 0.677 & 0.039 & 0.775 & 0.081 & 0.700 & 0.086 & 0.491 & 0.021 & 0.575 & 0.009 & 0.508 \\\hline
\multicolumn{13}{c}{\it Quadratic Treatment Assignment Model} \\  \hline
Linear$^R$& 0.033 & 0.118 & 0.002 & 0.121 & 0.027 & 0.122 & 0.008 & 0.113 & 0.002 & 0.113 & 0.006 & 0.116  \\
FGS$^R$& 0.035 & 0.691 & 0.296 & 0.794 & 0.101 & 0.707 & 0.007 & 0.508 & 0.035 & 0.600 & 0.037 & 0.519 \\ \hline
\multicolumn{13}{c}{\it InterQuad Treatment Assignment Model} \\ \hline
Linear$^R$& 0.020 & 0.119 & 0.015 & 0.122 & 0.045 & 0.122 & 0.009 & 0.114 & 0.003 & 0.114 & 0.031 & 0.116\\
FGS$^R$&0.083 & 0.681 & 0.344 & 0.786 & 0.145 & 0.704 & 0.031 & 0.509 & 0.164 & 0.610 & 0.032 & 0.527\\
\hline
\end{tabular}}
\label{tab:MLbeta1}
\end{table}


\FloatBarrier

\subsection{Additional simulation results for the non-regular setting}

In Section 5.2 of the main paper, we showed that the proposed method can provide valid inferences in non-regular settings 
when the assumptions of Corollary 1 are satisfied. Table \ref{tab:nrsup} supplements the results of Table 3 for the main simulation study 
by providing additional comparisons with Q$_{N,N}$ and dWOLS$_{N,N}$. 

\begin{table}[t]
\centering
\caption{Performance of the proposed Q-learning method under different degrees of non-regularity.}
\resizebox{\textwidth}{!} {\begin{tabular}{lccc|ccc}
  \hline
           & \multicolumn{3}{c}{$\beta^*_{1,1}$} &\multicolumn{3}{c}{$\beta^*_{1,2}$} \\
Models & Q$_{N,N}$ & Proposed & dWOLS$_{N,N}$ & Q$_{N,N}$ & Proposed & dWOLS$_{N,N}$ \\
  \hline
\multicolumn{7}{c}{\it Randomized Treatment Assignment Model} \\ \hline
Linear$^{NR,0}$&  0.939(0.26) & 0.956(0.42) & 0.940(0.41) & 0.932(0.26)$^\dag$  & 0.959(0.40) & 0.933(0.41)$^\dag$\\ 
 Non-linear$^{NR,0}$&  0.949(0.95) & 0.965(1.40)$^\dag$ & 0.933(1.40)$^\dag$ & 0.949(0.44) & 0.952(0.56) & 0.946(0.55)\\ 
Linear$^{NR,1}$&  0.937(0.45)  & 0.962(0.46) & 0.938(0.70) & 0.936(0.45) & 0.966(0.45)$^\dag$ & 0.933(0.70)$^\dag$  \\ 
 Non-linear$^{NR,1}$& 0.953(0.99)  & 0.966(1.41) & 0.924(1.46)$^\dag$ & 0.936(0.51)$^\dag$ & 0.946(0.60) & 0.929(0.69)$^\dag$\\ \hline
\multicolumn{7}{c}{\it Linear Treatment Assignment Model} \\ \hline
Linear$^{NR,0}$& 0.913(0.26)$^\dag$ & 0.965(0.45)$^\dag$ & 0.940(0.45) &  0.898(0.26)$^\dag$ & 0.948(0.44) & 0.940(0.45)\\ 
 Non-linear$^{NR,0}$&  0.396(0.96)$^\dag$ & 0.949(1.58) & 0.921(1.56)$^\dag$ & 0.102(0.48)$^\dag$ & 0.955(0.65) & 0.936(0.65) \\ 
Linear$^{NR,1}$&  0.856(0.45)$^\dag$ & 0.953(0.51) & 0.939(0.76) & 0.845(0.45)$^\dag$ & 0.955(0.50) & 0.924(0.76)$^\dag$ \\ 
 Non-linear$^{NR,1}$&  0.343(1.10)$^\dag$ & 0.952(1.61) & 0.923(1.63)$^\dag$ & 0.101(0.55)$^\dag$ & 0.948(0.70) & 0.942(0.79) \\  \hline
\multicolumn{7}{c}{\it InterQuad Treatment Assignment Model} \\ \hline
Linear$^{NR,0}$&  0.949(0.29) & 0.966(0.46)$^\dag$ & 0.956(0.45) & 0.939(0.29) & 0.957(0.45) & 0.934(0.45)$^\dag$ \\
 Non-linear$^{NR,0}$&  0.858(1.10)$^\dag$ & 0.959(1.45) & 0.786(1.45)$^\dag$ & 0.897(0.50)$^\dag$ & 0.962(0.63) & 0.789(0.62)$^\dag$\\ 
Linear$^{NR,1}$&  0.951(0.49) & 0.964(0.52) & 0.948(0.77) & 0.937(0.50)$^\dag$ & 0.950(0.51) & 0.939(0.77) \\ 
 Non-linear$^{NR,1}$& 0.861(1.05)$^\dag$ & 0.964(1.47) & 0.804(1.53)$^\dag$ & 0.910(0.58)$^\dag$ & 0.967(0.68) & 0.850(0.78)$^\dag$ \\ 
\hline
\end{tabular}}
{\small Numbers in parentheses correspond to average confidence interval length.}
\label{tab:nrsup}
\end{table}

Below, we also consider a simulation that corresponds to a non-regular setting when the assumptions of Corollary 1 and
Assumption 6 are both violated. We consider the same general set-up as in Section 5.2, but modify the outcome models as described below:
\begin{itemize}
\item Linear$^{NR,\varpi}$: $Y= \bX_1^\top \balpha_{1}+ \breve{\bX}_2^\top \balpha_{2} +A_1\bX_1^\top \btheta_{1}+ A_2 R (\theta_{21} A_1+ \theta_{22} \breve{X}_{21} )+\epsilon$ where $\balpha_{1}=(1,0.1,0.1,0.1,0.1)^\top$, $\balpha_{2}=(1,0.1,0.1,0.1)^\top,$ $\btheta_{1}=(0,0,0,0,0)^\top,$ and $\btheta_2 = 2 (\varpi,\varpi)^\top;$  
\item Non-linear$^{NR,\varpi}$: $Y= f(\bX_1)+A_1\bX_1^\top \btheta_{1}+   A_2 R (\theta_{21} A_1+ \theta_{22} \breve{X}_{21})+\epsilon$ where $\btheta_{1}=(0,0,0,0,0)^\top,$
$\btheta_2 = 2 (\varpi,\varpi)^\top,$  and
for $\bm{x} = (x_1,x_2,x_3,x_4,x_5)^\top,$ we set
\[
f(\bm{x}) = \sin(\pi x_1 x_2) + 2
  (x_3-0.5)^2 +  x_4 + 1.5 \frac{x_1}{|x_2|+|x_3|} +2 x_1 (x_2+x_3).
\]
\end{itemize}
Here,
$\breve{\bm{X}}_2=(\breve X_{21},\breve X_{22},X_{23},X_{24})^\top$ where $\breve X_{21}$ is generated from a Bernoulli distribution with 
success probability $[ 1+\exp \{-(2X_{11}-2A_1-1)\} ]^{-1},$  
$X_{22} = U_{11}$, $X_{23} = U_{12},$
$X_{24} = 0.35 X_{15} + U_{13},$
and $X_{25} = U_{14},$
where $U_{1l}, l =1,\ldots,3$ are independent and uniformly
distributed on [0,1]; the noise variable $\epsilon$ is again generated from $ \mathrm{N}(0,\sigma=0.5)$. 


In the above generative models, $\bm{S}^0 = (\bm{X}^\top_1,A_1,\breve{\bm{X}}^\top_2)^\top$ and $\bm{W}^0 = \bm{X}_1;$
the vectors $\bm{S}=R (1,A_1, \breve{X}_{21}, \breve{X}_{22})^\top$ and $\bm{W} = (1, X_{11}, X_{12})^\top$ are respectively used 
to model the second and first stage Q-functions. 
The second stage blip functions are linear and correctly specified in both models, leading
to $\bm{\beta}_2^* = (0,\btheta^\top_2,0)^\top.$
In both scenarios, for each subject $i$, we define the first-stage pseudo outcome as 
\begin{equation*}
\tilde {Y}_i
=Y_i + \varpi  R \bm{S}_i^\top \bm{\beta}^*_2 \{I(\bm{S}_i^\top \bm{\beta}^*_{2}  >0) -A_{2i}\},
\end{equation*}
its estimate $\hat{\tilde {Y}}_i$ being calculated by substituting $\hat{\bm{\beta}}_{2n}$ in for $\bm{\beta}^*_{2}.$
The construction of the pseudo-outcome, specifically the projection $\bm{S}^\top \bm{\beta}^*_2,$ results in 
the violation of Assumption \ref{assump:unique} for both generative models.
In particular,  $\varpi=0$ corresponds to no second-stage effect modifier because
$\bm{\beta}^*_2 = \bm{0};$ hence, we have $P( |\bm{S}^\top \bbeta^*_2| = 0) = 1.$
Setting $\varpi=1$ instead implies that there is no second-stage treatment effect 
when $R (A_1+ \breve X_{21}) = 0$ and a reasonably strong effect
when $R (A_1+ \breve X_{21}) = 1;$ in this case, we have $\| \bm{\beta}^*_2 \| > 0$
and $0 < P( |\bm{S}^\top \bbeta^*_2| = 0) < 1$.
In addition to violating Assumption 6, both generative models also violate the key assumption of Corollary 1 
(i.e., $E\left( A_{1i} - {\mu}_{1Ai}  \mid \bm{W}_{i}, \bm{S}_{i},  \bm{I}_n \right) = 0, i \in \bm{I}_n)$ 
because the Stage 2 decision rule depends on  $A_1$.


We compare the coverages of our estimators with those obtained for dWOLS using the $m$-out-of-$n$ bootstrap (i.e., dWOLS$_{m,N}^{\kappa=0.05}$).  
Table \ref{tab:nr2} presents the results based on 1000 datasets of size 2000. In the table, $\beta^*_{1,1}$, $\beta^*_{1,2}$ and $\beta^*_{1,3}$ 
respectively the true parameter values corresponding to variables $A_1$ (i.e., we assume $\bm{W}$ has an intercept), $A_1 X_{11}$ and $A_1 X_{12}$.  
When $\varpi=1$, $(\beta^*_{1,1},\beta^*_{1,2}, \beta^*_{1,3}) \approx (0.55,0.00,-0.57)$ and  $(\beta^*_{1,1},\beta^*_{1,2}, \beta^*_{1,3}) \approx (0.77,0.00,-0.29)$ for the linear and non-linear outcome models, respectively (note: these are approximated by simulation). Also, when $\varpi=0$, $(\beta^*_{1,1},\beta^*_{1,2}, \beta^*_{1,3}) \approx (-0.23,0.00,-0.29)$ and  $(\beta^*_{1,1},\beta^*_{1,2}, \beta^*_{1,3}) = (0,0,0)$ for the linear and non-linear outcome models, respectively. Despite violation
of the regularity assumption, the proposed method continues to provide valid confidence intervals for the parameters $\beta_{1,2}$, and $ \beta_{1,3}$. However, the confidence intervals for $\beta_{1,1}$ exhibit less-than-nominal coverage in a majority of cases. As expected, when either of the outcome or the treatment assignment models are correctly specified the dWOLS$_{m,N}^{\kappa=0.05}$ provides valid, if typically conservative, confidence intervals for all parameters. However, when both of these models are misspecified, the coverage rates are more substantially compromised.

\begin{table}[t]
\centering
\caption{Performance of proposed Q-learning method under different degrees of non-regularity when Corollary 1
and Assumption 6 both fail.}
\resizebox{\textwidth}{!} {\begin{tabular}{lcc|cc|cc}
  \hline
           & \multicolumn{2}{c}{$\beta^*_{1,1}$} &\multicolumn{2}{c}{$\beta^*_{1,2}$}&\multicolumn{2}{c}{$\beta^*_{1,3}$} \\
Models & Proposed & dWOLS$_{m,N}^{\kappa=0.05}$ &  Proposed & dWOLS$_{m,N}^{\kappa=0.05}$&  Proposed & dWOLS$_{m,N}^{\kappa=0.05}$ \\
  \hline
\multicolumn{7}{c}{\it Randomized Treatment Assignment Model} \\ \hline
Linear$^{NR,0}$& 0.938(0.12) & 0.981(0.14)$^\dag$ & 0.949(0.40) & 0.968(0.46)$^\dag$ & 0.965(0.40)$^\dag$ & 0.984(0.46)$^\dag$ \\ 
 Non-linear$^{NR,0}$&  0.940(0.29) & 0.967(0.43)$^\dag$ & 0.965(1.29)$^\dag$ & 0.977(1.69)$^\dag$ & 0.949(0.39) & 0.973(0.68)$^\dag$\\ 
Linear$^{NR,1}$&  0.951(0.22) & 0.976(0.26)$^\dag$ & 0.948(0.77) & 0.973(0.89)$^\dag$ & 0.950(0.77) & 0.983(0.89)$^\dag$\\ 
 Non-linear$^{NR,1}$& 0.928(0.34)$^\dag$ & 0.957(0.48) & 0.967(1.43)$^\dag$ & 0.980(1.82)$^\dag$ & 0.941(0.69) & 0.980(0.93)$^\dag$\\ \hline
\multicolumn{7}{c}{\it Linear Treatment Assignment Model} \\ \hline
Linear$^{NR,0}$& 0.908(0.12)$^\dag$ & 0.980(0.17)$^\dag$ & 0.954(0.43) & 0.968(0.50)$^\dag$ & 0.939(0.43) & 0.976(0.50)$^\dag$\\ 
 Non-linear$^{NR,0}$& 0.909(0.29)$^\dag$ & 0.984(0.47)$^\dag$ & 0.975(1.32)$^\dag$ & 0.972(1.77)$^\dag$ & 0.942(0.46) & 0.981(0.78)$^\dag$\\ 
Linear$^{NR,1}$&  0.951(0.24) & 0.972(0.28)$^\dag$ & 0.943(0.85) & 0.970(0.97)$^\dag$ & 0.949(0.86) & 0.969(0.98) \\ 
 Non-linear$^{NR,1}$&  0.955(0.35) & 0.962(0.47) & 0.971(1.50)$^\dag$ & 0.976(1.93)$^\dag$ & 0.942(0.80) & 0.978(1.06)$^\dag$ \\  \hline
\multicolumn{7}{c}{\it InterQuad Treatment Assignment Model} \\ \hline
Linear$^{NR,0}$& 0.934(0.12)$^\dag$ & 0.983(0.15)$^\dag$ & 0.952(0.41) & 0.975(0.48)$^\dag$ & 0.948(0.41) & 0.975(0.48)$^\dag$ \\
 Non-linear$^{NR,0}$&0.939(0.30) & 0.914(0.45)$^\dag$ & 0.944(1.29) & 0.884(1.75)$^\dag$ & 0.958(0.45) & 0.885(0.76)$^\dag$\\ 
Linear$^{NR,1}$&   0.950(0.26) & 0.972(0.31)$^\dag$ & 0.962(0.89) & 0.983(1.02)$^\dag$ & 0.941(0.90) & 0.963(1.03)\\ 
 Non-linear$^{NR,1}$& 0.907(0.38)$^\dag$ & 0.971(0.55)$^\dag$ & 0.968(1.51)$^\dag$ & 0.923(1.95)$^\dag$ & 0.962(0.87) & 0.945(1.12)\\ 
\hline
\end{tabular}}
{\small Numbers in parentheses correspond to average confidence interval length.}
\label{tab:nr2}
\end{table}

\FloatBarrier

\subsection{Value Function Estimates for Regular Case}

In this section we plot the value functions corresponding to the estimated decision rules obtained
in the regular setting (see Section 5.1 of the main paper). 
All methods perform essentially identically as far as the underlying value function
when the outcome model is linear, even for $N=250$; these results are not included. 
In the case where the outcome model is not linear, we see that the proposed
method typically results in a value function that is closest to the true optimal value, 
followed by dWOLS and then standard Q-learning. However, the overall 
degree of discrepancy is generally small, particularly between dWOLS and
the proposed approach.

 \begin{figure}[t] 
 \centering 
\includegraphics[width = 0.95\textwidth]{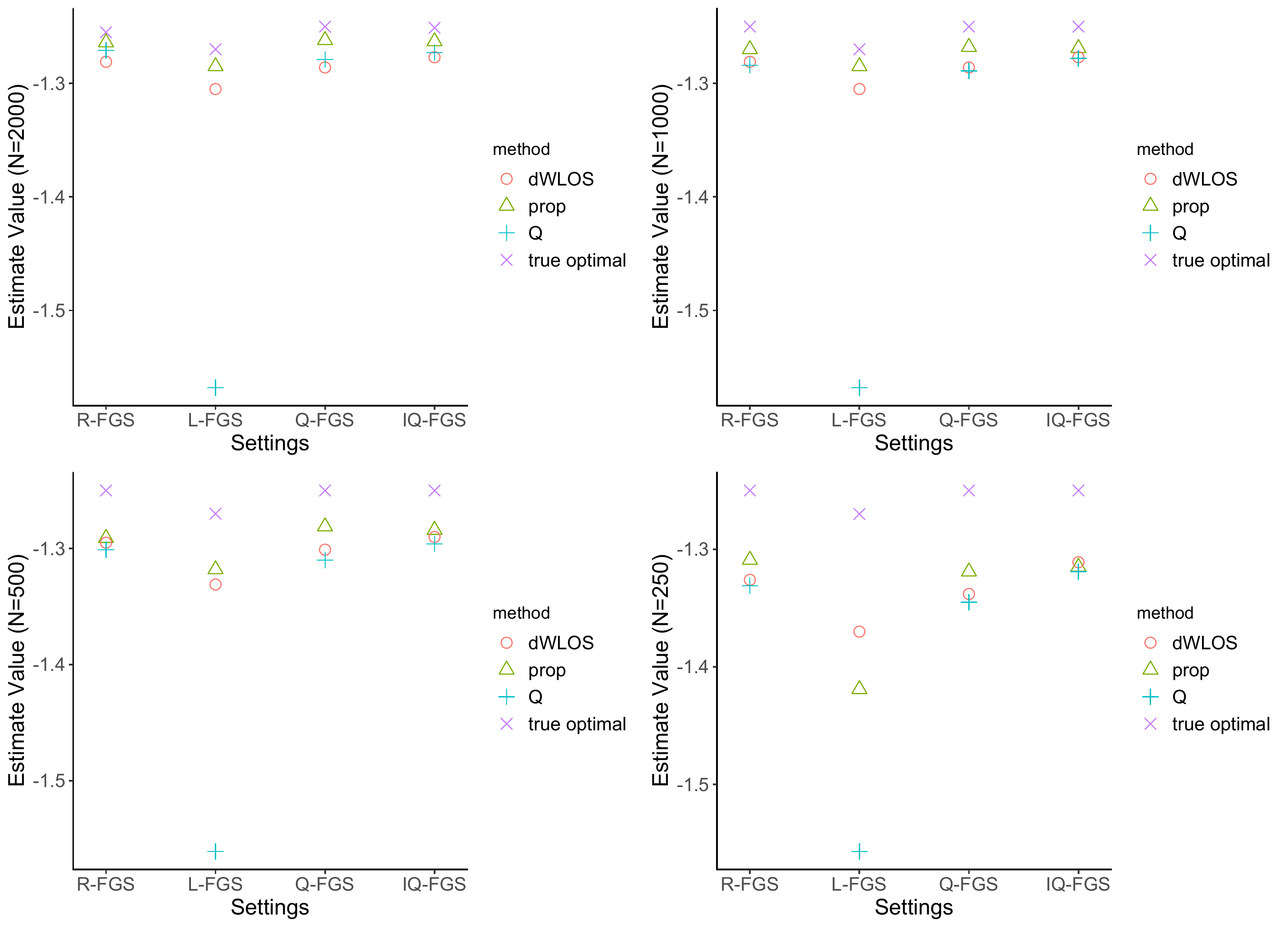}
\caption{\footnotesize Value functions for estimated decision rules for regular case and different sample sizes. } \label{fig:value}
\end{figure}

 \FloatBarrier

\end{document}